%% file: Public_LogicReg_8thEdit.tex
\documentclass[11pt]{article}
\usepackage{fullpage}

\usepackage[utf8]{inputenc} 
\usepackage[T1]{fontenc}    
\usepackage{hyperref}       
\usepackage{url}            
\usepackage{booktabs}       
\usepackage{amsfonts}       
\usepackage{nicefrac}       
\usepackage{microtype}      
\usepackage{tcolorbox}
\usepackage{diagbox}

\usepackage{enumerate}
\usepackage[shortlabels]{enumitem}
\usepackage{algorithmic}
\usepackage{algorithm}
\usepackage{subfigure}
\usepackage{graphicx} 
\usepackage{caption}
\usepackage{amsmath}
\usepackage{amsthm}
\usepackage{amssymb}
\usepackage{tikz}
\usepackage{tablefootnote}
\usepackage{multirow}
\usepackage{enumerate}
\usepackage{color}
\usepackage{xcolor}
\usepackage{natbib}

\usetikzlibrary{arrows}

\allowdisplaybreaks[4]

\usepackage{mathrsfs}
\usepackage{scalerel}

\usepackage{hyperref}
\usepackage{bm,todonotes}
\usepackage{upgreek}
\usepackage{adjustbox}

\allowdisplaybreaks

\newtheorem{thm}{Theorem}[section]
\newtheorem{lem}{Lemma}[section]
\newtheorem{cor}{Corollary}[section]

\newtheorem{defn}{Definition}[section]

\newtheorem{rem}{Remark}[section]

\newtheorem{prf}{Proof}[section]

\hypersetup{
	colorlinks=true,
	filecolor=blue,
	citecolor = blue,
	urlcolor=cyan,
}

\input{tex/math_commands}

\title{
Differentially Private Truncation of Unbounded Data via Public Second Moments
}

\author{
Zilong Cao\thanks{School of Mathematics, Northwest University; email: \texttt{nwu\_czl@stumail.nwu.edu.cn}. } \\
	\and
	Xuan Bi \thanks{Department of Information and Decision Sciences, Minnesota Carlson, University of Minnesota; email: \texttt{xbi@umn.edu.} 
}
	\and
	Hai Zhang \thanks{(\textbf{Corresponding author}) School of Mathematics, Northwest University; email: \texttt{zhanghai@nwu.edu.cn}.  
}}
\begin{document}

\maketitle

\begin{abstract}
Data privacy is important in the AI era, and differential privacy (DP) is one of the golden solutions. However, DP is typically applicable only if data have a bounded underlying distribution. We address this limitation by leveraging second-moment information from a small amount of public data. We propose Public-moment-guided Truncation (PMT), which transforms private data using the public second-moment matrix and applies a principled truncation whose radius depends only on non-private quantities: data dimension and sample size. This transformation yields a well-conditioned second-moment matrix, enabling its inversion with a significantly strengthened ability to resist the DP noise. Furthermore, we demonstrate the applicability of PMT by using penalized and generalized linear regressions. Specifically, we design new loss functions and algorithms, ensuring that solutions in the transformed space can be mapped back to the original domain. We have established improvements in the models' DP estimation through theoretical error bounds, robustness guarantees, and convergence results, attributing the gains to the conditioning effect of PMT. Experiments on synthetic and real datasets confirm that PMT substantially improves the accuracy and stability of DP models.
\end{abstract}

\section{Introduction}
Data privacy has become an increasingly critical challenge in today's data-driven world. Protecting large volumes of sensitive information is essential, especially when analytical results are derived from private data. Differential Privacy (DP; \cite{Dwork2006CalibratingNT}) provides a rigorous mathematical framework that ensures the output of an algorithm is statistically insensitive to any single data point, thus offering strong privacy guarantees. To enhance the utility and composability of standard DP, \cite{dong2022gaussian} introduced Gaussian Differential Privacy (GDP), which achieves the tightest known composition bounds for Gaussian mechanisms. GDP has since found broad applications across statistics \citep{awan2023canonical, avella2023differentially, zhao2025minimax}, machine learning \citep{bu2020deep, cao2023privacy, nasr2023tight}, and other areas, demonstrating its theoretical and practical effectiveness.

GDP relies on adding Gaussian noise to the output of a data-dependent algorithm, a process known as the Gaussian mechanism in differential privacy. However, when such mechanisms rely solely on private data, they often suffer from limited utility. More importantly, it requires data to be strictly bounded in order for GDP to be applicable \citep{Su2020DeepLG}. One common practice for handling unbounded data is to apply truncation. However, this inevitably compromises data usefulness. On the one hand, using a small truncation radius can substantially distort the original data distribution. On the other hand, using a large truncation radius necessitates injecting a large noise to satisfy the same level of differential privacy guarantee, which can also degrade utility. In either case, truncation imposes unavoidable alterations to the data and compromise subsequent analysis. Public data, as a high-quality and privacy-free resource, has been increasingly leveraged to improve the performance and utility of DP algorithms. More practically, many public datasets do not contain raw sensitive data, but rather privacy-preserving yet informative statistics, such as means or moments. For statistical estimation, \citet{ferrando2021combining} proposes a principled approach to combining public and private data. In the context of differentially private Gaussian distribution estimation, \citet{bie2022private} shows that even a small amount of public data can substantially enhance performance. Public data has also been used to boost the utility of DP gradient-based learning algorithms, as demonstrated in works such as \citet{kairouz2021nearly}, \citet{amid2022public}, and \citet{nasr2023effectively}. Beyond estimation and learning, public data has been applied in a variety of DP tasks, including private query release, synthetic data generation, and prediction \citep{ji2013differential, nandi2020privately, bassily2020learning, liu2021leveraging, bi2023distribution}. 

In the context of differentially private regression, the literature mainly focuses on the linear regression, penalized linear regression, and generalized linear models. The prior studies on DP ordinary least squares (OLS) focus on linear regression under the assumption of bounded private data \citep{sheffet2017differentially, wang2018revisiting, bernstein2019differentially}.  DP penalized regression with the $L_2$ regularization also is studied, where injecting noise into the sufficient statistics can produce a DP-compliant solution \citep{wang2018revisiting, bernstein2019differentially}. These methods typically treat the $L_2$ regularization as a tuning parameter to stabilize the inverse of the noisy second-moment matrix. Generalized linear models with DP guarantees is another direction in the literature, which can be solved via a convex loss, requiring iterative solving to the convex optimization problem. A large body of work explore DP gradient-based optimization methods under bounded-data assumptions \citep{abadi2016deep, ivkin2019communication, wang2019distributed, koloskova2023revisiting}. More recently, DP Newton methods, which leverage second-order information, have gained attention for their faster convergence and reduced privacy budget consumption \citep{ganesh2023faster, cao2023privacy, nasr2023tight}. However, these approaches exhibit sensitivity to the ill-conditioning of the Hessian matrix, frequently leading to numerical instability during iterations, and struggle with appropriately adjusting the regularization value within the DP framework. Selecting an appropriate regularization is challenging, especially when it depends on private data. Small regularization is unable to resist the DP noises effectively, but large regularization leads to over-regularization and substantial estimation bias.

In this paper, we provide a practical solution to achieve DP with unbounded data. The proposed method leverages a public second-moment matrix. We demonstrate its applicability through regression settings. We firstly propose Public-moment-guided Truncation (PMT), a transformation-truncation framework that uses the public second-moment matrix to map private data into an approximately isotropic space, followed by principled truncation with a radius determined solely by the data dimension and sample size (non-private quantities). This transformation produces a well-conditioned second-moment matrix, making its inverse more stable and less sensitive to DP noise. In order to show the applicability of the proposed method, we develop two algorithms: one for ridge regression (named as DP-PMTRR, an abbreviation for differentially private PMT ridge regression), which combines PMT with sufficient statistics perturbation (SSP) and a tailored loss function to yield a closed-form, robust estimator; and one for logistic regression (named DP-PMTLR, an abbreviation for differentially private PMT logistic regression), which integrates PMT into a DP Newton's method with a modified cross-entropy loss, improving convergence, numerical stability, and accuracy without manual regularization tuning. Our main contributions are summarized as follows:

\begin{itemize}
	\item \textbf{Public-moment-guided truncation method.} We introduce a transformation-truncation procedure that uses the public second-moment matrix to map private data into a space where the $l_2$-norm of each data point is smaller than $\sqrt{d(1 + \log n)}$ with high probability. Then, we take $\sqrt{d(1 + \log n)}$ as a principled radius to truncate the mapped private data without requiring extra private information. This alleviates the adverse impact of truncation and improves the performance of the downstream DP algorithm.

	\item \textbf{Applications in regression settings.} We discuss the general applicability of the proposed method in the penalized generalized linear model in DP scenarios, especially in models where the DP Hessian matrix is used to estimate model parameters. Specifically, we design new loss functions for ridge and logistic regression in the transformed scenarios, prove their estimation invariant, and show how optimal solutions in the transformed domain can be mapped back to the original parameter space.

	\item \textbf{Theoretical guarantees and utility gains.} We derive formal error bounds for DP ridge regression and prove the stable convergence of DP logistic regression based on our PMT method. Comparisons with private-data-only methods demonstrate substantial utility improvements by mitigating the effects of ill-conditioning, large inverse norms, and heavy regularization dependence.
	
	\item \textbf{Robust inverse analysis.} We show that PMT greatly improves the utility and robustness of the inverse DP second-moment matrix. The improvement originates from the transformed second-moment matrix holding a better condition number than the original, yielding a more stable inverse and reduced sensitivity to DP noise. We establish theoretical results, showing that our approach yields smaller inverse estimation error and requires a smaller private sample size. Specifically, our approach mitigates the impact of the true second-moment matrix on the inverse estimation error and only weakly depends on the regularization term. We also extend the stability analysis to weighted second-moment matrices.
\end{itemize}

\section{Preliminaries and Motivation}
\label{sec:preliminaries}
\subsection{Background of Privacy}\label{subsec:privacy}
\begin{defn}[Differential Privacy \cite{dwork2014algorithmic}]
	A randomized algorithm $\mathcal{M}:\mathcal{X}^n \to \mathcal{S}$ satisfies $(\epsilon,\delta)-$differential privacy ($(\epsilon,\delta)$-DP), if for any neighboring datasets $\mathbf{X},\mathbf{X}' \in \mathcal{X}^n$ (they differ in only one sample) and $\forall S \subseteq \mathcal{S}, \epsilon>0, \delta>0$, the following probability inequality hold
	\begin{equation*}
		 \mathbb{P}[\mathcal{M}(\mathbf{X})\in S] \leq \exp(\epsilon) \mathbb{P}[\mathcal{M}(\mathbf{X}')\in S] + \delta, 
	\end{equation*} when $\delta = 0$ means $\epsilon$-DP.
\end{defn}
This definition strictly controls the distinguishability of outputs under neighboring datasets; hence, algorithms satisfying DP are insensitive to someone individual and protect the individual privacy. \cite{dong2022gaussian} proposes another definition of DP as follows.

\begin{defn}[$f$-Differential Privacy \cite{dong2022gaussian}]
	Let $X$ and $X'$ be two neighboring datasets with domain $\mathcal{X}^n$ and a randomized function $\mathcal{M}:\mathcal{X}^n \to \mathbb{R}^p$. We say that $\mathcal{M}$ satisfies $f$-differential privacy ($f$-DP), if any $\alpha$-level test on the hypotheses of $\mathcal{M}$ outputs, $H_0: \mathcal{M} \text{ based on } X\ vs.\ H_1:\mathcal{M} \text{ based on } X'$, has power function $\beta(\alpha) \leqslant 1 - f(\alpha)$, where $f$ is a convex, continuous, non-decreasing function satisfying $f(\alpha) \leqslant 1-\alpha$ for all $\alpha \in [0,1]$.
\end{defn}

\begin{defn}[Guassian Differential Privacy \cite{dong2022gaussian}]
	If $\mathcal{M}$ satisfies $f$-DP and $f(\alpha) \geqslant \Upphi(\Upphi^{\!\scriptscriptstyle -\!1} (1 - \alpha) - \mu)$ for all $\alpha \in [0,1]$, where $\Upphi$ is the cumulative distribution function of the standard normal distribution and $\mu$ is a constant, then $\mathcal{M}$ satisfies $\mu$-Gaussian differential privacy ($\mu$-GDP).
\end{defn}

There exists a relationship between $(\epsilon,\delta)$-DP and $\mu$-GDP as the following lemma. These DPs hold two important properties: the composition and post-processing. (i) The \textit{post-processing} property means that the output of a DP algorithm will not lose any privacy budget after any post-processing of not contacting the private data. (ii) The \textit{composition} property means that the privacy loss of two DP algorithms is the composition of their privacy losses. Because $\mu$-GDP possesses the better composition property, we state it as the following theorem.

\begin{lem}[\textbf{Corollary 1} in \cite{dong2022gaussian}]
	A mechanism $\mathcal{M}$ satisfies $\mu$-GDP if and only if it is $(\epsilon,\delta(\epsilon))$-DP for $\forall \ \epsilon \geqslant 0$, where $\delta(\epsilon) = \Upphi(-\frac{\epsilon}{\mu} + \frac{\mu}{2})-e^\epsilon \Upphi(-\frac{\epsilon}{\mu} - \frac{\mu}{2})$.
\end{lem}

\begin{thm}[The n-fold \textit{composition}, \textbf{Corollary 2} in \cite{dong2022gaussian}] \label{thm:composition}
	If $\mathcal{M}_i(\mathbf{X},\mathcal{M}_{i-1},...,\mathcal{M}_1): \mathcal{X}^n \times Y_1  \times \cdots \times Y_{i-1}   \to Y_i$  is $\mu_i$-GDP for $i=1,..,T$, then $\mathcal{M}(\mathbf{X}) = \mathcal{M}_1\circ \mathcal{M}_2\circ ...\circ \mathcal{M}_T: \mathcal{X}^n \to Y_1  \times \cdots \times Y_T$ is $\mu = \sqrt{\sum_{i=1}^T \mu^2_i}$-GDP.
\end{thm}

And the simple way to achieve GDP is to add Gaussian noise to the output of the algorithm, called the Gaussian mechanism. The Gaussian mechanism requires the bounded algorithmic sensitivity. The sensitivity of an algorithm always depends on the domain of the private data. Particularly, when the data domain is unbounded, the sensitivity needs to be controlled by truncating the data. The related definition and theorem are as follows.
\begin{defn}[Sensitivity]
	The $l_2$-sensitivity of a function $h:\mathcal{X}^n \to \mathbb{R}^p$ is defined as
	\begin{equation*}
		\Delta_h = \max_{\mathbf{X},\mathbf{X}' \in \mathcal{X}^n} ||f(\mathbf{X}) - f(\mathbf{X}')||_2,
	\end{equation*}
	where $\mathbf{X}$ and $\mathbf{X}'$ are neighboring datasets.
\end{defn}

\begin{thm}[Gaussian Mechanism in \cite{dong2022gaussian}]
	\label{thm:gaussian-mechanism}
	Let $h:\mathcal{X}^n \to \mathbb{R}^p$ be a function with $l_2$-sensitivity $\Delta_h$. Let $\mathbf{g} \in \mathbb{R}^p$ be a standard normal random vector, $\mathbf{g} \sim \mathcal{N}(\mathbf{0}, \mathbf{I}_{p \times p})$. The Gaussian mechanism $\mathcal{G}(\mathbf{X})$ is defined as
	\begin{equation*}
		\mathcal{G}(\mathbf{X}) = h(\mathbf{X}) + \frac{\Delta_h}{\mu} \mathbf{g},
	\end{equation*}$\mathcal{G}(\mathbf{X})$ is $\mu$-GDP.
\end{thm}

\subsection{Problems and Motivation}

In this paper, we assume that both private and public data are independently and identically distributed ($i.i.d.$) from a sub-Gaussian distribution $subG(\Sigma)$, with second-moment matrix $\Sigma$. Note that the sub-Gaussian distribution is a family of distributions including many common distributions such as bounded distributions, Gaussian distribution, and many other unbounded distributions \citep{wainwright2019high}. However, such unboundedness leads to infinite sensitivity for many common statistics, such as the sample mean and sample variance. A widely used approach to address this issue is data truncation,

\begin{equation*}
	\begin{array}{lr}
		T(\mathbf{x})_R =\left\{\begin{array}{lr}
						\mathbf{x}, &\|\mathbf{x}\|_2 \leqslant R \\
						\frac{\mathbf{x}}{\|\mathbf{x}\|_2} R, & \|\mathbf{x}\|_2 > R,
					\end{array}
		\right.
	\end{array}	
\end{equation*} where $\mathbf{x}\in \mathbb{R}^d$ is a sample and $R \in \mathbb{R}$ is the truncation radius. A crucial trade-off in data truncation lies in the choice of the truncation radius. A larger radius preserves more information from the original dataset, incurring less truncation loss, but it leads to higher sensitivity and thus a greater impact of DP noise. Conversely, a smaller radius reduces sensitivity, but at the cost of greater information loss. This trade-off motivates our first core question: 

\begin{center}
	 \textit{How to truncate the unbounded data with a principled radius?}
\end{center}

The second-moment matrix and its inverse are fundamental components in many statistical models and methods, such as linear regression estimators and the Fisher information matrix. However, their use poses significant challenges in the DP setting. We consider the sample second-moment matrix $\hat{\Sigma} = \frac{1}{n} \sum_{i=1}^n \mathbf{x}_i \mathbf{x}_i^{\scriptscriptstyle T}$ and, more generally, the weighted version $\hat{\Sigma} = \frac{1}{n} \sum_{i=1}^n w_i \mathbf{x}_i \mathbf{x}_i^{\scriptscriptstyle T}$. Its DP counterpart is typically formed as $\tilde{\Sigma} = \hat{\Sigma} + \mathbf{G}$, where $\mathbf{G}$ is a symmetric Gaussian noise matrix, and the DP inverse $\tilde{\Sigma}^{\scriptscriptstyle -1}$ is computed accordingly. When $\hat{\Sigma}$ has a large condition number, the ill-condition leads to instability and poor accuracy of $\tilde{\Sigma}^{\scriptscriptstyle -1}$. For achieving accurate and robust estimation of the inverse second-moment matrix, a common approach is to compute the DP regularized form $\tilde{\Sigma}_\lambda^{\scriptscriptstyle -1} = (\hat{\Sigma} + \lambda \mathbf{I} + \mathbf{G})^{\scriptscriptstyle -1}$, where $\lambda > 0$ is a regularization parameter. However, when the original second-moment matrix $\hat{\Sigma}$ is ill-conditioned, a large $\lambda$ is required to stabilize the inverse, which introduces significant bias between $\tilde{\Sigma}_\lambda^{\scriptscriptstyle -1}$ and $\hat{\Sigma}^{\scriptscriptstyle -1}$. This, in turn, undermines the utility of DP methods based on the DP second-moment matrix, such as the DP linear regressions and the DP Newton's method. This challenge motivates our second core question:
\begin{center}
	\textit{How to make the inverse perturbed second-moment matrix more robust and accurate?}	
\end{center}

The isotropic sub-gaussian random vector $\mathbf{z}\sim subG(\mathbf{I}) \in \mathbb{R}^d$ possesses a bound $\|\mathbf{z}\|_2^2 \leq  O(d(1+\log(\frac{2}{\eta})))$ and falls into the radius $R = O(\sqrt{d(1+\log(\frac{2}{\eta}))})$ with high probability $1 - \eta$. The truncation radius $R$ is an appropriate balance between the data utility and the noisy impact. However, the general distribution is with a non-isotropic second-moment matrix $\Sigma$ and the radius $O\big(\sqrt{\Tr(\Sigma)+d\log(\frac{2}{\eta})}\big)$ is impacted by the eigenvalues of $\Sigma$. There is an intuitive method that we find a symmetric matrix $\mathbf{M}$ to transform $\mathbf{z}$ such that $\tilde{\mathbf{z}} = \mathbf{M}^{\!\scriptscriptstyle -\!1\!/\!2}\mathbf{z} $ and 
\begin{equation*}
	\mathbb{E} \tilde{\mathbf{z}} \tilde{\mathbf{z}}^T = \mathbf{M}^{\scriptscriptstyle \!-\!1\!/\!2}\Sigma\mathbf{M}^{\scriptscriptstyle \!-\!1\!/\!2} \approx \I,
\end{equation*} this transforms the non-isotropic variable to the approximate isotropic variable; meanwhile, its second-moment matrix achieves a smaller condition number. This reveals a new perspective: by applying an appropriate transformation, the second-moment matrix of the transformed data can have a condition number approaching one, $\kappa(\mathbb{E}[\tilde{\mathbf{z}}\tilde{\mathbf{z}}^{\scriptscriptstyle T}]) \approx 1$, resulting in samples of principled length and enabling a more stable and accurate inverse. This frees the data truncation from the above dilemma and enables follow-up DP procedures to be performed directly on the transformed data with improved robustness and utility. So we need to choose a proper transformation matrix, and we shall demonstrate that the estimation $\hat{\Sigma}_{pub}$ from the public data is a proper choice.

Such transformations alter the geometry of the data, and consequently, the resulting model parameters may differ from those estimated on the original scale. We design a new loss function for ridge regression that ensures the solution obtained under the transformed setting is equivalent to that of the original ridge regression. For the more complex case of logistic regression, which requires iterative optimization, we similarly construct a corresponding loss function that preserves the original solution. Furthermore, we prove that the invariance property holds throughout each iteration of Newton's method. Detailed discussions of both the ridge and logistic regression cases are provided in the respective sections.

\section{Public-moment-guided Truncation}
\label{sec:methods}
In this section, we propose the public-moment-guided truncation (PMT) in Subsection \ref{subsec:public-truncation}. Subsection \ref{subsec:DP_secondmoment} will show the inverse second-moment matrix estimation of the transformed data is more robust and accurate. In particular, we show the detailed comparison between PMT and the private-data-only method in theoretical results.

\subsection{Public-moment-guided Truncation Method} \label{subsec:public-truncation}

In this subsection, we propose \textbf{Algorithm \ref{alg:PMT}}, which uses the public second-moment matrix to transform the private data to be an approximate isotropic form and then truncates the transformed data within a principled radius. \textbf{Theorem \ref{thm:secondmoment_bound}} guarantees the transformed data with an approximate isotropic form. \textbf{Corollary \ref{cor:utility_trunc}} shows the transformed data can be truncated with a principled radius. 
\begin{thm}[Bound the second-moment matrix]\label{thm:secondmoment_bound}
	Denote that a random vector $\bm{\xi} \in \mathbb{R}^{d \times 1} \sim subG(\Sigma)$, where $\Sigma = \mathbb{E}(\bm{\xi} \bm{\xi}^{\scriptscriptstyle T})$ is the second-moment matrix. Suppose $\bm{\Upsilon}\in \mathbb{R}^{n \times d}$ is a data matrix whose elements $\bm{\upsilon}_i's$ are an $i.i.d.$ sample drawn from $subG(\Sigma)$ and $\hat{\Sigma} = \frac{1}{n}\bm{\Upsilon}^{\scriptscriptstyle T}\bm{\Upsilon}$ is an estimation of the second-moment matrix. Then $\tilde{\bm{\xi}} = \hat{\Sigma}^{\scriptscriptstyle -\!1\!/\!2}\bm{\xi} \sim subG(\hat{\Sigma}^{\scriptscriptstyle -\!1\!/\!2}\Sigma\hat{\Sigma}^{\scriptscriptstyle -\!1\!/\!2})$ and, with at least probability $1 - 2\eta$,  
	\begin{equation*}
		L\mathbf{I}\preceq \hat{\Sigma}^{\scriptscriptstyle -\!1\!/\!2}\Sigma\hat{\Sigma}^{\scriptscriptstyle -\!1\!/\!2}\preceq U\mathbf{I},
	\end{equation*} where {\scriptsize $L = \frac{n}{(\sqrt{n} + O(\sqrt{d} + \sqrt{\log(\frac{1}{\eta})}))^2}$} and {\scriptsize $U = \frac{n}{(\sqrt{n} - O(\sqrt{d} + \sqrt{\log(\frac{1}{\eta})}))^2}$}.
\end{thm}

\begin{cor}[Untility of truncation]\label{cor:utility_trunc}
	Denote that random vectors $\bm{\xi}_i \in \mathbb{R}^{d \times 1} \overset{i.i.d.}{\sim} subG(\Sigma),\ i=1,...,n_{\xi}$, where $\Sigma = \mathbb{E}(\bm{\xi}_i \bm{\xi}_i^T)$ is the second-moment matrix. Suppose $\bm{\Upsilon}\in \mathbb{R}^{n_{\!\scriptscriptstyle \upsilon} \times d}$ is a data matrix whose elements are an $i.i.d.$ sample drawn from $subG(\Sigma), \ i.i.d.$ and $\hat{\Sigma} \!= \frac{1}{n_{\!\scriptscriptstyle \upsilon}}\bm{\Upsilon}^T\bm{\Upsilon}$ is an estimation of the second-moment matrix. Let $\tilde{\Sigma} = \hat{\Sigma}^{\scriptscriptstyle -\!1\!/\!2}\Sigma\hat{\Sigma}^{\scriptscriptstyle -\!1\!/\!2}$, then $\tilde{\bm{\xi}}_i = \hat{\Sigma}^{\scriptscriptstyle -\!1\!/\!2}\bm{\xi}_i \sim subG(\tilde{\Sigma})$ and, with at least probability $1 - 3\eta$,
	\begin{equation*}
		\|\tilde{\bm{\xi}}_i\|_2^2 \leq \Tr(\tilde{\Sigma})+O(d\log(\frac{2n_{\xi}}{\eta})) \leq O(d(1 +\log(\frac{2n_{\xi}}{\eta}))), \ \ \forall i \in [n_{\xi}].
	\end{equation*}.
\end{cor}

\begin{algorithm}[H]
	\caption{Public-Moment-guided Truncation (PMT)}\label{alg:PMT}
	\begin{algorithmic}[1]
		\STATE {\bfseries Input:} Private dataset $\{\bm{\xi}_i \}^{n_\xi}_{i=1}$, the public second-moment matrix $\hat{\Sigma} = \frac{1}{n_{\upsilon}}\sum_{i=1}^{n_{\upsilon}} \bm{\upsilon}_i\bm{\upsilon}_i^T$, parameters $d$, $n_\xi$, $n_{\upsilon}$ and $\eta$.
		
		\STATE {\bfseries Transform private data:} $\tilde{\bm{\xi}_i} = \hat{\Sigma}^{\scriptscriptstyle -\!1\!/\!2} \bm{\xi}_i, \ i=1,...,n_\xi.$
		\STATE {\bfseries Truncate data:} for every transformed data $\tilde{\bm{\xi}_i}, i = 1,...,n_{\xi}$,
			\WHILE {$i \in [n_\xi]$} 

				\IF{$\|\tilde{\bm{\xi}_i}\|_2 \geq \sqrt{d(1 + \log(\frac{2n_{\xi}}{\eta}))}$} 
					\STATE $\tilde{\bm{\xi}_i} \leftarrow  \sqrt{d(1 + \log(\frac{2n_{\xi}}{\eta}))}\cdot \frac{\tilde{\bm{\xi}_i}}{\|\tilde{\bm{\xi}_i}\|_2}$,
				\ELSE
					\STATE $\tilde{\bm{\xi}_i}$ is itself.  
				\ENDIF
			\ENDWHILE
			
		\STATE {\bfseries Output:}  Truncated dataset $\{\tilde{\bm{\xi}_i}\}^{n_\xi}_{i=1}$ and the public second-moment matrix $\hat{\Sigma}$.
	\end{algorithmic}
\end{algorithm}
\begin{rem}
	The algorithm only needs a second-moment matrix $\hat{\Sigma}$. Hence, we can weaken the public data requirement to a public second-moment estimation $\hat{\Sigma}$ which is easier to attain and more safe for privacy.
\end{rem}

\subsection{Robust Private Second-moment Matrix}\label{subsec:DP_secondmoment}
Next, we give the DP second-moment matrix of the transformed data and guarantee its $\mu$-GDP.
\begin{thm}[Private Second-moment]\label{thm:DP_secmom}
	Denote that private random $\bm{\xi}_i \in \mathbb{R}^{d \times 1} \overset{i.i.d.}{\sim} subG(\Sigma),\ i=1,...,n_{\xi}$, where $\Sigma = \mathbb{E}(\bm{\xi} \bm{\xi}^T)$ is the second-moment matrix. Suppose $\bm{\Upsilon}\in \mathbb{R}^{n_{\upsilon} \times d}$ is a public data matrix whose elements $\bm{\upsilon}_i$'s are an $i.i.d.$ sample drawn from $subG(\Sigma)$. If the transformed data $\tilde{\bm{\xi}}_i$'s are from \textbf{Algorithm \ref{alg:PMT}} with $1> \eta > 0$ , the DP transformation second-moment matrix
	\begin{equation*}
		\tilde{\Sigma}_{\!\scriptscriptstyle D\!P} = \frac{1}{n_\xi} \sum_{i=1}^{n_\xi} \tilde{\bm{\xi}}_i \tilde{\bm{\xi}}_i^T + \mathbf{G}
	\end{equation*} satisfies the $\mu$-GDP, where $\mathbf{G} \sim SG_d(\sigma^2)$ is a $d-$dimensional symmetric Gaussian random matrix with the parameter $\sigma = \frac{2d(1 + \log(\frac{2n_{\xi}}{\eta}))}{\mu \cdot n_\xi}$. And with at least probability $1 - O(\eta)$, the $\mathbf{G}$ has bound as following:
	\begin{equation}\label{eq:GMbound}
		\|\mathbf{G}\|_2 \leq \frac{\sqrt{d^3 \log(\frac{2d}{\eta})}(1 + \log(\frac{2n_{\xi}}{\eta}))}{\mu \cdot n_\xi}.
	\end{equation}
\end{thm}

There exist relationships between the transformed matrices and the original matrices. Moreover, we consider the regularized second-moment matrix, which is more general than the second-moment matrix. The following theorem shows the transformed data offer the more robust and accurate inverse second-moment matrix.
\begin{thm}[Robust DP Inverse]\label{thm:public_inverse_error}
	Denote that private random $\bm{\xi}_i \in \mathbb{R}^{d \times 1} \overset{i.i.d.}{\sim} subG(\Sigma),\ i=1,...,n_{\xi}$, where $\Sigma = \mathbb{E}(\bm{\xi}_i \bm{\xi}_i^T)$ is the second-moment matrix and $\hat{\Sigma}_{\scriptscriptstyle \!\xi} = \frac{1}{n_{\scriptscriptstyle \!\xi}} \sum_{i=1}^{n_\xi} \bm{\xi}_i\bm{\xi}_i^{T}$. Suppose $\bm{\Upsilon}\in \mathbb{R}^{n_{\upsilon} \times d}$ is a public data matrix whose elements $\bm{\upsilon}_i$'s are an $i.i.d.$ sample drawn from $subG(\Sigma)$ and $\hat{\Sigma}_{\scriptscriptstyle \upsilon} = \frac{1}{n_{\!\scriptscriptstyle \upsilon}}\bm{\Upsilon}^T\bm{\Upsilon}$ is an estimation of the second-moment matrix. The transformed data $\tilde{\bm{\xi}}_i$s and their data matrix $\tilde{\bm{\Upxi}}\in \mathbb{R}^{n_\xi \times d}$ are from \textbf{Algorithm \ref{alg:PMT}} with $1> \eta > 0$. Considering the regularized second-moment matrix with the regularization parameter $\lambda$
	\begin{equation*}
		\tilde{\Sigma}_{\scriptscriptstyle \!\xi} + \lambda \hat{\Sigma}_{\!\scriptscriptstyle \upsilon}^{\scriptscriptstyle -\!1} = \frac{1}{n_\xi} \sum_{i=1}^{n_\xi} \tilde{\bm{\xi}}_i \tilde{\bm{\xi}}_i^T + \lambda \hat{\Sigma}_{\scriptscriptstyle \upsilon}^{\scriptscriptstyle -\!1},
	\end{equation*} we have the following results:

	\textbf{1. Recovery.} With at least probability $1 - O(\eta)$,
		\begin{equation*}
			\hat{\Sigma}_{\scriptscriptstyle \upsilon}^{\scriptscriptstyle \!1\!/\!2}(\tilde{\Sigma}_{\scriptscriptstyle \!\xi} + \lambda \hat{\Sigma}_{\!\scriptscriptstyle \upsilon}^{\scriptscriptstyle -\!1})\hat{\Sigma}_{\scriptscriptstyle \upsilon}^{\scriptscriptstyle \!1\!/\!2} = \hat{\Sigma}_{\scriptscriptstyle \xi} + \lambda \mathbf{I}. 
		\end{equation*} 

	\textbf{2. Inverse error.} When the public data satisfies $\sqrt{n_{\!\scriptscriptstyle \upsilon}} \geqslant O(\sqrt{d} + \sqrt{\log(1/\eta)})$ and the private data $n_{\!\scriptscriptstyle \xi}$ makes $\frac{\sqrt{d^3 \log(\frac{2d}{\eta})}(1 + \log(\frac{2n_{\!\scriptscriptstyle \xi}}{\eta}))}{(L(1 - O(\sqrt{\frac{d}{n_{\!\scriptscriptstyle \xi}}}+\sqrt{\frac{\log(1/\eta)}{n_{\!\scriptscriptstyle \xi}}}))^2 + \lambda\|\hat{\Sigma}_{\!\scriptscriptstyle \upsilon}\|^{\scriptscriptstyle \!-\!1})\cdot \mu \cdot n_{\! \scriptscriptstyle \xi}} \leqslant \frac{1}{2}$, with at least probability $1 - O(\eta)$, we have
		\begin{equation}\label{eq:inverse_trans_error}
			\| (\tilde{\Sigma}_{\scriptscriptstyle \!\xi} + \lambda \hat{\Sigma}_{\!\scriptscriptstyle \upsilon}^{\scriptscriptstyle -\!1} + \mathbf{G})^{\scriptscriptstyle \!-\!1} - (\tilde{\Sigma}_{\scriptscriptstyle \!\xi} + \lambda \hat{\Sigma}_{\!\scriptscriptstyle \upsilon}^{\scriptscriptstyle -\!1})^{\scriptscriptstyle -\!1}\|_2 \leqslant \frac{\sqrt{d^3 \log(\frac{2d}{\eta})}}{\mu n_{\! \scriptscriptstyle \xi}} \cdot \frac{(1 + \log(\frac{2n_{\!\scriptscriptstyle \xi}}{\eta}))}{L^2(1 -  O(\sqrt{\frac{d}{n_{\!\scriptscriptstyle \xi}}}+\sqrt{\frac{\log(1/\eta)}{n_{\!\scriptscriptstyle \xi}}}))^4 +  \lambda^2\|\hat{\Sigma}_{\!\scriptscriptstyle \upsilon}\|^{\scriptscriptstyle \!-\!2}},
		\end{equation} where $L = \frac{n_{\!\scriptscriptstyle \upsilon}}{(\sqrt{n_{\!\scriptscriptstyle \upsilon}} + O(\sqrt{d} + \sqrt{\log(\frac{1}{\eta})}))^2}$. 
		
		Moreover, with at least probability $1 - O(\eta)$, we have the error bound about the original inverse regularized second-moment matrix
		\begin{equation}\label{eq:inverse_orig_error}
			\| \hat{\Sigma}_{\scriptscriptstyle \upsilon}^{\scriptscriptstyle -\!1\!/\!2}(\tilde{\Sigma}_{\scriptscriptstyle \!\xi} + \lambda \hat{\Sigma}_{\!\scriptscriptstyle \upsilon}^{\scriptscriptstyle -\!1} + \mathbf{G})^{\scriptscriptstyle \!-\!1}\hat{\Sigma}_{\scriptscriptstyle \upsilon}^{\scriptscriptstyle -\!1\!/\!2} - (\hat{\Sigma}_{\scriptscriptstyle \xi} + \lambda \mathbf{I})^{\scriptscriptstyle -\!1}\|_2 \leqslant \frac{\sqrt{d^3 \log(\frac{2d}{\eta})}}{\mu n_{\! \scriptscriptstyle \xi}} \cdot \frac{\|\hat{\Sigma}_{\scriptscriptstyle \upsilon}^{\scriptscriptstyle -\!1}\|(1 + \log(\frac{2n_{\!\scriptscriptstyle \xi}}{\eta}))}{L^2(1 -  O(\sqrt{\frac{d}{n_{\!\scriptscriptstyle \xi}}}+\sqrt{\frac{\log(1/\eta)}{n_{\!\scriptscriptstyle \xi}}}))^4 +  \lambda^2\|\hat{\Sigma}_{\!\scriptscriptstyle \upsilon}\|^{\scriptscriptstyle \!-\!2}}.
		\end{equation}
\end{thm}

It is worth noting that the terms involving public data, $\hat{\Sigma}_{\scriptscriptstyle \upsilon}$, are retained in the bound Eq.~\eqref{eq:inverse_orig_error}, since they typically have negligible impact on the overall utility of the inverse matrix. The following theorem shows that, based solely on private data, the utility of the inverse DP second-moment matrix is impacted by $\hat{\Sigma}_{pri}$ or $\Sigma$. Without loss of generality, we assume that the $l_2$-norm of private data $\bm{\xi}_i \in \mathbb{R}^d$ is bounded by $\sqrt{\Tr(\Sigma) + d\log(\frac{n_{\scriptscriptstyle \!\xi}}{\eta})},\ \ w.p. \ 1- \eta >0$.

\begin{thm}\label{thm:nonpublic_inverse_error}
	Denote that private random $\bm{\xi}_i \in \mathbb{R}^{d \times 1}, \|\bm{\xi}_i\|\leqslant \sqrt{\Tr(\Sigma) + d\log(\frac{n_{\scriptscriptstyle \!\xi}}{\eta})} ,\ i=1,...,n_{\xi}$ with the second-moment $\Sigma = \mathbb{E}(\bm{\xi}_i \bm{\xi}_i^T)$ and $\hat{\Sigma}_{\scriptscriptstyle \!\xi} = \frac{1}{n_{\scriptscriptstyle \!\xi}} \sum_{i=1}^{n_\xi} \bm{\xi}_i\bm{\xi}_i^{T}$. Considering the regularized second-moment matrix with the regularization parameter $\lambda$
	\begin{equation*}
		\hat{\Sigma}_{\scriptscriptstyle \xi} + \lambda \mathbf{I},
	\end{equation*} and its $\mu$-GDP form
	\begin{equation*}
		\hat{\Sigma}_{\scriptscriptstyle \xi} + \lambda \mathbf{I} + \mathbf{G},
	\end{equation*} where $\mathbf{G} \sim SG_d(\sigma^2),\ \sigma = \frac{2(\Tr(\Sigma) + d\log(\frac{n_{\scriptscriptstyle \!\xi}}{\eta}))}{\mu \cdot n_\xi}$. If $n_\xi$ makes $\frac{\sqrt{d^{\scriptscriptstyle 3}\log(\frac{2d}{\eta})} (d^{\scriptscriptstyle \!-\!1} \!\Tr(\Sigma) + \log(\frac{n_{\scriptscriptstyle \!\xi}}{\eta})) }{\mu \cdot n_\xi \big(\lambda_{\scriptscriptstyle min} (\Sigma)\big( 1 - O(\sqrt{\frac{d}{n_{\!\scriptscriptstyle \xi}}}+\sqrt{\frac{\log(1/\eta)}{n_{\!\scriptscriptstyle \xi}}})\big)^2 + \lambda \big)}\leqslant \frac{1}{2}$, we have the error bound of the inverse regularized second-moment matrix, with at least probability $1 - O(\eta)$, 
	\begin{equation*}
		\|(\hat{\Sigma}_{\scriptscriptstyle \xi} + \lambda \mathbf{I} + \mathbf{G})^{\scriptscriptstyle -\!1} - (\hat{\Sigma}_{\scriptscriptstyle \xi} + \lambda \mathbf{I})^{\scriptscriptstyle -\!1}\| \leqslant \frac{\sqrt{d^3 \log(\frac{2d}{\eta})}}{\mu n_{\! \scriptscriptstyle \xi}} \cdot \frac{ (d^{\scriptscriptstyle \!-\!1}\!\Tr(\Sigma) + \log(\frac{n_{\scriptscriptstyle \!\xi}}{\eta}))}{\lambda^2_{\scriptscriptstyle min} (\Sigma)\big(1 - O(\sqrt{\frac{d}{n_{\!\scriptscriptstyle \xi}}}+\sqrt{\frac{\log(1/\eta)}{n_{\!\scriptscriptstyle \xi}}})\big)^4 +  \lambda^2},
	\end{equation*} where $\lambda_{\min}(\Sigma)$ represents the smallest eigenvalue of  $\Sigma$. 
\end{thm}

We compare the bounds in \textbf{Theorem \ref{thm:public_inverse_error}} and \textbf{Theorem \ref{thm:nonpublic_inverse_error}}. For convenience, we can consider that the amount of private data $n_\xi$ is enough and makes $\lambda_{\scriptscriptstyle min} (\Sigma)(1 - O(\sqrt{\frac{d}{n_{\!\scriptscriptstyle \xi}}}+\sqrt{\frac{\log(1/\eta)}{n_{\!\scriptscriptstyle \xi}}}))^2 = \lambda_{\scriptscriptstyle min} (\Sigma)(1 - o(1))^2$. Actually, the appropriate regularization parameter $\lambda$ is very small so that it doesn't produce large bias to the original matrix; especially, it may be $o(1)$ compared to large $\lambda_{\scriptscriptstyle min} (\Sigma)$, or approximate to small $\lambda_{\scriptscriptstyle min} (\Sigma)$ . Hence, in \textbf{Theorem \ref{thm:nonpublic_inverse_error}}, the necessary $n_{\! \scriptscriptstyle \xi}$ to hold the stable condition is simplified as 

\begin{equation*}
\frac{\sqrt{d^3 \log(\frac{2d}{\eta})}}{\mu n_{\! \scriptscriptstyle \xi}} \cdot \frac{ d^{\scriptscriptstyle \!-\!1}\!\Tr(\Sigma) + \log(\frac{n_{\scriptscriptstyle \!\xi}}{\eta}) }{\lambda_{\scriptscriptstyle min} (\Sigma)\big(1 - o(1)\big)^2 + o(1)} \simeq  \frac{\sqrt{d^3 \log(\frac{2d}{\eta})}}{\mu n_{\! \scriptscriptstyle \xi}} \cdot (\bar{\kappa}(\Sigma) + \|\Sigma^{\scriptscriptstyle -\!1}\|\log(\frac{n_{\scriptscriptstyle \!\xi}}{\eta})) \leqslant \frac{1}{2},
\end{equation*}
and the error bound is simplified as
\begin{equation*}
\frac{\sqrt{d^3 \log(\frac{2d}{\eta})}}{\mu n_{\! \scriptscriptstyle \xi}} \cdot \frac{ d^{\scriptscriptstyle \!-\!1}\!\Tr(\Sigma) + \log(\frac{n_{\scriptscriptstyle \!\xi}}{\eta}) }{\lambda^2_{\scriptscriptstyle min} (\Sigma)\big(1 - o(1)\big)^4 + o(1)} \simeq  \frac{\sqrt{d^3 \log(\frac{2d}{\eta})}}{\mu n_{\! \scriptscriptstyle \xi}} \cdot \|\Sigma^{\scriptscriptstyle -\!1}\| \cdot (\bar{\kappa}(\Sigma) + \|\Sigma^{\scriptscriptstyle -\!1}\|\log(\frac{n_{\scriptscriptstyle \!\xi}}{\eta})),
\end{equation*} where $\bar{\kappa}(\Sigma) = d^{\scriptscriptstyle \!-\!1}\!\sum_{i=1}^{d} \kappa_i(\Sigma) \geqslant 1 $, $\kappa_i(\Sigma) = \frac{\lambda_{\scriptscriptstyle i}(\Sigma)}{\lambda_{\scriptscriptstyle \min}(\Sigma)} \geqslant 1$ and $\lambda_i(\cdot)$ is the $i$-th eigenvalue. It is worth noting that an ill-conditioned matrix can significantly inflate the first some condition numbers, thereby making $\bar{\kappa}(\Sigma)$ disproportionately large. This behavior is commonly encountered in practical scenarios. Then, we analyze the advantages of \textbf{Theorem~\ref{thm:public_inverse_error}}.
\begin{itemize}
	\item \textbf{Few necessary public data.} $U$ and $L$ tend to $1$ fast as the increase of the amount of public data $n_{\!\scriptscriptstyle \upsilon}$. $U = \Big(1 - O\Big(\sqrt{\frac{d}{n_{\!\scriptscriptstyle \upsilon}}} + \sqrt{\frac{\log(1/\eta)}{n_{\!\scriptscriptstyle \upsilon}}}\Big)\Big)^{-2}$ and $L= \Big(1 + O\Big(\sqrt{\frac{d}{n_{\!\scriptscriptstyle \upsilon}}} + \sqrt{\frac{\log(1/\eta)}{n_{\!\scriptscriptstyle \upsilon}}}\Big)\Big)^{-2}$ show that the necessary amount of public data $n_{\!\scriptscriptstyle \upsilon}$ controlling the ill-condition just needs $O(d + \log(1/\eta))$ to make $U$ and $L$ tend to $1$.

	\item \textbf{Weakly impacted by $\Sigma$.} The bounds and the condition about the necessary amount of private data are independent of the average condition number $\bar{\kappa}(\Sigma)$ and eliminate a factor $\|\Sigma^{\scriptscriptstyle -\!1}\|$.
	\begin{enumerate}
		\item \textbf{Better robustness condition.} That is a significant improvement for resisting the impact of DP. Namely, it makes the stable condition, 
	\begin{equation*}
	\frac{\sqrt{d^3 \log(\frac{2d}{\eta})}(1 + \log(\frac{n_{\!\scriptscriptstyle \xi}}{\eta}))}{\mu n_{\! \scriptscriptstyle \xi} \cdot (L(1 - o(1))^2 + \lambda\|\hat{\Sigma}_{\!\scriptscriptstyle \upsilon}\|^{\scriptscriptstyle \!-\!1})}\simeq  \frac{\sqrt{d^3 \log(\frac{2d}{\eta})}}{\mu n_{\! \scriptscriptstyle \xi}} \cdot \log(\frac{n_{\scriptscriptstyle \xi}}{\eta}) \leqslant \frac{1}{2},
	\end{equation*} easier to satisfy, meaning better robustness and stable efficiency. Particularly, the term $\lambda\|\hat{\Sigma}_{\!\scriptscriptstyle \upsilon}\|^{\scriptscriptstyle \!-\!1}$ degrades the impact of the regularization parameter.
		\item  \textbf{Tighter error bound.} The error bound of \textbf{Theorem \ref{thm:public_inverse_error}} is simplified as 
	\begin{equation*}
	\frac{\sqrt{d^3 \log(\frac{2d}{\eta})}}{\mu n_{\! \scriptscriptstyle \xi}} \cdot \frac{\|\hat{\Sigma}_{\scriptscriptstyle \upsilon}^{\scriptscriptstyle -\!1}\|(1 + \log(\frac{2n_{\!\scriptscriptstyle \xi}}{\eta}))}{L^2(1 -  o(1))^4 +  o(1)} \simeq  \frac{\sqrt{d^3 \log(\frac{2d}{\eta})}}{\mu n_{\! \scriptscriptstyle \xi}} \cdot \log(\frac{n_{\scriptscriptstyle \xi}}{\eta}) \cdot \|\Sigma^{\scriptscriptstyle -\!1}\|,
	\end{equation*} where the public data $n_{\scriptscriptstyle \upsilon}$ makes $L$ tend to $1$ and $\|\hat{\Sigma}_{\scriptscriptstyle \upsilon}^{\scriptscriptstyle -\!1}\|$ tend to $\|\Sigma^{\scriptscriptstyle -\!1}\|$. A smaller error bound represents improved utility of the DP inverse second-moment matrix, making the DP estimation substantially more reliable.
	\end{enumerate}
	  
	\item \textbf{Weakly depending on the regularization parameter.} The robustness and utility are dependent on the regularization parameter $\lambda$ weakly. The $\lambda$ term is always divided by the largest eigenvalue, $\|\hat{\Sigma}_{\scriptscriptstyle \upsilon}\|$. That avoids the underlying bias $\|(\hat{\Sigma} + \lambda \mathbf{I})^{\scriptscriptstyle -\!1} - \hat{\Sigma}^{\scriptscriptstyle -\!1}\|$ resulting from the large regularization $\lambda$, where the large regularization $\lambda$ is generally for the inverse stability.

\end{itemize} 

\section{Application to Penalized Regression}\label{sec:DP_ridge}

In this section, we apply the proposed method to accommodate penalized regression. For notational simplicity, we use ridge regression based on the PMT and the transformed data for demonstration.
Specifically, we first consider the linear regression model:
\begin{equation*}
	\mathbf{y}_{\!\!\scriptscriptstyle A} = \mathbf{A} \bm{\beta}  + \bm\epsilon,
\end{equation*} 
where $\mathbf{A}$ is a $n \times d$ data matrix and $\mathbf{A} = [A_1^T\ ...\  A_n^T]^T$, $A_i \sim subG(\Sigma) \in \mathbb{R}^{d \times 1}$ is the $i$-th sample. $\mathbf{y}_{\! \! \scriptscriptstyle A}$ is the response variable vector. The noise vector $\bm \epsilon \sim \mathcal{N}(0,\sigma^2 \mathbf{I}_{n \times n})$. For ridge regression, the general loss function is defined as
\begin{equation*}
	\mathcal{L}(\bm{\beta};\mathbf{A},\mathbf{y}_{\! \! \scriptscriptstyle A}) = \frac{1}{2n} ||\mathbf{y}_{\! \! \scriptscriptstyle A} - \mathbf{A}\bm{\beta}||_2^2 + \frac{\lambda}{2} ||\bm{\beta}||_2^2,
\end{equation*}
where $\lambda$ is the regularization parameter. The analytical solution of the ridge regression is given by
\begin{equation*}
	\hat{\bm{\beta}} = (\frac{\mathbf{A}^{\!\scriptscriptstyle T} \mathbf{A}}{n} + \lambda \mathbf{I}_{d \times d})^{-1} \frac{\mathbf{A}^{\!\scriptscriptstyle T}  \mathbf{y}_{\! \! \scriptscriptstyle A}}{n}.
\end{equation*}

\subsection{Public-moment-guided Ridge Regression}
We firstly give the DP ridge regression based on PMT. The algorithm is as follows:

	\begin{algorithm}[H]
		\caption{Differentially Private PMT Ridge Regression (DP-PMTRR)}\label{alg:DP-PMTRR}
		\begin{algorithmic}[1]
			\STATE {\bfseries Input:} Private dataset $\{\bm{\xi}_i = (\mathbf{A}_i,y_{\!\scriptscriptstyle A_i})^T\in \mathbb{R}^{d+1} \}^{n_{\!\scriptscriptstyle A}}_{i=1}$, public dataset $\{\bm{\upsilon}_i =(\mathbf{B}_i,y_{\!\scriptscriptstyle B_i})^T\in \mathbb{R}^{d+1}\}^{n_{\!\scriptscriptstyle B}}_{i=1}$. Parameters $\mu$, $\lambda$, $d$, $n_{\!\scriptscriptstyle A}$, $n_{\!\scriptscriptstyle B}$ and $\eta$.
			
			\STATE {\bfseries Transform covariates:} 
		$(\{\tilde{\mathbf{A}}_i\}^{n_{\!\scriptscriptstyle A}}_{i=1},\hat{\Sigma}_{\!\scriptscriptstyle B}) =PMT(\{\mathbf{A}_i \}^{n_{\!\scriptscriptstyle A}}_{i=1},\{\mathbf{B}_i \}^{n_{{\!\scriptscriptstyle B}}}_{i=1},  d, n_{\!\scriptscriptstyle A}, n_{\!\scriptscriptstyle B},\eta)$.
			\STATE {\bfseries Transform responses:} 
		$(\{\tilde{y}_{\!\scriptscriptstyle A_i}\}^{n_{\!\scriptscriptstyle A}}_{i=1},\hat{\sigma}^2_{\!\scriptscriptstyle B}) =PMT(\{y_{\!\scriptscriptstyle A_i} \}^{n_{\!\scriptscriptstyle A}}_{i=1},\{y_{\!\scriptscriptstyle B_i} \}^{n_{\!\scriptscriptstyle B}}_{i=1}, 1, n_{\!\scriptscriptstyle A}, n_{\!\scriptscriptstyle B},\eta)$.
			\STATE {\bfseries Private parameter:} 
				\begin{center}
					$\sigma_1 = \frac{2d(1 + \log(\frac{2n_{\!\scriptscriptstyle A}}{\eta}))}{\mu \cdot n_{\!\scriptscriptstyle A}}\Big.$, 
					$\sigma_2 = \frac{2\sqrt{d}(1 + \log(\frac{2n_{\!\scriptscriptstyle A}}{\eta}))}{\mu \cdot n_{\!\scriptscriptstyle A}}.$
				\end{center}
			\STATE {\bfseries Gaussian mechanism:} $\bm{\tilde{\beta}}^{\scriptscriptstyle D\!P} =  \Big(\frac{\tilde{\mathbf{A}}^T \tilde{\mathbf{A}}}{n_{\!\scriptscriptstyle A}} + \lambda \hat{\Sigma}_{\!\scriptscriptstyle B}^{\scriptscriptstyle -\!1} + \mathbf{G}\Big)^{\!\!\scriptscriptstyle -\!1}\Big(\frac{\tilde{\mathbf{A}}^T \tilde{\mathbf{y}}_{\!\scriptscriptstyle A}}{n_{\!\scriptscriptstyle A}} + \mathbf{g}\Big)$, where $\mathbf{G} \sim SG_d(\sigma_1^2)$ and $\mathbf{g} \sim \mathcal{N}(0,\sigma_2^2\mathbf{I})$. 
			\STATE {\bfseries Recover:} $\bm{\bar{\beta}}^{\scriptscriptstyle D\!P} \leftarrow \hat{\sigma}_{\!\scriptscriptstyle B}\cdot\hat{\Sigma}_{\!\scriptscriptstyle B}^{\scriptscriptstyle -\!1\!/\!2}\cdot\bm{\tilde{\beta}}^{\scriptscriptstyle D\!P}$.
			\STATE {\bfseries Output:}  DP estimator $\bm{\bar{\beta}}^{\scriptscriptstyle D\!P}$.
		\end{algorithmic}
	\end{algorithm}
\begin{rem}
	Note that $\tilde{y}_{\!\scriptscriptstyle A_i} = \hat{\sigma}_{\!\scriptscriptstyle B}^{\scriptscriptstyle -\!1} y_{\!\scriptscriptstyle A_i}$ and $\hat{\sigma}_{\!\scriptscriptstyle B} =  \sqrt{\frac{1}{n_{\!\scriptscriptstyle B}}\sum_{i=1}^{n_{\!\scriptscriptstyle B}} y_{\!\scriptscriptstyle B_i}^2}$ is the root of second-moment estimation of the public response $\mathbf{y}_{\!\!\scriptscriptstyle B}$.
\end{rem}

\subsection{Theoretical Analysis}
We firstly show that there exists the transformation between the original ridge regression and the PMT ridge regression. 
\begin{thm}[Equivalent ridge estimator]\label{thm:beta_transform}
	The setting is same as \textbf{Algorithm \ref{alg:DP-PMTRR}}. Let the original ridge regression be $\bm{\hat{\beta}} =  \Big(\frac{\mathbf{A}^{\!\scriptscriptstyle T} \mathbf{A}}{n_{\!\scriptscriptstyle A}} + \lambda \mathbf{I}\Big)^{\!\!\scriptscriptstyle -\!1}\Big(\frac{\mathbf{A}^{\!\scriptscriptstyle T} \mathbf{y}_{\!\scriptscriptstyle A}}{n_{\!\scriptscriptstyle A}}\Big)$ and the PMT ridge regression is $\bm{\tilde{\beta}} =  \Big(\frac{\tilde{\mathbf{A}}^{\!\scriptscriptstyle T} \tilde{\mathbf{A}}}{n_{\!\scriptscriptstyle A}} + \lambda \hat{\Sigma}_{\!\scriptscriptstyle B}^{\scriptscriptstyle -\!1}\Big)^{\!\!\scriptscriptstyle -\!1}\Big(\frac{\tilde{\mathbf{A}}^{\!\scriptscriptstyle T} \tilde{\mathbf{y}}_{\!\scriptscriptstyle A}}{n_{\!\scriptscriptstyle A}}\Big)$, then with at least probability $1 - O(\eta)$, we have
	\begin{equation}\label{eq:beta_transform}
		\bm{\hat{\beta}} = \hat{\sigma}_{\!\scriptscriptstyle B}\hat{\Sigma}_{\!\scriptscriptstyle B}^{\scriptscriptstyle -\!1\!/\!2}\bm{\tilde{\beta}}.
	\end{equation}
\end{thm}

\begin{rem}
	The PMT ridge regression estimator $\bm{\tilde{\beta}} =  \Big(\frac{\tilde{\mathbf{A}}^{\!\scriptscriptstyle T} \tilde{\mathbf{A}}}{n_{\!\scriptscriptstyle A}} + \lambda \hat{\Sigma}_{\!\scriptscriptstyle B}^{\scriptscriptstyle -\!1}\Big)^{\!\!\scriptscriptstyle -\!1}\Big(\frac{\tilde{\mathbf{A}}^{\!\scriptscriptstyle T} \tilde{\mathbf{y}}_{\!\scriptscriptstyle A}}{n_{\!\scriptscriptstyle A}}\Big)$ is the minimizer of the loss function $\mathcal{L}(\hat{\Sigma}_{\!\scriptscriptstyle B}^{\scriptscriptstyle -\!1\!/\!2}\bm{\beta};\mathbf{A},\tilde{\mathbf{y}}_{\! \! \scriptscriptstyle A}) = \frac{1}{2n} ||\tilde{\mathbf{y}}_{\! \! \scriptscriptstyle A} - \tilde{\mathbf{A}}\bm{\beta}||_2^2 + \frac{\lambda}{2} ||\hat{\Sigma}_{\!\scriptscriptstyle B}^{\scriptscriptstyle -\!1\!/\!2}\bm{\beta}||_2^2$. 
\end{rem}
\textbf{Theorem \ref{thm:beta_transform}} guarantees that the "recover" operation in the $line \ 6$ of \textbf{Algorithm \ref{alg:DP-PMTRR}} is right. Moreover, it holds regularization parameter $\lambda$ invariant. The next theorem will show the utility of the DP-PMTRR.

\begin{thm}[DP-PMTRR]\label{thm:DP-PMTRR} \textbf{Algorithms \ref{alg:DP-PMTRR}} satisfies $\sqrt{2}\mu$-GDP. When the number of public data makes $\sqrt{n_{\scriptscriptstyle \!B}} \geqslant O(\sqrt{d} + \sqrt{\log(1/\eta)})$ and the private data $n_{\!\scriptscriptstyle A}$ makes $\frac{\sqrt{d^3 \log(\frac{2d}{\eta})}(1 + \log(\frac{2n_{\!\scriptscriptstyle \!A}}{\eta}))}{(L(1 - O(\sqrt{\frac{d}{n_{\!\scriptscriptstyle \!A}}}+\sqrt{\frac{\log(1/\eta)}{n_{\!\scriptscriptstyle \!A}}}))^2 + \lambda\|\hat{\Sigma}_{\!\scriptscriptstyle B}\|^{\scriptscriptstyle \!-\!1})\cdot \mu \cdot n_{\! \scriptscriptstyle \!A}} \leqslant$~$\frac{1}{2}$, the DP estimator $\bm{\tilde{\beta}}^{\scriptscriptstyle D\!P}$ satisfies, with at least probability $1 - O(\eta)$,
\begin{equation*}
	\|\bm{\tilde{\beta}}^{\scriptscriptstyle D\!P} - \bm{\tilde{\beta}}\| \leqslant O\Big(\frac{\sqrt{d^3 \log(\frac{2d}{\eta})}}{\mu n_{\! \scriptscriptstyle A}} \cdot \frac{\hat{\sigma}_{\scriptscriptstyle \!B}^{\scriptscriptstyle -\!1} \|\hat{\Sigma}_{\scriptscriptstyle \!B}^{\scriptscriptstyle -\!1\!/\!2}\| \|\Sigma\| \|\bm{\beta}\| (1 + \log(\frac{2n_{\!\scriptscriptstyle A}}{\eta}))}{L^2(1 -  O(\sqrt{\frac{d}{n_{\!\scriptscriptstyle A}}}+\sqrt{\frac{\log(1/\eta)}{n_{\!\scriptscriptstyle A}}}))^4 +  \lambda^2\|\hat{\Sigma}_{\!\scriptscriptstyle B}\|^{\scriptscriptstyle \!-\!2}}\Big),
\end{equation*} where $L = \frac{n_{\scriptscriptstyle \!B}}{(\sqrt{n_{\scriptscriptstyle \!B}} + O(\sqrt{d} + \sqrt{\log(\frac{1}{\eta})}))^2}$. Moreover, based on \textbf{Theorem \ref{thm:beta_transform}}, the output $\bm{\bar{\beta}}^{\scriptscriptstyle D\!P}$ satisfies, with at least probability $1 - O(\eta)$,
\begin{equation*}
	\|\bm{\bar{\beta}}^{\scriptscriptstyle D\!P} - \bm{\hat{\beta}}\| \leqslant O\Big(\frac{\sqrt{d^3 \log(\frac{2d}{\eta})}}{\mu n_{\! \scriptscriptstyle A}} \cdot \frac{\|\hat{\Sigma}_{\scriptscriptstyle \!B}^{\scriptscriptstyle -\!1}\| \|\Sigma\| \|\bm{\beta}\| (1 + \log(\frac{2n_{\!\scriptscriptstyle A}}{\eta}))}{L^2(1 -  O(\sqrt{\frac{d}{n_{\!\scriptscriptstyle A}}}+\sqrt{\frac{\log(1/\eta)}{n_{\!\scriptscriptstyle A}}}))^4 +  \lambda^2\|\hat{\Sigma}_{\!\scriptscriptstyle B}\|^{\scriptscriptstyle \!-\!2}}\Big).
\end{equation*}
	
\end{thm}

The following theorem will give the DP ridge regression of private data only. Without loss of generality, we give the result under the assumption that the private data are bounded by $\sqrt{\Tr(\Sigma) + d\log(\frac{n_{\scriptscriptstyle \!A}}{\eta})}$ and the private responses are bounded by $\sqrt{\Tr(\Sigma) + d\log(\frac{n_{\scriptscriptstyle \!A}}{\eta})}\|\bm{\beta}\| + o(1)$, with at least probability $1 - \eta$.

\begin{thm}[DP-RR]\label{thm:DP-RR} Assume that truncated data $\|\mathbf{A}_i\| \leqslant \sqrt{\Tr(\Sigma) + d\log(\frac{n_{\scriptscriptstyle \!A}}{\eta})}$ and $y_{\scriptscriptstyle \!A_i} \leqslant \sqrt{\Tr(\Sigma) + d\log(\frac{n_{\scriptscriptstyle \!A}}{\eta})}\|\bm{\beta}\|$ $w.p. \ 1 - \eta$, $\forall i\in[n_{\scriptscriptstyle \!A}]$. Denote $\hat{\Sigma}_{\scriptscriptstyle \!A} = \frac{\mathbf{A}^{\scriptscriptstyle \!T}\!\mathbf{A}}{n_{\scriptscriptstyle \!A}}$ and $\hat{\Sigma}_{\scriptscriptstyle \!Ay} = \frac{\mathbf{A}^{\scriptscriptstyle \!T}\!\mathbf{y}_{\scriptscriptstyle \!\!A}}{n_{\scriptscriptstyle \!A}}$. Then the $\sqrt{2}\mu$-GDP ridge regression is 
	\begin{equation*}
		\hat{\bm{\beta}}_{\scriptscriptstyle \! D\!P} = (\hat{\Sigma}_{\scriptscriptstyle \!A} + \lambda \mathbf{I} + \mathbf{G})^{\scriptscriptstyle \!-1} (\hat{\Sigma}_{\scriptscriptstyle \!Ay} + \mathbf{g}),
	\end{equation*}where $\mathbf{G} \sim SG_d(\sigma_1^2),\ \sigma_1 = \frac{2(\Tr(\Sigma) + d\log(\frac{n_{\scriptscriptstyle \!A}}{\eta}))}{\mu \cdot n_{\scriptscriptstyle \!A}}$ and $\mathbf{g} \sim \mathcal{N}(0,\sigmoid^2_2),\ \sigma_2 =  \frac{2(\Tr(\Sigma) + d\log(\frac{n_{\scriptscriptstyle \!A}}{\eta}))\|\bm{\beta}\|}{\mu \cdot n_{\scriptscriptstyle \!A}}$. 
	
	If $n_{\! A}$ makes $\frac{\sqrt{d^{\scriptscriptstyle 3}\log(\frac{2d}{\eta})} (d^{\scriptscriptstyle \!-\!1}\!\Tr(\Sigma) + \log(\frac{n_{\scriptscriptstyle \!\xi}}{\eta}))}{\mu \cdot n_{\! A} \big(\lambda_{\scriptscriptstyle min} (\Sigma)\big( 1 - O(\sqrt{\frac{d}{n_{\!\scriptscriptstyle {\! A}}}}+\sqrt{\frac{\log(1/\eta)}{n_{\!\scriptscriptstyle {\! A}}}})\big)^2 + \lambda \big)}\leqslant \frac{1}{2}$, with at least probability $1 - O(\eta)$, we have
	\begin{equation*}
		\|\hat{\bm{\beta}}_{\scriptscriptstyle \! D\!P}  - \hat{\bm{\beta}}\| \leqslant O\Big(\frac{\sqrt{d^3 \log(\frac{2d}{\eta})}}{\mu n_{\! \scriptscriptstyle \! A}} \cdot \frac{ \|\Sigma\| \|\bm{\beta}\|(d^{\scriptscriptstyle \!-\!1}\!\Tr(\Sigma) + \log(\frac{n_{\scriptscriptstyle \!\xi}}{\eta}))}{\lambda^2_{\scriptscriptstyle min} (\Sigma)\big(1 - O(\sqrt{\frac{d}{n_{\!\scriptscriptstyle \! A}}}+\sqrt{\frac{\log(1/\eta)}{n_{\!\scriptscriptstyle \! A}}})\big)^4 +  \lambda^2}\Big).
	\end{equation*}
\end{thm}

\begin{rem}
The improvements achieved by DP-PMTRR are analogous to those in \textbf{Theorem~\ref{thm:public_inverse_error}}, where the influence of the average condition number $\bar{\kappa}(\Sigma)$ and the matrix norm $\|\Sigma^{\scriptscriptstyle -1}\|$ is effectively removed. Moreover, DP-PMTRR significantly reduces the sensitivity of $\hat{\Sigma}{\scriptscriptstyle Ay}$ from $d \left(\bar{\kappa}(\Sigma) + \log(n_{\scriptscriptstyle A})\right)\|\bm{\beta}\|$ to $\sqrt{d} \log(n_{\scriptscriptstyle A})$. This eliminates the dependence on the unknown $\bm{\beta}$ and $\Sigma$. Detailed justifications are provided in the proofs.
\end{rem}

\begin{rem}
	In practice, when the information $\Sigma$ and $\bm{\beta}$ from \textbf{Theorem~\ref{thm:DP-RR}} are unknown and hence truncation is not doable, we can replace the $\Sigma$ with the empirical estimation $\hat{\Sigma}$. For the private responses $y_i$s, the bound can be replaced by the empirical second-moment estimation $\hat{\sigma}^2 = \frac{1}{n}\sum_{i=1}^{n} y_i^2$ and the truncation can be done by the radius $\sqrt{\hat{\sigma}^2 + \log(\frac{n}{\eta})}$. Then the DP parameter $\sigma_2 = \frac{\sqrt{(\Tr(\Sigma) + d\log(\frac{n_{\scriptscriptstyle \!A}}{\eta}))(\hat{\sigma}^2 + \log(\frac{n}{\eta}))}}{n}$. These make the DP-RR practical and able be implemented in real applications. In our experiment based on real datasets, we use this idea.
\end{rem}

Notice that, the proposed method can be applicable to gradient-based DP penalized regressions, such as Lasso regularization and elastic net, where the solutions require the truncation of DP gradients based on the norm of each data point. The proposed method can be used to support these DP gradient methods and provide a principled truncation.

\section{Application to Generalized Linear Models}\label{sec:logistic}
In this section, we firstly study the application to a classical generalized linear model, logistic regression, which is a common classifying model. Then, we extend to generalized linear models.
\subsection{Application to Logistic Regression}
We consider DP logistic regression based on the PMT and the transformed data. Logistic regression is a classical generalized linear model that needs to be solved by iterative optimization. We consider logistic regression model:
\begin{equation}\label{eq:sigmoid}
	\begin{aligned}
		y_i &\sim Bernoulli(p_i)\\
		p_i &= \frac{1}{1 + e^{-A_i^{\scriptscriptstyle T} \!\bm{\beta}}}, \quad i = 1,2,\cdots,n_{\scriptscriptstyle \!A}.
	\end{aligned}
\end{equation} Moreover, we write the responses $y_i$'s joint probability density function as the exponential distribution form 
\begin{equation*}
	\prod_i^{n} f(y_i;\theta_i) =\prod_i^{n} \exp\big\{y_i\theta_i - b(\theta_i)\},
\end{equation*} where $\theta_i = \log(\frac{p_i}{1 - p_i}) = A_i^\top \bm{\beta} $ and $b(\theta_i) = \log(1+\exp(\theta_i))$.

Denote $\mathbf{A}$ is a $n \times d$ data matrix and $\mathbf{A} = [A_1^T\ ...\  A_n^T]^T$, $A_i \sim subG(\Sigma) \in \mathbb{R}^{d \times 1}$ is the $i$-th sample. $\mathbf{y}_{\! \! \scriptscriptstyle A} \in \{0,1\}^{n_{\scriptscriptstyle \!A}}$ is the respond variable. For logistic regression, the negative log-likelihood loss function with the regularization is defined as
\begin{equation}\label{eq:trans_loss}
	\begin{aligned}
		\mathcal{L}(\bm{\beta};\mathbf{A}) &= -\frac{1}{n_{\scriptscriptstyle \!A}} \sum_{i=1}^{n_{\scriptscriptstyle \!A}}\Big[y_i \log(p_i) + (1-y_i) \log(1 - p_i)\Big] + \frac{\lambda}{2} ||\bm{\beta}||_2^2\\
		&= -\frac{1}{n_{\scriptscriptstyle \!A}}\sum_{i=1}^{n_{\scriptscriptstyle \!A}} \big[ y_i\sum_{j=1}^{d} \beta_j A_{ij} - \log\big(1 + \exp(\sum_{j=1}^{d} \beta_j A_{ij}) \big) \big] + \frac{\lambda}{2}\|\bm{\beta}\|_2^2\\
		&= -\frac{1}{n_{\scriptscriptstyle \!A}}\sum_{i=1}^{n_{\scriptscriptstyle \!A}} l_i(A_i^T \bm{\beta}) + \frac{\lambda}{2} \|\bm{\beta}\|_2^2.
	\end{aligned}
\end{equation}
where $l_i(A_i^T \bm{\beta})=y_i A_i^{\scriptscriptstyle T}\bm{\beta} - \log(1 + \exp(A_i^{\scriptscriptstyle T}\bm{\beta}))$ and $\lambda$ is the regularization parameter.

When the data are transformed, the loss function is
\begin{equation*}
	\bm{\beta'} = \arg\min_{\bm{\beta}}\mathcal{L}(\bm{\beta};\mathbf{\tilde{A}})= \arg\min_{\bm{\beta}}\Big\{-\frac{1}{n_{\scriptscriptstyle \!A}}\sum_{i=1}^{n_{\scriptscriptstyle \!A}} l_i(\tilde{A}_i^T \bm{\beta}) + \frac{\lambda}{2} \|\bm{\beta}\|_2^2\Big\}.
\end{equation*} That results in the different estimation, $\bm{\beta'} \neq \arg\min_{\bm{\beta}} \mathcal{L}(\bm{\beta};\mathbf{A})$. We give the following theorem and redefine a new loss function so as to get the same estimation as the original one.
\begin{thm}[Equivalent estimation]\label{thm:Equivalent estimation}
	Define the following loss function 
	\begin{equation}\label{eq:trans_fun}
		\tilde{\mathcal{L}}(\bm{\beta};\mathbf{\tilde{A}}) = -\frac{1}{n_{\scriptscriptstyle \!A}}\sum_{i=1}^{n_{\scriptscriptstyle \!A}} l_i(\tilde{A}_i^T \bm{\beta}) + \frac{\lambda}{2}\|\hat{\Sigma}_{\scriptscriptstyle \!B}^{\scriptscriptstyle -\!1\!/\!2}\bm{\beta}\|_2^2,
	\end{equation} where $\hat{\Sigma}_{\scriptscriptstyle \!B}$ is the second-moment estimation of the public data, and its minimizer $\tilde{\bm{\beta}} = \arg\min_{\bm{\beta}}\tilde{\mathcal{L}}(\bm{\beta};\mathbf{\tilde{A}})$. Then we have the equivalent optimization
	\begin{equation*}
		\hat{\Sigma}_{\scriptscriptstyle \!B}^{\scriptscriptstyle -\!1\!/\!2}\tilde{\bm{\beta}} \ = \ \hat{\bm{\beta}}:=\arg\min_{\bm{\beta}} \mathcal{L}(\bm{\beta};\mathbf{A}).
	\end{equation*} Especially, for Newton's method, we also have the equation at every $t$ iteration
	\begin{equation*}
		\hat{\Sigma}_{\scriptscriptstyle \!B}^{\scriptscriptstyle -\!1\!/\!2} \tilde{\bm{\beta}}^{\scriptscriptstyle (t)} = \hat{\bm{\beta}}^{\scriptscriptstyle (t)}.
	\end{equation*}
\end{thm}

We list the gradients and Hessian matrices about the two loss functions so as to discuss them in the following. The gradients are
\begin{equation*}
	\nabla \mathcal{L}(\bm{\beta};\mathbf{A}) = -\frac{1}{n_{\scriptscriptstyle \!A}} \mathbf{A}^{\scriptscriptstyle \!T} (\mathbf{y} - \mathbf{p}) + \lambda \bm{\beta},
\end{equation*} and 
\begin{equation*}
	\nabla \tilde{\mathcal{L}}(\bm{\beta};\mathbf{\tilde{A}}) = -\frac{1}{n_{\scriptscriptstyle \!A}} \tilde{\mathbf{A}}^{\scriptscriptstyle \!T} (\mathbf{y} - \tilde{\mathbf{p}}) + \lambda  \hat{\Sigma}_{\scriptscriptstyle \!B}^{\scriptscriptstyle -\!1} \bm{\beta}
\end{equation*}
where $\mathbf{p} = (p_1,...,p_{n_{\scriptscriptstyle \!A}})^\top \in [0,1]^{n_{\scriptscriptstyle \!A}}$ and $\tilde{\mathbf{p}} = (\tilde{p}_1,...,\tilde{p}_{n_{\scriptscriptstyle \!A}})^\top \in [0,1]^{n_{\scriptscriptstyle \!A}}$. Note that $p_i = \phi(A_i^\top \bm{\beta})$ and $\tilde{p}_i = \phi(\tilde{A}_i^\top \bm{\beta})$ are different due to different samples. The Hessian matrices are
\begin{equation*}
	\nabla^2 \mathcal{L}(\bm{\beta};\mathbf{A}) = \mathbf{H_{\scriptscriptstyle \!\beta}} + \lambda \mathbf{I} = \frac{1}{n_{\scriptscriptstyle \!A}}\mathbf{A^{\scriptscriptstyle \!\!T} W_{\scriptscriptstyle \!\!\!\beta} A}+ \lambda \mathbf{I},
\end{equation*} and 
\begin{equation}\label{eq:tilde_Hessian}
	\nabla^2 \tilde{\mathcal{L}}(\bm{\beta};\mathbf{\tilde{A}}) =  \mathbf{\tilde{H}_{\scriptscriptstyle \!\beta}} + \lambda  \hat{\Sigma}_{\scriptscriptstyle \!B}^{\scriptscriptstyle -\!1} = \frac{1}{n_{\scriptscriptstyle \!A}}\mathbf{\tilde{A}^{\scriptscriptstyle \!\!T} \tilde{W}_{\scriptscriptstyle \!\!\!\beta} \tilde{A}}+ \lambda \hat{\Sigma}_{\scriptscriptstyle \!B}^{\scriptscriptstyle -\!1},
\end{equation} where $ \mathbf{W_{\scriptscriptstyle \!\!\!\beta}}$ is a diagonal matrix with $ (\mathbf{W_{\scriptscriptstyle \!\!\!\beta}})_{ii} = p_i(1 - p_i)$, similarly, $ (\mathbf{\tilde{W}_{\scriptscriptstyle \!\!\!\beta}})_{ii} = \tilde{p}_i(1 - \tilde{p}_i)$, $i= 1,...,n_{\scriptscriptstyle \!A}$.
\subsubsection{Algorithm and Privacy}
We propose \textbf{Algorithm \ref{alg:DP-PMTLR}}. Typically, in the logistic regression, the responses $y_i \in \{0,1\}$ are bounded naturally and we don't transform these. 

\begin{algorithm}[H]
	\caption{Differentially Private PMT Logistic Regression (DP-PMTLR)}\label{alg:DP-PMTLR}
	\begin{algorithmic}[1]
		\STATE {\bfseries Input:} Private dataset $\{(\mathbf{A}_i,y_i) \in \mathbb{R}^{d}\times\{0,1\} \}^{n_{\!\scriptscriptstyle A}}_{i=1}$, public dataset $\{\mathbf{B}_i\in \mathbb{R}^{d}\}^{n_{\!\scriptscriptstyle B}}_{i=1}$. Parameters $\mu$, $\lambda$, $d$, $n_{\!\scriptscriptstyle A}$, $n_{\!\scriptscriptstyle B}$ and $\eta$.
		
		\STATE {\bfseries Transform covariates:} 
    $(\{\tilde{\mathbf{A}}_i\}^{n_{\!\scriptscriptstyle A}}_{i=1},\hat{\Sigma}_{\!\scriptscriptstyle B}) =PMT(\{\mathbf{A}_i \}^{n_{\!\scriptscriptstyle A}}_{i=1},\{\mathbf{B}_i \}^{n_{{\!\scriptscriptstyle B}}}_{i=1},  d, n_{\!\scriptscriptstyle A}, n_{\!\scriptscriptstyle B},\eta)$.
		\STATE {\bfseries Private parameter:} 
			\begin{center}
				$\sigma_1 = \frac{\sqrt{T}d(1 + \log(\frac{2n_{\!\scriptscriptstyle A}}{\eta}))}{2 \mu n_{\scriptscriptstyle A}}$, 
				$\sigma_2 =\frac{2\sqrt{Td(1 + \log(\frac{2n_{\!\scriptscriptstyle A}}{\eta}))}}{\mu n_{\scriptscriptstyle \!A}}.$
			\end{center}
		\STATE $\bm{\beta}^{\scriptscriptstyle (0)} = \mathbf{0} $
			\FOR{$t = 0,...,T $ }{
				\STATE {\bfseries Gaussian mechanism and Newton update:} 
				\begin{center}
				$\bm{{\beta}}^{\scriptscriptstyle (t+1)} =  \bm{{\beta}}^{\scriptscriptstyle (t)} - \Big(\nabla^2 \tilde{\mathcal{L}}(\bm{\beta}^{\scriptscriptstyle (t)};\mathbf{\tilde{A}}) + \mathbf{G}\Big)^{\!\!\scriptscriptstyle -\!1}\Big(\nabla \tilde{\mathcal{L}}(\bm{\beta}^{\scriptscriptstyle (t)};\mathbf{\tilde{A}}) + \mathbf{g}\Big)$, 
				\end{center}
				where $\mathbf{G} \sim SG_d(\sigma_1^2)$ and $\mathbf{g} \sim \mathcal{N}(0,\sigma_2^2\mathbf{I})$.}
			\ENDFOR

		\STATE {\bfseries Recover:} $\bm{\bar{\beta}}^{\scriptscriptstyle D\!P} \leftarrow \hat{\Sigma}_{\!\scriptscriptstyle B}^{\scriptscriptstyle -\!1\!/\!2}\cdot\bm{\beta}^{\scriptscriptstyle (\!T\!)}$.
		\STATE {\bfseries Output:}  DP estimator $\bm{\bar{\beta}}^{\scriptscriptstyle D\!P}$.
	\end{algorithmic}
\end{algorithm}

\begin{thm}[Privacy]\label{thm:DPLR_pri}
	\textbf{Algorithm \ref{alg:DP-PMTLR}} satisfies $\sqrt{2}\mu$-GDP.
\end{thm}

\subsubsection{Theoretical Analysis}\label{subsec:theory}

The assumption about the Hessian matrix is necessary to ensure the convergence of logistic regression and show the advantage of our method. It's well known the Hessian matrix of $-\frac{1}{n_{\scriptscriptstyle \!A}}\sum_{i=1}^{n_{\scriptscriptstyle \!A}} l_i(A_i^T \bm{\beta})$ is
\begin{equation*}
	\mathbf{H}_{\scriptscriptstyle \!\beta} = \frac{\partial^2(-\frac{1}{n_{\scriptscriptstyle \!\!A}}\sum_{i=1}^{n_{\scriptscriptstyle \!A}} l_i(A_i^T \bm{\beta}))}{\partial \bm{\beta}^2} = \frac{1}{n_{\scriptscriptstyle A}}\mathbf{A^{\scriptscriptstyle \!\!T} W_{\scriptscriptstyle \!\!\!\beta} A},
\end{equation*} seeing Eq.\eqref{eq:tilde_Hessian}. Then, we assume that

\begin{equation}\label{eq:Hessian_assumption}
	\text{\textbf{Hessian assumption:}}\quad\tau_{\scriptscriptstyle 0} \lambda_{\scriptscriptstyle \min} (\hat{\Sigma}_{\scriptscriptstyle \!A}) \preccurlyeq \tau_{\scriptscriptstyle \!\beta} \lambda_{\scriptscriptstyle \min} (\hat{\Sigma}_{\scriptscriptstyle \!A}) \preccurlyeq \mathbf{H_{\scriptscriptstyle \!\beta}} \preccurlyeq \frac{1}{4} \lambda_{\scriptscriptstyle \max} (\hat{\Sigma}_{\scriptscriptstyle \!A}),
\end{equation}where $\tau_{\scriptscriptstyle \!\beta} \in (0,1/4]$ depends on the classified probabilities $p_i$s which are the function of $\bm{\beta}$ and $\tau_0 = \inf_\beta \tau_{\scriptscriptstyle \!\beta} $. The lower bound $\tau_{\scriptscriptstyle \!\beta} $ tends to be small easily at the latter stages of iteration, because the classified accuracy becomes high and leads to $p_i(1 - p_i) \to 0$. The condition number of the Hessian matrix $ \mathbf{H_{\scriptscriptstyle \!\beta}}$ is approximately
\begin{equation*}
	\kappa({\mathbf{H_{\scriptscriptstyle \!\beta}}}) = O(\frac{\lambda_{\scriptscriptstyle \max} (\hat{\Sigma}_{\scriptscriptstyle \!A})}{\tau_{\scriptscriptstyle 0} \lambda_{\scriptscriptstyle \min} (\hat{\Sigma}_{\scriptscriptstyle \!A})}) \overset{n_A\text{ large}}{=}  O(\frac{\kappa(\Sigma)}{\tau_{\scriptscriptstyle 0}}).
\end{equation*} This suggests that the weight matrix $\mathbf{W}_{\scriptscriptstyle \!\!\beta}$ leads to a more ill-conditioned Hessian, thereby making Newton's method more vulnerable to disruption from DP noise. Introducing a large regularization parameter $\lambda$ can alleviate this issue by decreasing the condition number, thus enhancing resistance to noise. However, this comes at the cost of increased bias in the estimation, which may lead to an underfitted model. Intuitively, the proposed transformation substantially reduces the condition number of the Hessian and also enhances the following favorable properties of the associated loss function.

\begin{lem}[Strong convexity]\label{lem:strongconvex}
	Assume that the $subG(\Sigma)$ data $\mathbf{A}\in \mathbb{R}^{n_{\scriptscriptstyle \!A} \times d}$, $\mathbf{y} = \{y_1,...,y_{n_{\scriptscriptstyle \!A}}\},\ y_i \in \{0,1\}$ and $\hat{\Sigma}_{\scriptscriptstyle \!B}$ is the second-moment estimation from other $n_{\scriptscriptstyle \!B}$ samples. Considering the convexity of two loss functions $\mathcal{L}(\bm{\beta};\mathbf{A})$ and $\tilde{\mathcal{L}}(\bm{\beta};\mathbf{\tilde{A}})$, with at least probability $1 - O(\eta)$, we have
	\begin{equation*}
		\nabla^2 \mathcal{L}(\bm{\beta};\mathbf{A}) \succcurlyeq \gamma_{\scriptscriptstyle \Sigma}\mathbf{I},
	\end{equation*} where $\gamma_{\scriptscriptstyle \Sigma} = \tau_{\scriptscriptstyle 0} \lambda_{\scriptscriptstyle \min}(\Sigma) \big(1 - O\big(\sqrt{\frac{d}{n}} + \sqrt{\frac{\log(1/\eta)}{n}}\big)\big)^2 + \lambda$. And, with at least probability $1 - 2\eta$, we have
	\begin{equation*}
		\nabla^2 \tilde{\mathcal{L}}(\bm{\beta};\mathbf{\tilde{A}}) \succcurlyeq \gamma_{\scriptscriptstyle \!L}\mathbf{I},
	\end{equation*}where $\gamma_{\scriptscriptstyle \!L} = \tau_{\scriptscriptstyle 0}L\big(1 - O\big(\sqrt{\frac{d}{n_{\scriptscriptstyle \!A}}} + \sqrt{\frac{\log(1/\eta)}{n_{\scriptscriptstyle \!A}}}\big)\big)^2 + \frac{\lambda}{\lambda_{\scriptscriptstyle \max}(\hat{\Sigma}_{\scriptscriptstyle \!B})}$ and $L = \frac{n_{\scriptscriptstyle \!B}}{(\sqrt{n_{\scriptscriptstyle \!B}} + O(\sqrt{d} + \sqrt{\log(\frac{1}{\eta})}))^2}$.
\end{lem}
\begin{rem}
	That shows the better convexity of the transformation loss $\tilde{\mathcal{L}}(\bm{\beta};\mathbf{\tilde{A}})$. Because $L \to 1$ is greater than $\lambda_{\scriptscriptstyle \min}(\Sigma)$, and the regularization parameter $\frac{\lambda}{\lambda_{\scriptscriptstyle \max}(\hat{\Sigma}_{\scriptscriptstyle \!B})}$ weaken the dependence on the regularization parameter.
\end{rem}

\begin{lem}[Lipschitz continuity of Hessian]\label{lem:Lipschitz}
The Hessian matrix $\mathbf{H_{\scriptscriptstyle \!\beta}}$ is Lipschitz continuous with respect to the $\ell_2$ norm, i.e., there exists a constant $C_{\scriptscriptstyle \!\!A} > 0$ such that
\begin{equation*}
	\|\mathbf{H_{\scriptscriptstyle \!\beta}} - \mathbf{H_{\scriptscriptstyle \!\beta'}}\|_2 \leqslant C_{\scriptscriptstyle \!\!A} \|\bm{\beta} - \bm{\beta'}\|_2,
\end{equation*} for all $\bm{\beta},\bm{\beta'} \in \mathbb{R}^d$. In particular, we have $C_{\scriptscriptstyle \!\!A}=\frac{\sup_{i}\|A_i\|^3}{6\sqrt{3}}$. 
Moreover, 
\begin{equation*}
	\|\nabla^2 \mathcal{L}(\bm{\beta};\mathbf{A}) - \nabla^2 \mathcal{L}(\bm{\beta'};\mathbf{A})\|_2 \leqslant C_{\scriptscriptstyle \!\!A}\|\bm{\beta} - \bm{\beta'}\|_2,
\end{equation*}and
\begin{equation*}
	\|\nabla^2 \tilde{\mathcal{L}}(\bm{\beta};\mathbf{\tilde{A}}) - \nabla^2 \tilde{\mathcal{L}}(\bm{\beta'};\mathbf{\tilde{A}})\|_2 \leqslant C_{\scriptscriptstyle \!\!\tilde{A}}\|\bm{\beta} - \bm{\beta'}\|_2,
\end{equation*} where $C_{\scriptscriptstyle \!\!\tilde{A}} = \frac{(d(1 + \log(n_{\scriptscriptstyle \!A}/\eta)))^{\scriptscriptstyle 3\!/\!2}}{6\sqrt{3}}$ and $\eta > 0$.
\end{lem}
\begin{rem}
	\textbf{Lemma \ref{lem:Lipschitz}} shows that the $l_2$-norm of sample points also impacts on the Lipschitz continuity of the Hessian matrix. The PMT method normalizes the the $l_2$-norm of sample points into a principled radius and reduces the Lipschitz continuity parameter, which improves the converged efficiency of Newton's method. 
\end{rem}

These lemmas guarantee the convergence of DP-PMTLR and DP-LR.

\begin{thm}[DP-PMTLR]\label{thm:DP-PMTLR}
	Suppose $\mathbf{A}_i \in subG(\Sigma)$, the minimizer $\tilde{\bm{\beta}}$, $\bm{\beta}^{(0)} \in \mathcal{B}_{\scriptscriptstyle r}(\tilde{\bm{\beta}})$ and $\|\nabla \tilde{\mathcal{L}}(\bm{\beta}^{(0)};\mathbf{\tilde{A}})\| \leqslant \min\big\{ \gamma_{\scriptscriptstyle \!L}r, \frac{ \gamma_{\scriptscriptstyle \!L}^{\scriptscriptstyle 2}}{C_{\scriptscriptstyle \!\!\tilde{A}}}\big\}$, where $\gamma_L = \tau_{\scriptscriptstyle 0}L(1 - O(\sqrt{\frac{d}{n_{\!\scriptscriptstyle \!A}}}+\sqrt{\frac{\log(1/\eta)}{n_{\!\scriptscriptstyle \!A}}}))^2 + \lambda\|\hat{\Sigma}_{\!\scriptscriptstyle \!B}\|^{\scriptscriptstyle \!-\!1}$. Let  $\sqrt{n_{\scriptscriptstyle \!B}} \geqslant O(\sqrt{d} + \sqrt{\log(1/\eta)})$, $n_{\scriptscriptstyle \!A}$ makes $\frac{\sqrt{Td^3 \log(\frac{2Td}{\eta})}(1 + \log(\frac{2n_{\!\scriptscriptstyle \!A}}{\eta}))}{\mu \cdot n_{\! \scriptscriptstyle \!A}\cdot \gamma_L}$ small enough and $T = O(\log\log(n_{\scriptscriptstyle \!A}))$. The $T^{\scriptscriptstyle th}$ DP Newton's method iteration satisfies $\|\bm{\beta}^{\scriptscriptstyle (T)} - \tilde{\bm{\beta}}\| \leqslant O\Big(\frac{\sqrt{Td^3 \log(\frac{2Td}{\eta})}(1 + \log(\frac{2n_{\!\scriptscriptstyle \!A}}{\eta}))}{\mu \cdot n_{\! \scriptscriptstyle \!A}\cdot\gamma_L^{\scriptscriptstyle 2}}\Big)$, w.p. $1 - \eta$.
\end{thm}
\begin{rem}
	Due to the iteration invariant in \textbf{Theorem \ref{thm:Equivalent estimation}}, we also conclude the convergence of the DP estimator $\bm{\bar{\beta}}^{\scriptscriptstyle D\!P}$.
\end{rem}
\begin{thm}[DP-LR]\label{thm:DP-LR}
	Suppose every sub-Gaussian sample $\mathbf{A}_i$ is truncated as $\|\mathbf{A}_i\| \leqslant \sqrt{\Tr(\Sigma) + d\log(\frac{n_{\scriptscriptstyle \!A}}{\eta})}$, the minimizer $\hat{\bm{\beta}}$, $\hat{\bm{\beta}} \in \mathcal{B}_{\scriptscriptstyle r}(\hat{\bm{\beta}})$ and $\|\nabla \mathcal{L}(\bm{\beta}^{(0)};\mathbf{A})\| \leqslant \min\big\{\gamma_{\scriptscriptstyle \Sigma}r, \frac{\gamma_{\scriptscriptstyle \Sigma}^{\scriptscriptstyle 2}}{C_{\scriptscriptstyle \!A}}\big\}$, where $\gamma_{\scriptscriptstyle \Sigma} = \tau_{\scriptscriptstyle 0} \lambda_{\scriptscriptstyle \min}(\Sigma) \big(1 - O\big(\sqrt{\frac{d}{n}} + \sqrt{\frac{\log(1/\eta)}{n}}\big)\big)^2 + \lambda$. Let $n_{\scriptscriptstyle \!A}$ makes $\Big(\frac{\sqrt{Td^{\scriptscriptstyle 3}\log(\frac{2Td}{\eta})} (d^{\scriptscriptstyle \!-\!1}\!\Tr(\Sigma) + \log(\frac{n_{\scriptscriptstyle \!A}}{\eta}))}{\mu \cdot n_{\scriptscriptstyle \!A}\cdot \gamma_{\scriptscriptstyle \Sigma}}\Big)$ sufficiently small and $T = O(\log\log(n_{\scriptscriptstyle \!A}))$. In \textbf{Algorithm \ref{alg:DP-LR}}, the $T^{\scriptscriptstyle th}$ DP Newton's method iteration satisfies (i) $\bm{\beta}^{\scriptscriptstyle (T)}$ is $\sqrt{2}\mu$-GDP, and (ii) $\|\bm{\beta}^{\scriptscriptstyle (T)} - \hat{\bm{\beta}}\| \leqslant O\Big(\frac{\sqrt{Td^{\scriptscriptstyle 3}\log(\frac{2Td}{\eta})} (d^{\scriptscriptstyle \!-\!1}\!\Tr(\Sigma) + \log(\frac{n_{\scriptscriptstyle \!\xi}}{\eta})) }{\mu \cdot n_{\! A} \gamma_{\scriptscriptstyle \Sigma}^{\scriptscriptstyle \!2}}\Big)$, w.p. $1 - \eta$.
\end{thm}

\subsection{Generalized Linear Model}
In this subsection, the proposed method is applied to generalized linear models. Given the data $(A_i, y_i)\in \mathbb{R}^{d+1} ,\ i=1,...,n_{\scriptscriptstyle \!A}$ and $y_i$ following the exponential family of distribution, the joint probability density function is 
\begin{equation*}
	\prod_i^{n} f(y_i;\theta_i,\phi) =\prod_i^{n} \exp\big\{\frac{y_i\theta_i - b(\theta_i)}{r(\phi)} + c(y_i,\phi)\big\},
\end{equation*}where $\theta_i$ is the natural parameter and $\phi$ is the dispersion parameter.

Given the bounded model parameters $\bm{\beta}$ and  $\|\bm{\beta}\|_2\leq R_\beta$, the generalized linear model (GLM) states that there exists a link function $g(\cdot)$ which satisfies
\begin{equation}\label{eq:varpi}
	 g(\mathbb{E}(y_i)) = g(\varpi_i) = A_i^\top \bm{\beta}\ \ i=1,...,n. 
\end{equation} 
Here $\varpi_i = \mathbb{E}(y_i)$ and we consider the canonical link where $\theta_i = g(\varpi_i) = A_i^\top \bm{\beta}$. Notice that the exponential family of distribution makes
\begin{equation*}
		\varpi_i = \mathbb{E}(y_i) = b'(\theta_i),
\end{equation*}where $b'(\theta_i) = \frac{\partial b(\theta_i)}{\partial \theta_i}$ and it implies that $\theta_i = (b')^{\scriptscriptstyle -1}(\varpi_i)$ and the canonical link function $g(z) = (b')^{\scriptscriptstyle -1}(z)$. From Eq.\eqref{eq:varpi}, we have $\varpi_i = g^{\scriptscriptstyle -1}(A_i^\top \bm{\beta}) = b'(A_i^\top \bm{\beta})$ and $\theta_i = (b')^{\scriptscriptstyle -1}(A^\top_i \bm{\beta})$. Considering the negative log-likelihood function with the regularization term, we have
\begin{equation}\label{eq:GLM_loss}
	\mathcal{L}(\bm{\beta};\mathbf{A}) = -\frac{1}{n_{\scriptscriptstyle A}} \sum_{i=1}^{n_{\scriptscriptstyle A}} \Big[y_i\theta_i - b(\theta_i)\Big] + \frac{\lambda}{2}\|\bm{\beta}\|^2_2,
\end{equation} where $\theta_i = A^\top_i \bm{\beta}$ and we omit the scale parameter $r(\phi)$ unrelated to the optimization of $\bm{\beta}$. Denote the $l_i(A_i^\top \bm{\beta}) = y_i\theta_i - b(\theta_i)$, we get the $l_i$'s gradient and Hessian matrix as 
\begin{equation*}
	\nabla_{\!\! \beta} l_i(A_i^\top \bm{\beta}) = (y_i - b'(A_i^\top \bm{\beta})) A_i;\ \ \ \	\nabla^2_{\!\! \beta} l_i(A_i^\top \bm{\beta}) = -b''(A_i^\top \bm{\beta}) A_i A_i^\top.
\end{equation*} Then, the Hessian of $\mathcal{L}(\bm{\beta};\mathbf{A})$ is 
\begin{equation*}
	\nabla_{\!\! \beta}^2\mathcal{L}(\bm{\beta};\mathbf{A}) = -\frac{1}{n_{\scriptscriptstyle A}}\nabla^2_{\!\! \beta} l_i(A_i^\top \bm{\beta}) + \lambda \mathbf{I}  = \frac{1}{n_{\scriptscriptstyle A}}\mathbf{A}^\top \mathbf{W}_{\!\! \beta} \mathbf{A}  + \lambda \mathbf{I}
\end{equation*} with the weights $\mathbf{W}_{\!\! \beta} = diag(b''(A_1^\top \bm{\beta}), \cdots, b''(A_n^\top \bm{\beta}))$. That implies, in the generalized linear model, the Hessian matrix $\nabla_{\!\! \beta}\mathcal{L}(\bm{\beta};\mathbf{A})$ is also affected by the second-moment matrix $\hat{\Sigma}_A$ like the analysis of Subsection \ref{subsec:theory}. Namely, 
\begin{equation*}
	\tau_{\scriptscriptstyle 0} \lambda_{\scriptscriptstyle \min} (\hat{\Sigma}_{\scriptscriptstyle \!A}) \preccurlyeq \mathbf{H_{\scriptscriptstyle \!\beta}} = \frac{1}{n_{\scriptscriptstyle A}}\mathbf{A}^\top \mathbf{W}_{\!\! \beta} \mathbf{A}\preccurlyeq \tau_{\scriptscriptstyle 1} \lambda_{\scriptscriptstyle \max} (\hat{\Sigma}_{\scriptscriptstyle \!A}),
\end{equation*}where $\tau_{\scriptscriptstyle 0} = \inf\limits_z b''(z) >0$ and $\tau_{\scriptscriptstyle 1} = \sup\limits_z b''(z) < \infty$. The condition number of $\mathbf{H_{\scriptscriptstyle \!\beta}}$ is 
\begin{equation*}
	\kappa({\mathbf{H_{\scriptscriptstyle \!\beta}}}) = O(\frac{\tau_{\scriptscriptstyle 1}\lambda_{\scriptscriptstyle \max} (\hat{\Sigma}_{\scriptscriptstyle \!A})}{\tau_{\scriptscriptstyle 0} \lambda_{\scriptscriptstyle \min} (\hat{\Sigma}_{\scriptscriptstyle \!A})}) \overset{n_A\text{ large}}{=}  \frac{\tau_{\scriptscriptstyle 1}}{\tau_{\scriptscriptstyle 0}}\kappa(\Sigma).
\end{equation*} The rate $\frac{\tau_{\scriptscriptstyle 1}}{\tau_{\scriptscriptstyle 0}}$ is fixed by the function $b(\cdot)$, so reducing the condition number $\kappa(\Sigma)$ is the only choice. Next, we illustrate that our method PMT can improve the GLM case via eliminating the $\kappa(\Sigma)$.

\subsubsection{Public-moment-guided Generalized Linear Models}

Similar to the analysis in Eq.\eqref{eq:trans_loss}, we consider the transformed loss function	
\begin{equation}\label{eq:GLM_trans_loss}
	\tilde{\mathcal{L}}(\bm{\beta};\mathbf{\tilde{A}}) = -\frac{1}{n_{\scriptscriptstyle A}} \sum_{i=1}^{n_{\scriptscriptstyle A}} \Big[\frac{y_i\tilde{\theta}_i - b(\tilde{\theta}_i)}{r(\phi)}\Big] + \frac{\lambda}{2}\|\hat{\Sigma}_{\scriptscriptstyle \!B}^{\scriptscriptstyle -1\!/\!2}\bm{\beta}\|^2_2,
\end{equation} where $\tilde{\theta}_i = \tilde{A}^\top_i \bm{\beta}$ and $\tilde{A}_i = \hat{\Sigma}_{\scriptscriptstyle \!B}^{\scriptscriptstyle -1\!/\!2}A_i$ is the transformed data point from PMT. The gradient and Hessian matrix of $\tilde{\mathcal{L}}(\bm{\beta};\mathbf{\tilde{A}})$ are as follows
\begin{equation}\label{eq:trans_GLM}
	\begin{aligned}
		&\nabla_{\!\! \beta}\tilde{\mathcal{L}}(\bm{\beta};\tilde{\mathbf{A}}) = \frac{1}{n_{\scriptscriptstyle \!A}} \tilde{\mathbf{A}}^\top(\mathbf{y} - \tilde{\mathbf{b'}}) +  \hat{\Sigma}_{\scriptscriptstyle \!B}^{\scriptscriptstyle -\!1} \bm{\beta};\\
		&\nabla^2_{\!\! \beta}\tilde{\mathcal{L}}(\bm{\beta};\tilde{\mathbf{A}}) = \frac{1}{n_{\scriptscriptstyle \!A}}\tilde{\mathbf{A}}^\top \tilde{\mathbf{W}}_{\!\! \beta} \tilde{\mathbf{A}}  + \lambda \hat{\Sigma}_{\scriptscriptstyle \!B}^{\scriptscriptstyle -\!1},
     	\end{aligned}
\end{equation}where $\mathbf{y} = (y_1,\cdots,y_{n_{\!A}})^\top$, $\tilde{\mathbf{b'}} = (b'(\tilde{A}_1^\top \bm{\beta}), \cdots, b'(\tilde{A}_{n_{\! A}}^\top \bm{\beta}))^\top$, and $\tilde{\mathbf{W}}_{\!\! \beta} = diag(b''(\tilde{A}_1^\top \bm{\beta})), \cdots, b''(\tilde{A}_{n_{\! A}}^\top \bm{\beta}))$. 

Considering the convexity of GLM loss functions, we have the invariant iterations as the following corollary, which is similar to \textbf{Theorem \ref{thm:Equivalent estimation}}.
\begin{cor}[Equivalent GLM estimation]
	Considering GLM and loss functions \eqref{eq:GLM_loss} and \eqref{eq:GLM_trans_loss}, their estimations is equivalent in each iteration of Newton's method. Moreover, their minimizers are equivalent
	\begin{equation*}
		\hat{\Sigma}_{\scriptscriptstyle \!B}^{\scriptscriptstyle -\!1\!/\!2}\tilde{\bm{\beta}} \ = \ \hat{\bm{\beta}}:=\arg\min_{\bm{\beta}} \mathcal{L}(\bm{\beta};\mathbf{A}),
	\end{equation*} where $\tilde{\bm{\beta}} = \arg\min_{\bm{\beta}}\tilde{\mathcal{L}}(\bm{\beta};\mathbf{\tilde{A}})$.
\end{cor}Under the invariant iterations case, the transformed loss function holds Hessian matrices with a smaller condition number 
\begin{equation*}
	\kappa({\mathbf{\tilde{H}_{\scriptscriptstyle \!\beta}}}) =  O(\frac{\tau_{\scriptscriptstyle 1}}{\tau_{\scriptscriptstyle 0}}\kappa(\tilde{\Sigma}_A)) \overset{n_A\text{ large}}{=} O(\frac{\tau_{\scriptscriptstyle 1}}{\tau_{\scriptscriptstyle 0}})
\end{equation*}where $\mathbf{\tilde{H}_{\scriptscriptstyle \!\beta}}  = \frac{1}{n_{\scriptscriptstyle A}}\tilde{\mathbf{A}}^\top \mathbf{W}_{\!\! \beta} \tilde{\mathbf{A}}$ and $\tilde{\Sigma}_A = \frac{1}{n_{\scriptscriptstyle \!A}}\tilde{\mathbf{A}}^\top \tilde{\mathbf{A}}$, and the second equation follows from \textbf{Theorem \ref{thm:secondmoment_bound}}. That illustrates that our method is also able to improve the GLM estimations in DP Newton's methods. We propose \textbf{Algorithm \ref{alg:DP-PMTGLM}} that illustrates how PMT is used in generalized linear models.
\begin{algorithm}[H]
	\caption{Differentially Private PMT Generalized Linear Models(DP-PMTGLM)}\label{alg:DP-PMTGLM}
	\begin{algorithmic}[1]
		\STATE {\bfseries Input:} Private dataset $\{(\mathbf{A}_i,y_i) \in \mathbb{R}^{d}\times(-R_y,R_y) \}^{n_{\!\scriptscriptstyle A}}_{i=1}$, public dataset $\{\mathbf{B}_i\in \mathbb{R}^{d}\}^{n_{\!\scriptscriptstyle B}}_{i=1}$. Parameters $\mu$, $\lambda$, $d$, $n_{\!\scriptscriptstyle A}$, $n_{\!\scriptscriptstyle B}$ and $\eta$. Constraint $\|\bm{\beta}\|_2 \leq R_\beta$, $M_{b'} = max_{z} b'(z)$ and $M_{b''} = max_{z} b''(z)$, $z\in [-R_\beta\sqrt{d(1 + \log(\frac{2n_{\!\scriptscriptstyle A}}{\eta}))},R_\beta\sqrt{d(1 + \log(\frac{2n_{\!\scriptscriptstyle A}}{\eta}))}]$.
		
		\STATE {\bfseries Transform covariates:} 
    $(\{\tilde{\mathbf{A}}_i\}^{n_{\!\scriptscriptstyle A}}_{i=1},\hat{\Sigma}_{\!\scriptscriptstyle B}) =PMT(\{\mathbf{A}_i \}^{n_{\!\scriptscriptstyle A}}_{i=1},\{\mathbf{B}_i \}^{n_{{\!\scriptscriptstyle B}}}_{i=1},  d, n_{\!\scriptscriptstyle A}, n_{\!\scriptscriptstyle B},\eta)$. 
		\STATE {\bfseries Private parameter:} 
			\begin{center}
				$\sigma_1 =\frac{2\sqrt{T}d(1 + \log(\frac{2n_{\!\scriptscriptstyle A}}{\eta}))M_{b''}}{\mu n_{\scriptscriptstyle \!A}}$,
				$\sigma_2 = \frac{2\sqrt{Td(1 + \log(\frac{2n_{\!\scriptscriptstyle A}}{\eta}))}(R_y + M_{b'})}{\mu n_{\scriptscriptstyle \!A}}$.
			\end{center}
		\STATE $\bm{\beta}^{\scriptscriptstyle (0)} = \mathbf{0} $
			\FOR{$t = 0,...,T $ }{
				\STATE {\bfseries Gaussian mechanism and Newton update:} 
				\begin{center}
				$\bm{\tilde{\beta}}^{\scriptscriptstyle (t+1)} =  \bm{\tilde{\beta}}^{\scriptscriptstyle (t)}- \Big(\nabla^2 \tilde{\mathcal{L}}(\tilde{\bm{\beta}}^{\scriptscriptstyle (t)};\mathbf{\tilde{A}}) + \mathbf{G}\Big)^{\!\!\scriptscriptstyle -\!1}\Big(\nabla \tilde{\mathcal{L}}(\tilde{\bm{\beta}}^{\scriptscriptstyle (t)};\mathbf{\tilde{A}}) + \mathbf{g}\Big)$, 
				\end{center}
				where $\mathbf{G} \sim SG_d(\sigma_1^2)$ and $\mathbf{g} \sim \mathcal{N}(0,\sigma_2^2\mathbf{I})$.}
			\ENDFOR

		\STATE {\bfseries Recover:} $\bm{\bar{\beta}}^{\scriptscriptstyle D\!P} \leftarrow \hat{\Sigma}_{\!\scriptscriptstyle B}^{\scriptscriptstyle -\!1\!/\!2}\cdot\bm{\tilde{\beta}}^{\scriptscriptstyle (\!T\!)}$.
		\STATE {\bfseries Output:}  DP estimator $\bm{\bar{\beta}}^{\scriptscriptstyle D\!P}$.
	\end{algorithmic}
\end{algorithm}

\begin{thm}[DP-PMTGLM]\label{thm:DP-PMTGLM}
	Suppose $y_i \in [-R_y,R_y]$ with a constant $R_y$, $\mathbf{A}_i \in subG(\Sigma)$ and $b(\theta)$ is locally $\tau_0$-strong convexity. Constraint $\forall \ \|\bm{\beta}\|_2 \leq R_\beta$, the minimizer $\|\tilde{\bm{\beta}}\|_2 \leq R_\beta$, $\bm{\beta}^{(0)} \in \mathcal{B}_{\scriptscriptstyle r}(\tilde{\bm{\beta}})$, the Hessian matrix $\nabla^2 \tilde{\mathcal{L}}(\bm{\beta};\mathbf{\tilde{A}})$ is locally $C_l$-Lipschitz continuous when $\bm{\beta} \in \mathcal{B}_{\scriptscriptstyle r}(\tilde{\bm{\beta}})$, and $\|\nabla \tilde{\mathcal{L}}(\bm{\beta}^{(0)};\mathbf{\tilde{A}})\| \leqslant \min\big\{ \gamma_{\scriptscriptstyle \!L}r, \frac{ \gamma_{\scriptscriptstyle \!L}^{\scriptscriptstyle 2}}{C_l}\big\}$, where $\gamma_L = \tau_{\scriptscriptstyle 0}L(1 - O(\sqrt{\frac{d}{n_{\!\scriptscriptstyle \!A}}}+\sqrt{\frac{\log(1/\eta)}{n_{\!\scriptscriptstyle \!A}}}))^2 + \lambda\|\hat{\Sigma}_{\!\scriptscriptstyle \!B}\|^{\scriptscriptstyle \!-\!1}$. Let $\sqrt{n_{\scriptscriptstyle \!B}} \geqslant O(\sqrt{d} + \sqrt{\log(1/\eta)})$, $n_{\scriptscriptstyle \!A}$ makes $\Big(\frac{\sqrt{Td^3\log(Td/\eta)}(1 + \log(\frac{2n_{\!\scriptscriptstyle A}}{\eta}))M_{b''}}{\mu \cdot n_{\! \scriptscriptstyle \!A}\cdot\gamma_L}\Big)$ small enough and $T = O(\log\log(n_{\scriptscriptstyle \!A}))$. The outout of \textbf{Algorithm \ref{alg:DP-PMTGLM}} satisfies $\sqrt{2}\mu$-GDP, and we have $\|\bm{\beta}^{\scriptscriptstyle (T)} - \tilde{\bm{\beta}}\| \leqslant O\Big(\frac{\sqrt{Td^3 \log(\frac{2Td}{\eta})}(1 + \log(\frac{2n_{\!\scriptscriptstyle \!A}}{\eta}))}{\mu \cdot n_{\! \scriptscriptstyle \!A}\cdot\gamma_L^{\scriptscriptstyle 2}}\Big)$, w.p. $1 - \eta$.
\end{thm}

\section{Experiments}\label{sec:experiments}
In this section, we set up separate subsections for DP ridge regression and logistic regression and design experiments to evaluate our approach. Each regression is evaluated in a simulation and two real-world datasets.

\subsection{DP Ridge Regression}\label{subsec:DPRR} 
\subsubsection{Simulations} 
In the simulation, we generate data from a linear model using a feature matrix $\mathbf{X}$ with the dimension $d = 10$ and its distribution is $\mathcal{N}(\mu,\Psi)$, the noise $\omega \sim \mathcal{N}(0,(0.05)^2)$ and the parameter $\bm{\beta}\sim \mathcal{N}(0,\mathbf{I})$. Namely,
\begin{equation*}
	\mathbf{y} = \mathbf{X}\bm{\beta} + \bm{\omega}.
\end{equation*} According the linear model setting, we generate the private data $(\mathbf{A}, \mathbf{y}_{\scriptscriptstyle \!\!A}) \in \mathbb{R}^{\scriptscriptstyle  n_{\!A} \times (d+1)}$ and the public data $(\mathbf{B},\mathbf{y}_{\scriptscriptstyle \!\!B})\in \mathbb{R}^{\scriptscriptstyle n_{\!B} \times (d+1)}$. Typically, we denote the second moment of feature data as $\Sigma = \Psi + \mu \mu^T$. We present the averaged $l_2$-norm errors between the true model parameters and the DP estimations, with each experiment being conducted 300 times. 

Firstly, we compare our method with DP-RR and DP-GD (DP gradient descent, \cite{avella2023differentially}) under the same privacy parameters and without regularization, varying the private data size $n_{\scriptscriptstyle \!A}$. The amount of public data is fixed: $n_{\scriptscriptstyle \!B} = 200$ in the simulation setting.
For DP-GD, several tuning parameters must be specified to ensure convergence. We determine these values through preliminary experiments. In the simulation study, we set the number of iterations $T = 1000$, gradient clipping threshold $c = 1$, and step size $lr = 0.09$.

Figure \ref{fig:DPRR_methods_simulation} presents the averaged errors (depicted as lines) along with their corresponding standard deviations (represented by shaded regions). The results indicate that DP-GD exhibits superior robustness compared to the Sufficient Statistics Perturbation (SSP) methods, where the iterative standard deviation of DP-GD is so small that the figures cannot be displayed. However, its performance shows limited sensitivity to the size of the private dataset. In contrast, both DP-RR and DP-PMTRR demonstrate a clear decreasing trend in error as the amount of private data increases, with DP-PMTRR consistently outperforming DP-RR in terms of both robustness and accuracy. These differences can be attributed to the algorithmic characteristics of the respective methods. DP-GD is an iterative first-order method, where noise is injected into the gradients at every iteration. While this allows for general robustness, the convergence phase is particularly susceptible to noise disturb and the influence of hyperparameter tuning. As a result, the benefit from increasing data size is limited due to the persistent noise at each step. On the other hand, the SSP methods (including DP-RR and DP-PMTRR) are one-shot and precise procedures based on second-order information. These methods are tuning-free and inject noise into the Hessian matrix (i.e., second-order information), whose inverse is intrinsically more sensitive to noise. Nonetheless, they offer better data efficiency, and under favorable conditions—such as large private dataset sizes or highly ill-conditioned second-moment matrices (as in our method)—they surpass DP-GD in both accuracy and stability. Furthermore, SSP methods generally involve lower tuning and computational costs compared to DP-GD. Overall, DP-PMTRR achieves better accuracy and robustness than others via a public second-moment matrix, which may be from a small number of data.
\begin{figure}[H]
	\centering
	\includegraphics[width = 0.5\textwidth,height = 0.3\textwidth]{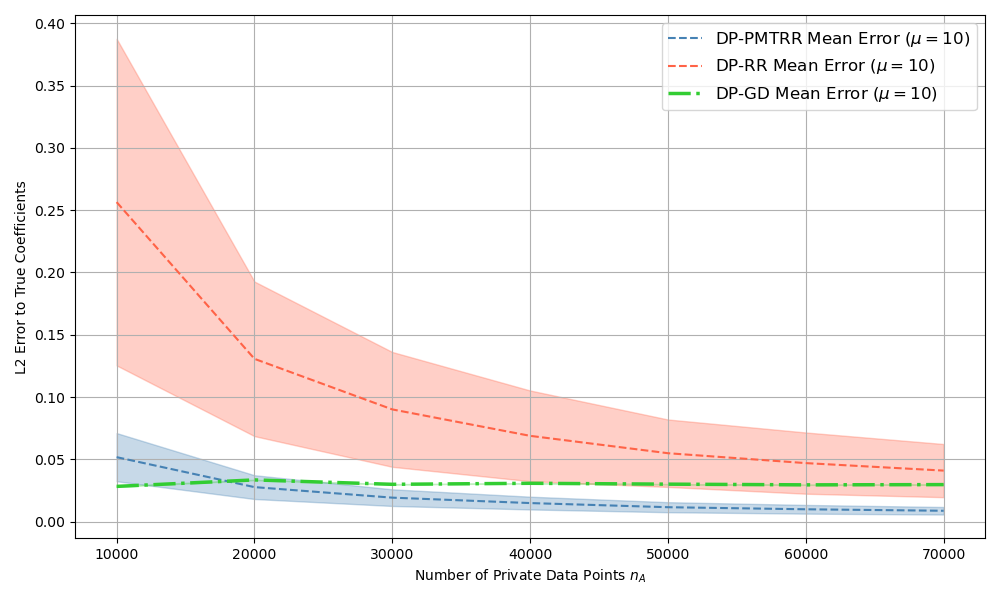}
	\caption{Simulations on DP-PMTRR, DP-RR and DP-GD.}\label{fig:DPRR_methods_simulation}
\end{figure}

Secondly, we evaluate the proposed method against the change of different privacy parameters $\mu$ and the numbers of the private data $n_{\scriptscriptstyle \!A}$, as shown in Figure \ref{fig:DPRR_contrast_simulation}. In the simulation, we fix the public data $n_{\scriptscriptstyle \!B} = 20$, the regularization parameter $\lambda = 0.01$, and the probability parameter $\eta = 0.05$.

\begin{figure}[H]
		\centering
		\includegraphics[width = 0.5\textwidth,height = 0.3\textwidth]{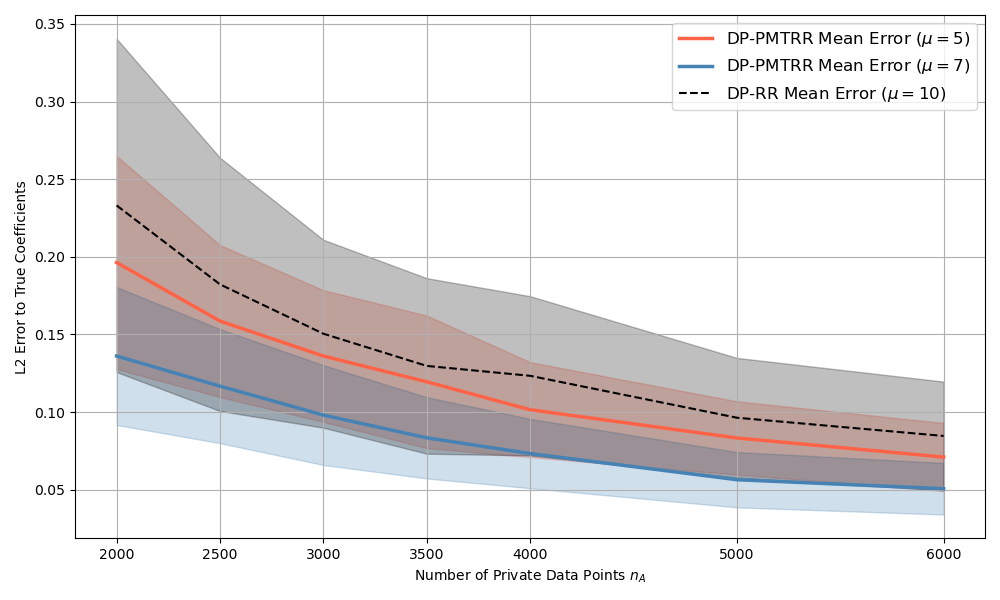}

		\caption{Simulations with different private data sizes and privacy parameters.}\label{fig:DPRR_contrast_simulation}
\end{figure}
Figure \ref{fig:DPRR_contrast_simulation} illustrates the averaged errors (depicted as lines) and the standard deviations of errors (shown as shaded areas). Overall, the results indicate that DP-PMTRR performs better accuracy and robustness than DP-RR, even with a smaller privacy-budget, and well across different privacy parameters and private data sizes. Notably, the method exhibits clear decreasing trends in averaged errors and standard deviations as the size of private data and the privacy parameters increase. These verify our theoretical results.

Finally, we explore the impact of the regularization parameter, seeing Figure \ref{fig:DPRR_lam_simulation}. In the simulation, we fix the private data $n_{\scriptscriptstyle \!A} = 1e4$, the public data $n_{\scriptscriptstyle \!B} = 100$, the privacy parameter $\mu = 3$, and the probability parameter $\eta = 0.05$.
\begin{figure}[H]
	\centering
	\includegraphics[width = 0.5\textwidth,height = 0.3\textwidth]{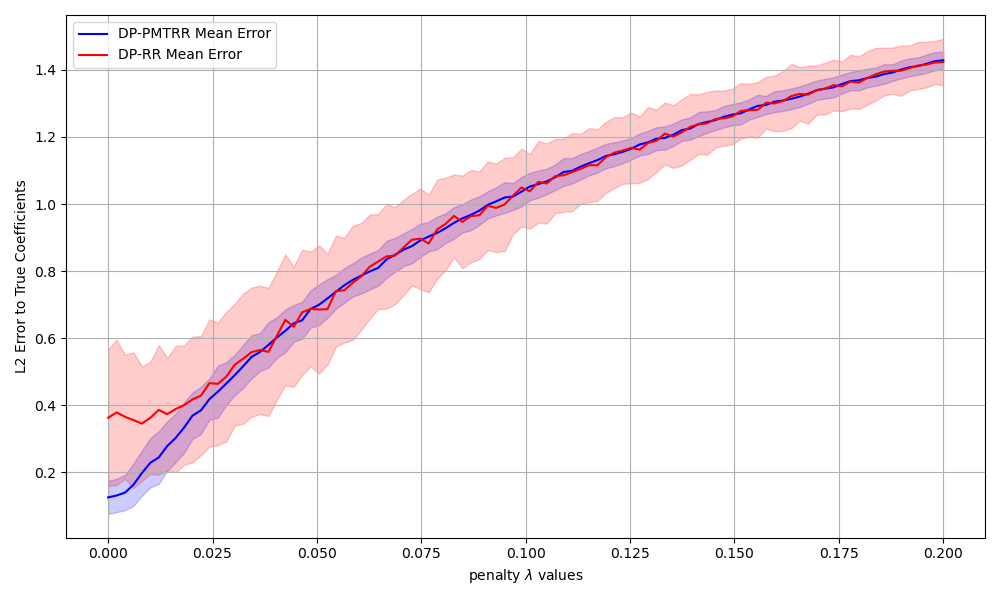}
		\caption{Simulations with different regularization parameters.}\label{fig:DPRR_lam_simulation}
\end{figure}

Figure \ref{fig:DPRR_lam_simulation} presents the averaged errors (shown as lines) and the standard deviations of errors (shown as shaded areas). Overall, increasing the regularization parameter $\lambda$ tends to reduce the standard deviation of the errors, indicating improved robustness, but at the cost of increased averaged errors. In particular, these reveal a trade-off in DP-RR: there exists an optimal $\lambda$ that minimizes the averaged error, but this choice does not offer significant robustness. Conversely, choosing a larger $\lambda$ improves robustness but results in higher averaged error. By contrast, our method effectively mitigates this trade-off. Both the accuracy and robustness of our approach are less sensitive to the choice of $\lambda$, indicating improved overall utility.

\subsubsection{Real-world Datasets}
We use two real-world datasets from UCI, White-wine Quality \cite{wine_quality_186} and Combined Cycle Power Plant \cite{combined_cycle_power_plant_294}. The goal of the White-wine Quality dataset is to model wine quality based on physicochemical tests, which includes $4898$ samples with $11$ continuous physicochemical features and one integer target (quality, scored between 0 and 10) and is viewed as a linear regression task. The Combined Cycle Power Plant dataset contains $9568$ samples collected from a Combined Cycle Power Plant over 6 years (2006-2011). The $4$ features consist of hourly average ambient variables to predict the net hourly electrical energy output of the plant. We separate them into two parts: private dataset and public dataset. We present the averaged $l_2$-norm errors between the DP and the non-DP estimations, with each experiment being conducted 300 times.

Firstly, we compare our method with DP-RR and DP-GD under the same privacy parameters and without regularization, varying the private data size $n_{\scriptscriptstyle \!A}$. In all experiments, the amount of public data is fixed: $n_{\scriptscriptstyle \!B} = 245$ for the White-wine Quality dataset and $n_{\scriptscriptstyle \!B} = 192$ for the Combined Cycle Power Plant dataset. For DP-GD, several tuning parameters must be specified to ensure convergence. We determine these values through preliminary experiments. In the White-wine Quality dataset, we use $T = 1000$, $c = 3.0$, and $lr = 0.5$. For the Combined Cycle Power Plant dataset, we use $T = 1000$, $c = 2.0$, and $lr = 0.04$.

Figure \ref{fig:DPRR_methods_real} presents the averaged errors (depicted as lines) along with their corresponding standard deviations (represented by shaded regions). The results in the real-world datasets are similar to the first simulation and indicate DP-PMTRR achieves better accuracy and robustness than others for real-world data.
\begin{figure}[H]
	\begin{subfigure}[White-wine Quality]{
		\includegraphics[width = 0.5\textwidth,height = 0.3\textwidth]{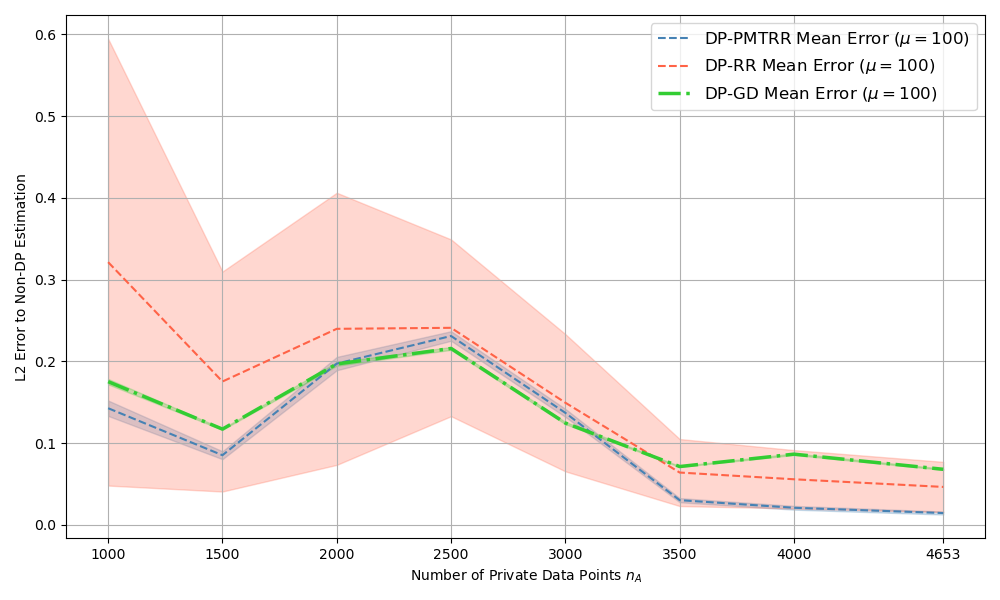}}
	\end{subfigure}
	\hfill
	\begin{subfigure}[Combined Cycle Power Plant]{
		\includegraphics[width = 0.5\textwidth,height = 0.3\textwidth]{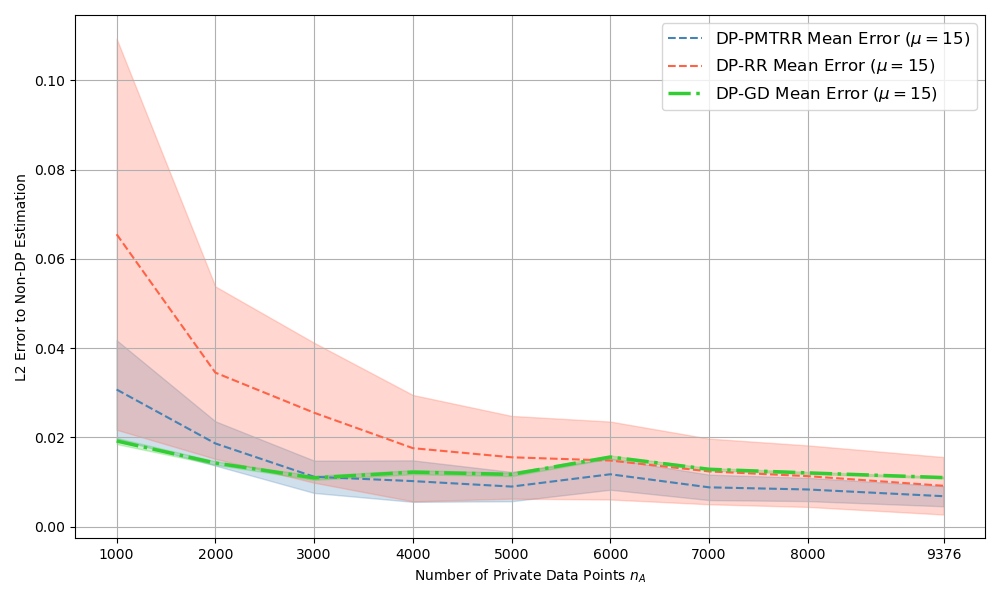}}
	\end{subfigure}
	\caption{Real-world experiments on DP-PMTRR, DP-RR and DP-GD.}\label{fig:DPRR_methods_real}
\end{figure}

Secondly, we evaluate the proposed method against the change of different privacy parameters $\mu$ and the numbers of the private data $n_{\scriptscriptstyle \!A}$, as shown in Figure \ref{fig:DPRR_contrast_real}. In the white-wine quality dataset, we fix the public data $n_{\scriptscriptstyle \!B} = 245$, the regularization parameter $\lambda = 0$, and the probability parameter $\eta = 1e-3$. In the Combined Cycle Power Plant dataset, we fix the public data $n_{\scriptscriptstyle \!B} = 192$, the regularization parameter $\lambda = 0$, and the probability parameter $\eta = 1e-3$. 

\begin{figure}[H]
	\begin{subfigure}[White-wine Quality]{
		\includegraphics[width = 0.5\textwidth,height = 0.3\textwidth]{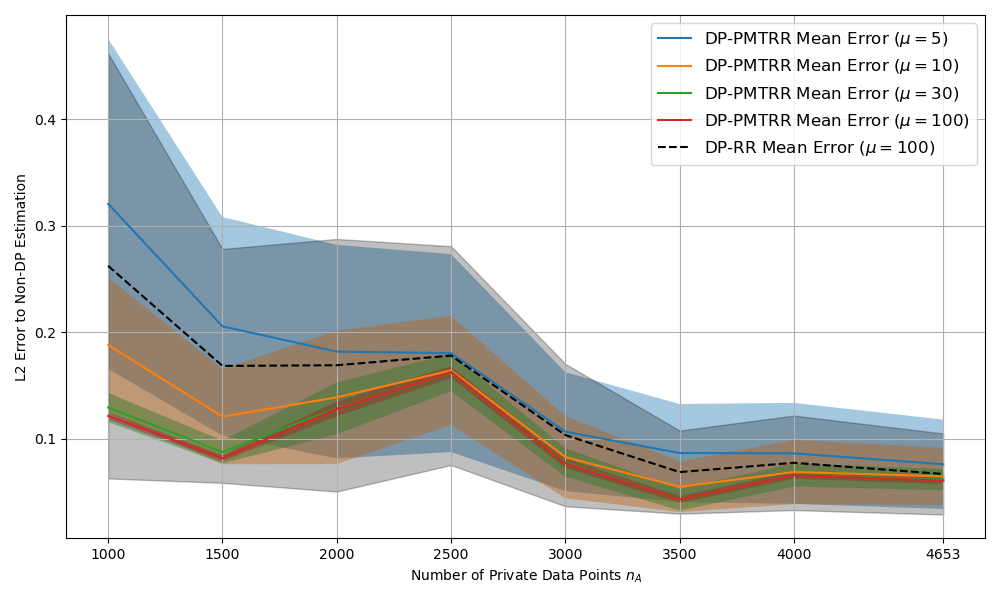}}
	\end{subfigure}
	\hfill
	\begin{subfigure}[Combined Cycle Power Plant]{
		\includegraphics[width = 0.5\textwidth,height = 0.3\textwidth]{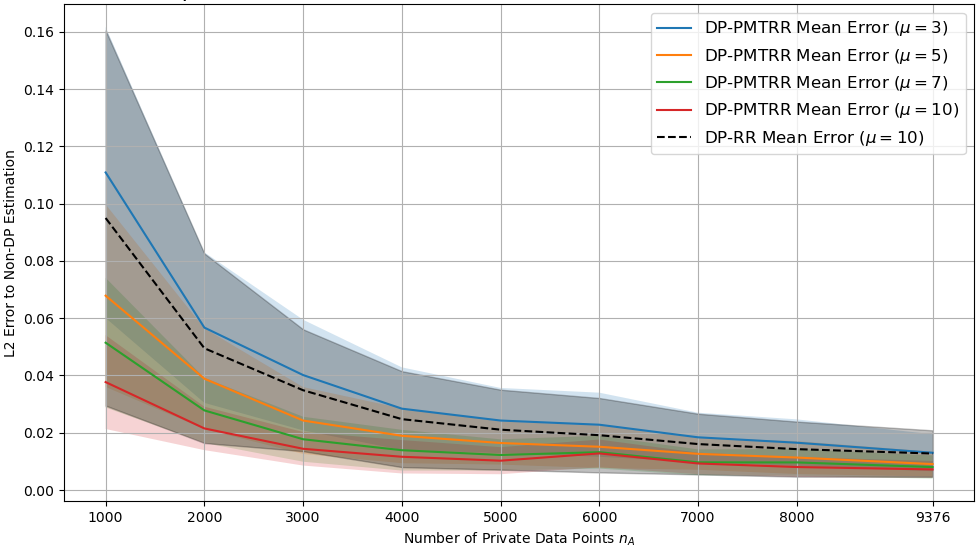}}
	\end{subfigure}
		\caption{Real-world experiments with different private data sizes and privacy parameters.}\label{fig:DPRR_contrast_real}
\end{figure}
Figure \ref{fig:DPRR_contrast_real} illustrates the averaged errors (depicted as lines) and the standard deviations of errors (shown as shaded areas). The results in real-world datasets indicate that DP-PMTRR performs as well as the second simulation; see Figure \ref{fig:DPRR_contrast_simulation}.

Finally, we explore the impact of the regularization parameter in real-world datasets, seeing Figure \ref{fig:DPRR_lam_real}. In the White-wine Quality dataset, we fix the private data $n_{\scriptscriptstyle \!A} = 4649$, the public data  $n_{\scriptscriptstyle \!B} = 245$, the privacy parameter $\mu = 20$, and the probability parameter $\eta = 1e-3$. In the Combined Cycle Power Plant dataset, we fix the private data $n_{\scriptscriptstyle \!A} = 9376$, the public data  $n_{\scriptscriptstyle \!B} = 192$, the privacy parameter $\mu = 3$, and the probability parameter $\eta = 0.05$.

\begin{figure}[H]
	\begin{subfigure}[White-wine Quality]{
		\includegraphics[width = 0.5\textwidth,height = 0.3\textwidth]{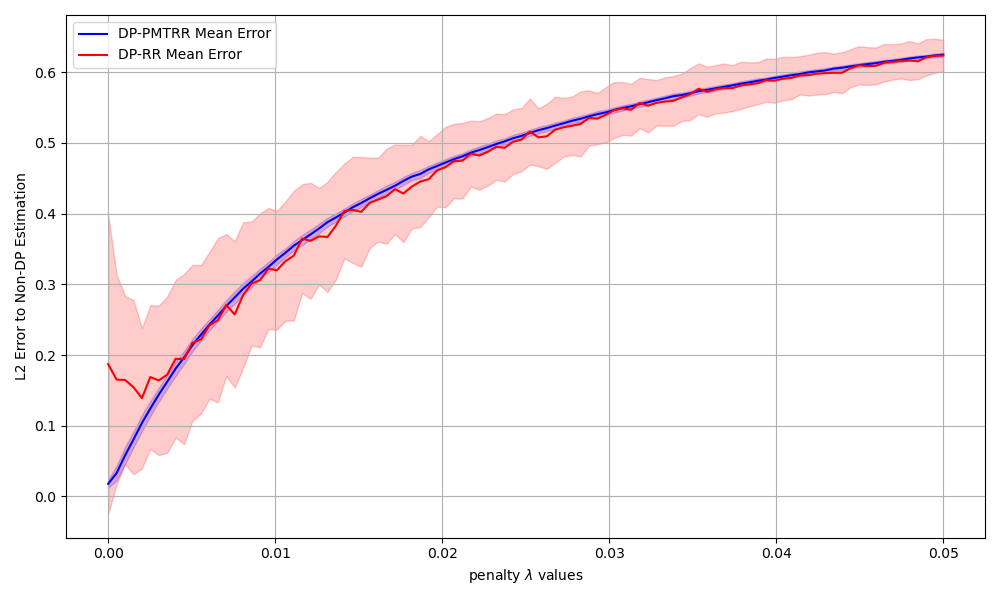}}
	\end{subfigure}
	\hfill
	\begin{subfigure}[Combined Cycle Power Plant]{
		\includegraphics[width = 0.5\textwidth,height = 0.3\textwidth]{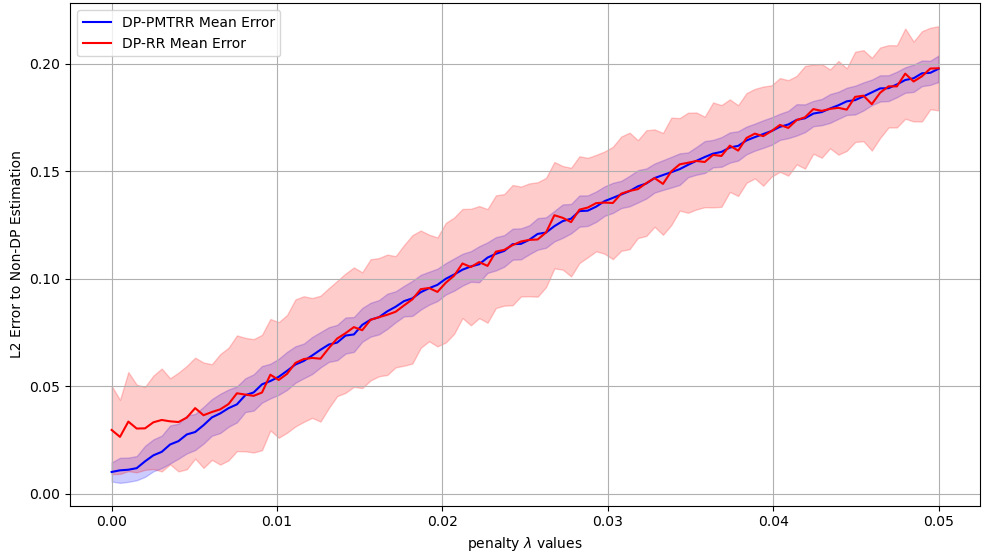}}
	\end{subfigure}
		\caption{Real-world experiments with different regularization parameters.}\label{fig:DPRR_lam_real}
\end{figure}
Figure \ref{fig:DPRR_lam_real} exists with the same effect in the third simulation. The real-world experiments verify our method in practice.

\subsection{DP Logistic Regression}\label{subsec:DPLR}
\subsubsection{Simulations}

We design a simulation of the logistic regression model, where the feature data $\mathbf{X}$ with the dimension $d = 8$ and the true model parameters are generated the same as the simulation in the subsection \ref{subsec:DPRR}. And the responses are generated according to these settings. All experiments are repeated 100 times, and we show the averaged $l_2$-norm errors between the true model parameters and the DP estimations at every iteration.

Firstly, we compare the DP-LR and DP-PMTLR (our method) with different privacy parameters $\mu$, as shown in Figure \ref{fig:DPLR_mu_sim}. In the simulation, we fix the private data $n_{\scriptscriptstyle \!B} = 1e4$, the public data $n_{\scriptscriptstyle \!B} = 100$, the regularization parameter $\lambda = 1e-3$, and the probability parameter $\eta = 0.05$.

\begin{figure}[H]

	\centering
	\includegraphics[width = 0.9\textwidth]{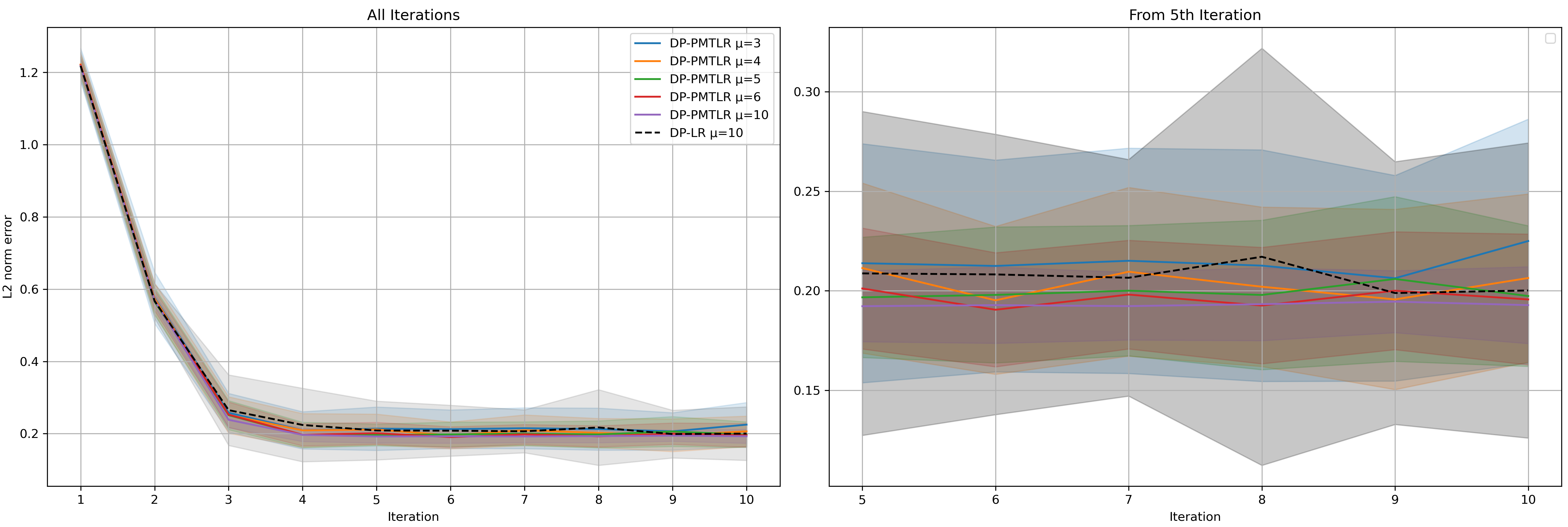}
	\caption{Simulations of DP-PMTLR and DP-LR with different privacy parameters.}\label{fig:DPLR_mu_sim}

\end{figure}

Figure \ref{fig:DPLR_mu_sim} displays the averaged errors (as lines) and the standard deviations of errors (as shaded areas). We typically show the final iterations in Figure \ref{fig:DPLR_mu_sim} to allow for detailed comparison. Overall, increasing the privacy parameter—corresponding to weaker privacy protection and lower noise—leads to more accurate DP estimations and more stable iterations. Notably, Figure \ref{fig:DPLR_mu_sim} shows that even when using smaller privacy parameters (i.e., stronger privacy protection and higher noise), DP-PMTLR achieves lower errors and greater robustness than DP-LR. In contrast, DP-LR with the largest privacy parameter still exhibits higher error and less robustness. These results demonstrate that our method outperforms standard DP logistic regression in both utility and robustness.

Secondly, we explore the impact of the regularization parameter, seeing Figure \ref{fig:DPLR_lam_sim}. In the simulation, we fix the private data $n_{\scriptscriptstyle \!A} = 8000$, the public data $n_{\scriptscriptstyle \!B} = 100$, the privacy parameter $\mu = 10$, and the probability parameter $\eta = 0.05$. 

\begin{figure}[H]

	\centering
	\includegraphics[width = 0.9\textwidth]{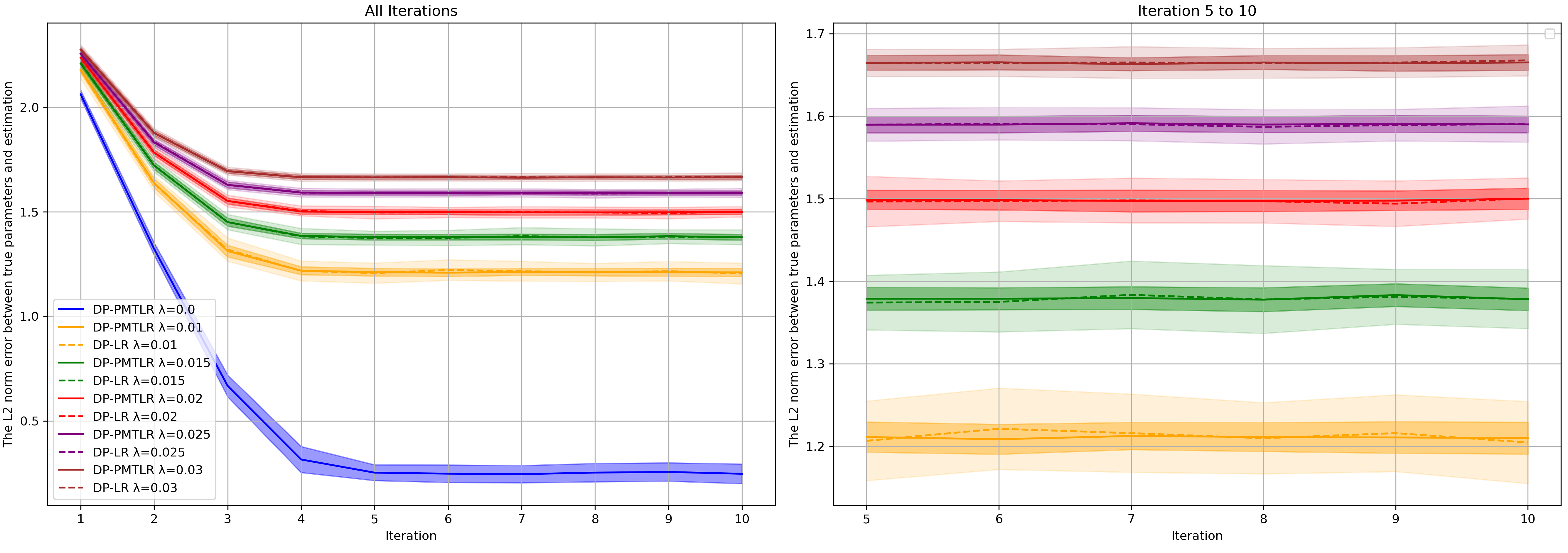}

	\caption{Simulations of DP-PMTLR and DP-LR with different regularization parameters.}\label{fig:DPLR_lam_sim}

\end{figure}

Figure \ref{fig:DPLR_lam_sim} illustrates the averaged errors (shown as lines) and the standard deviations of errors (depicted as shaded areas), where the narrow and dark shaded area is the standard deviation of DP-PMTLR. The result verifies DP-PMTLR is more robust in iterations. Increasing the regularization parameter $\lambda$ reduces the standard deviations, indicating improved robustness, but also results in larger averaged errors. Notably, DP-LR fails to converge when $\lambda = 0$ (original regularization parameter), while DP-PMTLR remains robust and converges successfully under the same condition.These results demonstrate that our approach effectively resolves the trade-off about the regularization and that its utility and robustness are only weakly dependent on the choice of $\lambda$.

\subsubsection{Real-world Datasets}
About the real-world datasets, we use two datasets from UCI, Bank Marketing \cite{bank_marketing_222} and Banknote Authentication \cite{banknote_authentication_267}. The Bank Marketing dataset is related to direct marketing campaigns (phone calls) of a Portuguese banking institution, including $45211$ samples and $16$ features. The classification goal is to predict if the client will subscribe to a term deposit. The Banknote Authentication dataset was extracted from images that were taken for the evaluation of an authentication procedure for banknotes, including $1372$ samples and $4$ features. we add a column of vectors containing all 1s, so the final feature dimension is $d = 5$. We set aside 10\% for the public dataset. All experiments are repeated 100 times, and we show the averaged $l_2$-norm errors between the DP and the non-DP estimations at every iteration.

Firstly, we compare the DP-LR and DP-PMTLR (our method) with different privacy parameters $\mu$, as shown in Figure \ref{fig:DPLR_mu_real}. In the real datasets, we use all private data and public data. In the Bank Marketing dataset, we fix the regularization parameter $\lambda = 0$ and the probability parameter $\eta = 1e-3$. In the Banknote Authentication dataset, we fix the regularization parameter $\lambda = 0.01$ and the probability parameter $\eta = 1e-3$.

\begin{figure}[H]
	\centering{
	\begin{subfigure}[Bank Marketing]{
		\centering
		\includegraphics[width = 0.9\textwidth]{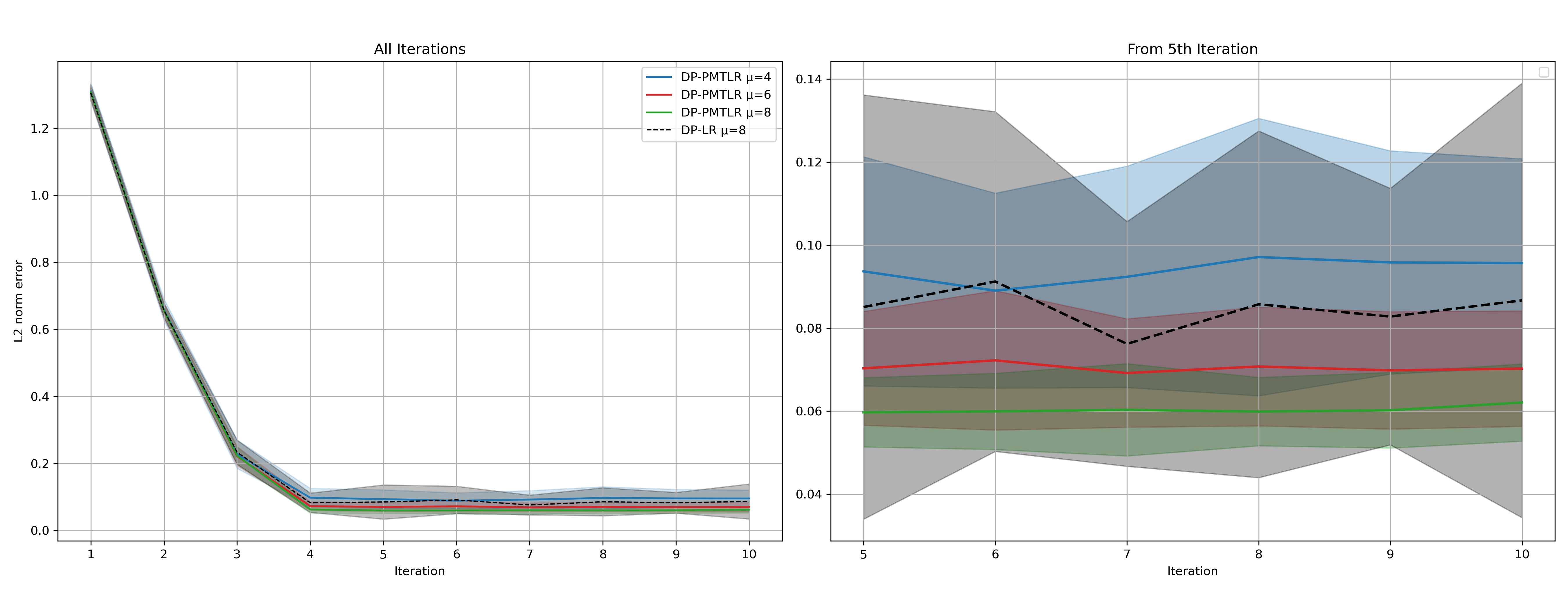}\label{fig:DP_PMTLR_mu_marketing}}
	\end{subfigure}
	\hfill
	\begin{subfigure}[Banknote Authentication]{
		\centering
		\includegraphics[width = 0.9\textwidth, height =0.3\textheight,keepaspectratio]{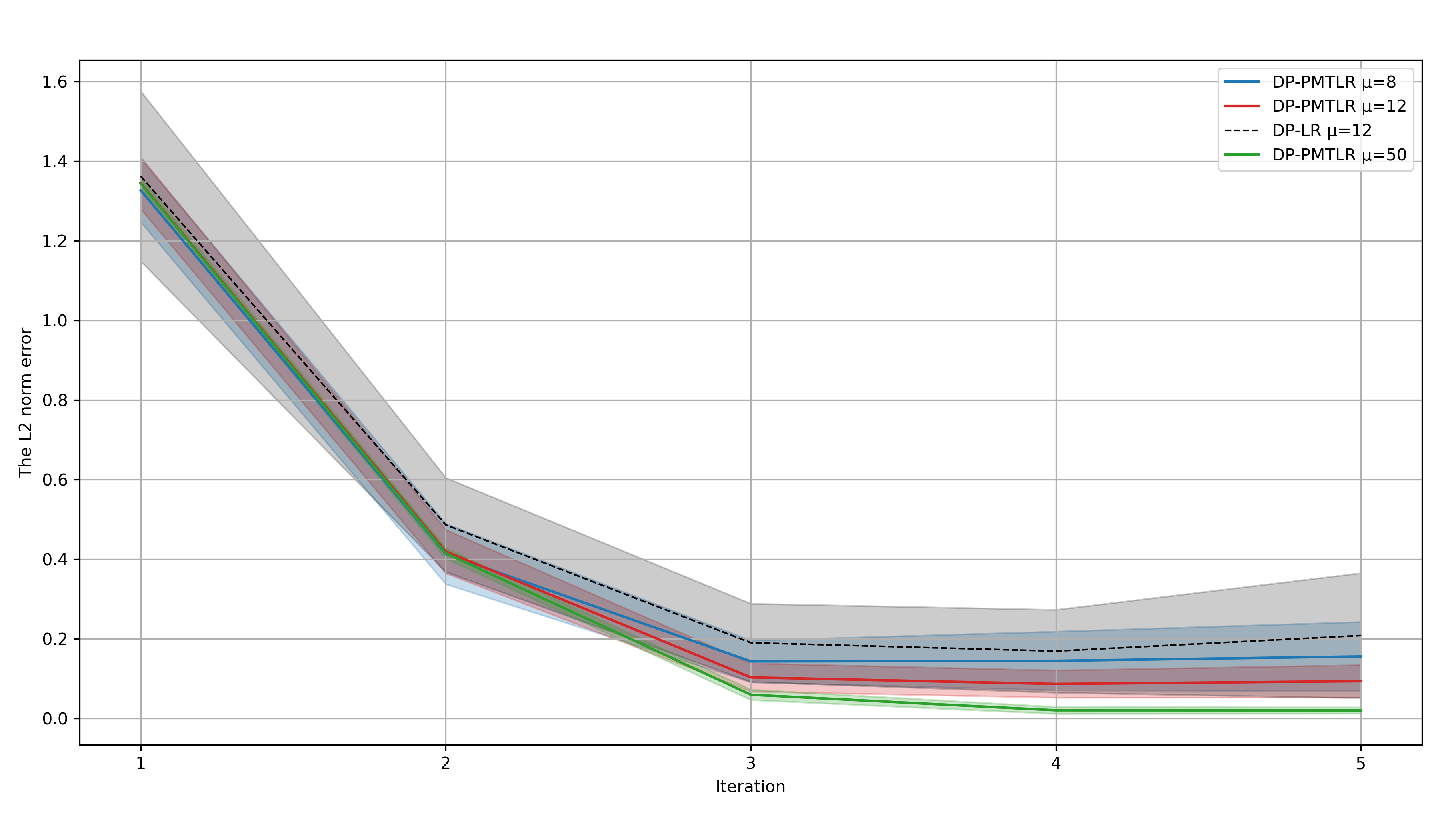}\label{fig:DP_PMTLR_mu_note}}
	\end{subfigure}
		\caption{Real-world experiments of DP-PMTLR and DP-LR with different privacy parameters.}\label{fig:DPLR_mu_real}
	}
\end{figure}

Figure \ref{fig:DPLR_mu_real} displays the averaged errors (as lines) and the standard deviations of errors (as shaded areas). We typically show the final iterations in the Fig.\ref{fig:DP_PMTLR_mu_marketing} to allow for detailed comparison. Overall, these results are similar to Figure \ref{fig:DPLR_mu_sim}, where our method also outperforms standard DP logistic regression in both utility and robustness in practice.

Secondly, we explore the impact of the regularization parameter in the real-world datasets, seeing Figure \ref{fig:DPLR_lam_real}. We fix the private data, the public data, and the probability parameter $\eta = 1e-3$. The privacy parameters $\mu = 3$ and $\mu = 10$ are in the Bank Marketing dataset and the Banknote Authentication dataset, respectively.

\begin{figure}[H]
	\centering{
	\begin{subfigure}[Bank Marketing]{
		\centering
		\includegraphics[width = 0.9\textwidth]{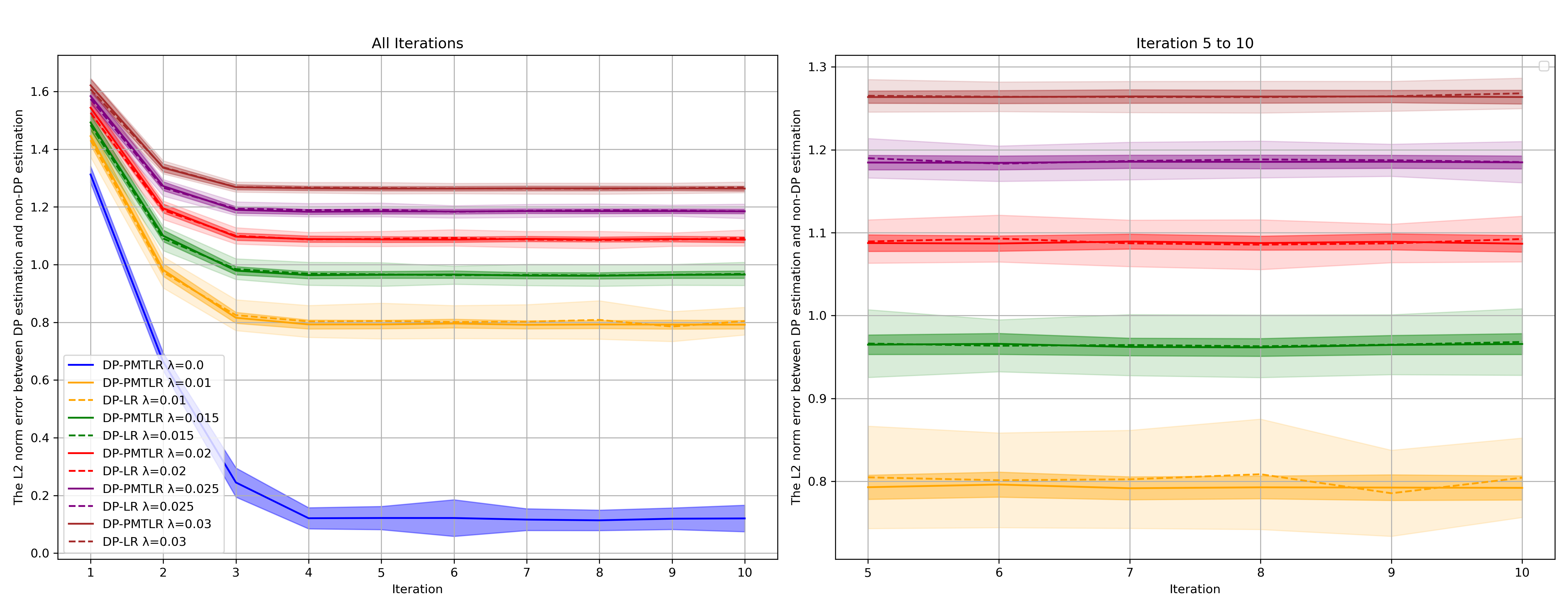}\label{fig:DP_PMTLR_lam_marketing}}
	\end{subfigure}

	\begin{subfigure}[Banknote Authentication]{
		\centering
		\includegraphics[width = 0.9\textwidth, height =0.3\textheight,keepaspectratio]{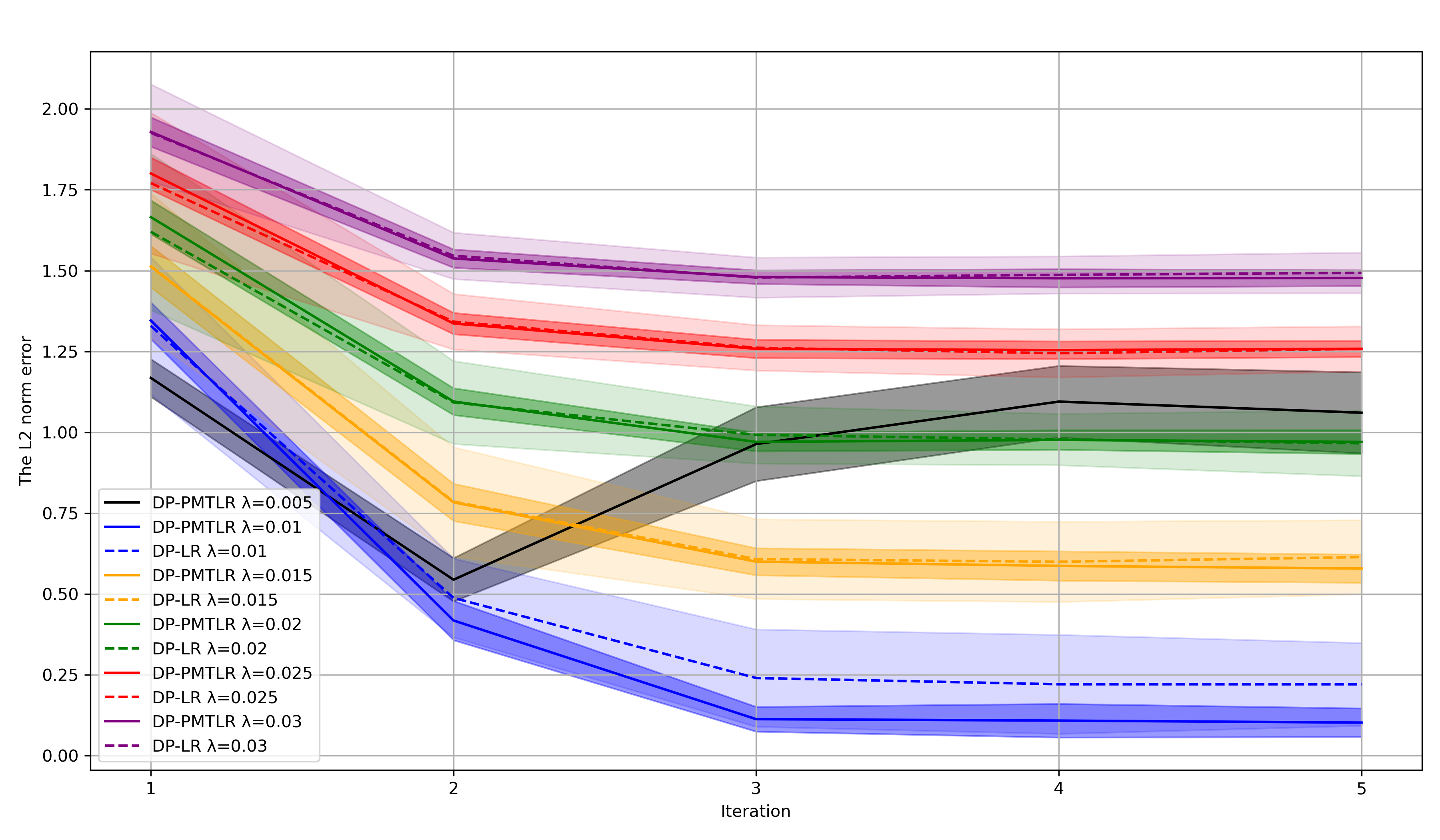}\label{fig:DP_PMTLR_lam_note}}
	\end{subfigure}

		\caption{Real experiments of DP-PMTLR and DP-LR with different regularization parameters.}\label{fig:DPLR_lam_real}
	}
\end{figure}

Figure \ref{fig:DPLR_lam_real} illustrates the averaged errors (shown as lines) and the standard deviations of errors (depicted as shaded areas), where the narrow and dark shaded area is the standard deviation of DP-PMTLR. All results verify DP-PMTLR is more robust in iterations. In the real-world experiments, increasing the regularization parameter $\lambda$ also reduces the standard deviations, indicating improved robustness, but also results in larger averaged errors. Notably, in Figure \ref{fig:DP_PMTLR_lam_marketing}, DP-LR fails to converge when $\lambda = 0$ (original regularization parameter), while DP-PMTLR remains robust and converges successfully under the same condition. Figure \ref{fig:DP_PMTLR_lam_note} shows that DP-LR fails to converge when $\lambda = 0.005$, and DP-PMTLR arrives at the best performance when $\lambda = 0.01$ (original regularization parameter). These illustrate that DP-LR must abandon the prior regularization parameter for stability; however, DP-PMTLR can make use of it and achieves a tuning-free effect. These results indicate that our method adequately addresses the trade-off concerning regularization, and its utility and resilience exhibit minimal dependence on the selection of $\lambda$.

\section{Conclusion}\label{sec:conclusion}
In this paper, we address the challenge of unbounded private data under differential privacy (DP), leveraging the second-moment estimation from a small amount of public data. We propose a novel method, PMT, which transforms private data using the public second-moment matrix. A principled truncation scheme is then applied, with the truncation radius determined solely by non-private quantities—data dimension and sample size. This design carefully balances the trade-off between information retention and DP sensitivity.

Using the transformed and truncated data, we demonstrate that the resulting second-moment matrix is well-conditioned and extend our theoretical analysis to include its regularized counterpart. Importantly, our transformation significantly enhances the robustness and error bounds of the inverse second-moment matrix under DP. It removes the dependency on the average condition number of the second-moment matrix and mitigates the influence of both the inverse norm of the second-moment matrix and the regularization parameter—challenges that are otherwise unavoidable with private data alone.

Building on this foundation, we develop two loss functions and algorithms, and demonstrate their applicability to classical penalized regression and generalized linear models in DP scenarios. Specifically, we use ridge regression and logistic regression to show that their DP estimators based on PMT get significant improvements -- the former achieves substantial gains in robustness and error control for closed-form solutions, while the latter stabilizes DP Newton's method without inflating the regularization parameter, leading to more accurate estimations and tighter theoretical guarantees. Our theoretical analysis attributes these improvements to the conditioning effect of the transformation on the second-moment matrix and the methods' weak dependence on the regularization parameter.

Extensive experiments validate our method's effectiveness, robustness, and alignment with theoretical insights, demonstrating that even a small amount of public data can significantly enhance the utility of differentially private regression models.

Our work demonstrates that public information is able to improve DP algorithms effectively. It inspires us to consider the following future directions. First, we can consider how to use the public second moment to improve other DP algorithms. Second, other public information may be explored and incorporated, e.g., other public statistics and published models.

\bibliography{bib/ref}
\bibliographystyle{plainnat}

\appendix
\input{tex/appendix}

\end{document}

%% file: tex/math_commands.tex
\usepackage{amsmath,amsfonts,bm}

\def\1{\bm{1}}

\DeclareMathAlphabet{\mathsfit}{\encodingdefault}{\sfdefault}{m}{sl}
\SetMathAlphabet{\mathsfit}{bold}{\encodingdefault}{\sfdefault}{bx}{n}

\def\I{{\bf I}}

\def\Xi{\boldsymbol{\xi}}

\def\Phi{\boldsymbol{\phi}}

\def\0{{\bf 0}}
\def\1{{\bf 1}}

\newcommand{\sigmoid}{\sigma}

\DeclareMathOperator{\Tr}{Tr}

%% file: tex/appendix.tex
\begin{appendix}
	\onecolumn
	\begin{center}
		{\huge {Supplementary Material to "Enhancing Differential Privacy with Public Data: New Frameworks for Ridge and Logistic Regression"}}
	\end{center}

\section{Useful Tools}

There are some facts about Gaussian vector and Gaussian matrix, which help us to control the DP noise. We denote the $d \times d$ symmetric Gaussian random matrix $\mathbf{G} \sim SG_d(\sigma^2)$ with elements $G_{ij}  \overset{i.i.d.}{\sim} \mathcal{N}(0, \sigma^2)$.
\begin{lem}[Symmetric Gaussian matrix bound \cite{avella2023differentially}]\label{lem:SGbound}
	For a $d \times d$ symmetric matrix $\mathbf{G}\sim SG_d(\sigma^2)$, with probability $1 - \eta$, 
	\begin{equation}\label{eq:symGau}
		\|\mathbf{G}\|_2 \leq \sigma\sqrt{2d\log(2d/\eta)}.
	\end{equation}
\end{lem}

\begin{lem}[Gaussian vector bound \cite{laurent2000adaptive}]\label{lem:gaussiannorm}
	For a gaussian vector $\mathbf{g} = (g_1,...,g_d),\ g_i\sim\mathcal{N}(0,\sigma^2)$, the $\|\mathbf{g}\|_2^2$ satisfies
   \begin{equation*}
	   \mathbb{P}[\|\mathbf{g}\|_2^2 \geq \sigma^2(2\sqrt{d\log(1/\eta)} + 2\log(1/\eta)+d)] \leq \eta.
   \end{equation*}
\end{lem}

The next lemma gives the largest and smallest singular values of the sub-Gaussian matrix. It helps us to analyze the second-moment matrix estimation and control its condition number. 
\begin{lem}[Singular values bound \cite{vershynin2010introduction}]\label{lem:singular_bound}
	Let $\mathbf{X}$ be a $n \times d$ random matrix whose each row $\mathbf{x}_i$ is independently non-isotropic sub-gaussian random vectors in $\mathbb{R}^d$ with the second-moment matrix $\Sigma $. Then for every $\eta > 0$ and $\sqrt{n} \geqslant O(\sqrt{d} + \sqrt{\log(1/\eta)})$, with at least $1 - 2\eta$, one has  
	\begin{equation*}
    \begin{aligned}
      &\sqrt{\lambda_{min}(\Sigma)} (\sqrt{n}- C\sqrt{d} - \sqrt{\frac{1}{c}\log(1/\eta)}) \leq \lambda_{min}(\mathbf{X}), \\
      &\sqrt{\lambda_{max}(\Sigma)} (\sqrt{n} + C\sqrt{d} + \sqrt{\frac{1}{c}\log(1/\eta)}) \geq \lambda_{max}(\mathbf{X}),\\
    \end{aligned}
	\end{equation*} where $\lambda(\cdot)$ means singular value, and $C = C_K$, $c=c_K >0$ depend on the sub-gaussian norm $K = max_i \|\mathbf{x}_i\|_{\psi_2}$. Typically, for the isotropic situation $\Sigma = \mathbf{I}$, one has 
	\begin{equation*}
    \begin{aligned}
      &\sqrt{n} - C\sqrt{d} - \sqrt{\frac{1}{c}\log(1/\eta)} \leq \lambda_{min}(\mathbf{X}),\\
      &\sqrt{n} + C\sqrt{d} + \sqrt{\frac{1}{c}\log(1/\eta)} \geq \lambda_{max}(\mathbf{X}).\\
    \end{aligned}
	\end{equation*}
\end{lem}

\begin{cor}\label{cor:secondmoment}
	Let $\mathbf{X}$ be a $n \times d$ random matrix whose each row $\mathbf{x}_i$ is independently isotropic sub-gaussian random vector in $\mathbb{R}^d$ with the second-moment matrix $\Sigma = \mathbf{I}$. Then for every $\eta > 0$ and $\sqrt{n} \geqslant O(\sqrt{d} + \sqrt{\log(1/\eta)})$, one has
	{\scriptsize
  \begin{equation*}
		\mathbb{P}\big[ \frac{\lambda_{max}(\mathbf{X}^T \mathbf{X})}{n} = \frac{\lambda_{max}^2(\mathbf{X})}{n} \geq \Big(1 + O\Big(\sqrt{\frac{d}{n}} + \sqrt{\frac{\log(1/\eta)}{n}}\Big)\Big)^2 \big] \leq \eta
	\end{equation*}}and
  {\scriptsize
	\begin{equation*}
		\mathbb{P}\big[ \frac{\lambda_{min}(\mathbf{X}^T \mathbf{X})}{n} = \frac{\lambda_{min}^2(\mathbf{X})}{n} \leq \Big(1 - O\Big(\sqrt{\frac{d}{n}} + \sqrt{\frac{\log(1/\eta)}{n}}\Big)\Big)^2 \big] \leq \eta.
	\end{equation*}}
\end{cor} 

The following lemma on the $l_2-norm$ of a sub-gaussian vector illustrates that the sample's length depends on the second-moment matrix. It provides the theoretical guarantee about a principled truncating radius. 
\begin{lem}[Bounded $l_2$-norm of Sub-gaussian Vector]\label{lem:nonisotropic_norm}
	Let $\mathbf{x} = (x_1,...x_d)^T \in \mathbb{R}^d$  is a non-isotropic sub-gaussian random vector with $\mathbb{E} \mathbf{x}\mathbf{x}^T = \Sigma$ . Then $\|\mathbf{x}\|_2^2$ is sub-exponential and 
	\begin{equation*}
	\mathbb{P}\Big[\big| \|\mathbf{x}\|_2^2 - \sum_{i=1}^{d}\lambda_i| \geq t\Big] \leq 2\exp\big(-\frac{ct}{dK^2} \big),
	\end{equation*} where $K = max_i \|x_i\|_{\psi_2}$, $c$ is an absolute constant on $K$ and $\lambda_1>...>\lambda_d$ are the eigenvalues of $\Sigma$. One more thing, with at least probability $1 - \eta$,
	\begin{equation*}
		\|\mathbf{x}\|_2^2 \leq \Tr(\Sigma) + \big(\frac{dK^2}{c}\log(\frac{2}{\eta})\big) \leq O\big(d(\overline{\Tr}(\Sigma) + \log(\frac{2}{\eta}))\big),
	\end{equation*}where $\overline{\Tr}(\Sigma) = \frac{1}{d}\Tr(\Sigma)$ is the average trace of $\Sigma$.
\end{lem}

The following lemmas are used to analyze the noisy inverse matrix and provide theoretical tools to quantify the robustness condition.
\begin{lem}\label{lem:inversebound}
	Denote a square matrix $\Sigma$ and a disturb matrix $\mathbf{G}$, the condition number of $\kappa(\Sigma) = \|\Sigma\|\|\Sigma^{\scriptscriptstyle -\!1}\|$. Then, 
	\begin{equation*}
		\frac{\|(\Sigma + \mathbf{G})^{\scriptscriptstyle \!-\!1} - \Sigma^{\scriptscriptstyle -\!1}\|}{\|(\Sigma + \mathbf{G})^{\scriptscriptstyle \!-\!1}\|} \leq \kappa(\Sigma) \frac{\|\mathbf{G}\|}{\|\Sigma\|}.
	\end{equation*} Moreover, if $\kappa(\Sigma)\frac{\|\mathbf{G}\|}{\|\Sigma\|} = \|\Sigma^{\scriptscriptstyle -\!1}\|\|\mathbf{G}\| \leq 1$, then 
	\begin{equation*}
		\|(\Sigma + \mathbf{G})^{\scriptscriptstyle -\!1}\| \leq \frac{\|\Sigma^{\scriptscriptstyle -\!1}\|}{1 - \kappa(\Sigma)\frac{\|\mathbf{G}\|}{\|\Sigma\|}}.
	\end{equation*} Moreover
	\begin{equation*}
		\frac{\|(\Sigma + \mathbf{G})^{\scriptscriptstyle \!-\!1} - \Sigma^{\scriptscriptstyle -\!1}\|}{\|\Sigma^{\scriptscriptstyle -\!1}\|} \leq \frac{\kappa(\Sigma) \frac{\|\mathbf{G}\|}{\|\Sigma\|}}{1 - \kappa(\Sigma) \frac{\|\mathbf{G}\|}{\|\Sigma\|}}.
	\end{equation*}
\end{lem}

\begin{lem}[Weyl's Inequality]\label{lem:welyinequal}
	Denote a matrix $\mathbf{M}$ and a $\mathbf{H}$ are real symmetric matrices , then 
	\begin{equation*}
		\lambda_{max}(\mathbf{M} + \mathbf{H}) \leqslant \lambda_{max}(\mathbf{M}) + \lambda_{max}(\mathbf{H}),
	\end{equation*}
	\begin{equation*}
		\lambda_{min}(\mathbf{M} + \mathbf{H}) \geqslant \lambda_{min}(\mathbf{M}) + \lambda_{min}(\mathbf{H}),
	\end{equation*}where $\lambda_{max}(\cdot)$ and $\lambda_{min}(\cdot)$ are the largest and smallest eigenvalues of the matrix, respectively.
\end{lem}

The following lemma (\textbf{Lemma 21} in \cite{avella2023differentially}) is the main tool to guarantee the converge of generalized linear regressions via GDP Newton's method.
\begin{lem}[\cite{avella2023differentially}]\label{lem:loss_property}
A loss function $\mathcal{L}(\bm{\beta};\mathbf{A}) = \frac{1}{n_{\!A}} \sum_{i=1}^{n}l(\bm{\beta};A_i) + \lambda \|\bm{\beta}\|^2$ is the locally $\gamma$-strong convexity. For any $\bm{\beta}$ in some neighboring area of the minimizer $\hat{\bm{\beta}}$, its Hessian matrix $\|\nabla^2 \mathcal{L}(\bm{\beta};\mathbf{A})\|_2 \leq B_h$ exhibits local $C_h$-Lipschitz continuity and the gradient $\|\nabla \mathcal{L}(\bm{\beta};\mathbf{A})\|_2 \leq B_g$. For the $\mu$-GDP Newton's method and all iteration $0\leq t \leq T$, when the private data $n_{\scriptscriptstyle \!A}$ makes $\frac{\Delta_h \sqrt{Td\log(Td/\eta)}}{\mu \cdot \gamma \cdot n_A}$ small enough and $\Delta_h = \max\limits_{\bm{\beta},A\sim A'}\|\nabla^2 l(\bm{\beta};A) - \nabla^2 l(\bm{\beta};A')\|_F$ is the sensitivity of the function $\nabla^2l(\bm{\beta};A)$, with at least probability $1-O(\eta)$, we have 
\begin{enumerate}
	\item $\|\bm{\beta}^{(t)} - \hat{\bm{\beta}}\|_2 \leq r$, where $\bm{\beta}^{(t)}$ is the iteration of the $t$-th GDP Newton's method and $\hat{\bm{\beta}}$ is the non-DP exact solution of the loss.
	\item $\|\nabla \mathcal{L}(\bm{\beta}^{(t)};\mathbf{A})\|_2 \leq \min \{\gamma r, \frac{\gamma^2}{C_h}\}$, and
	\item \begin{equation}\label{eq:loss_property}
		\frac{C_h}{2\gamma^2}\|\nabla\mathcal{L}(\bm{\beta}^{(t+1)};\mathbf{A})\|_2 \leq \big(\frac{C_h}{2\gamma^2} \|\nabla\mathcal{L}(\bm{\beta}^{(t)};\mathbf{A})\|_2\big)^2 + O(\frac{C_h\sigma_1\sqrt{d\log(Td/\eta)}}{\gamma^3}),
	\end{equation} where the second term is the GDP noisy bound and $\sigma_1 = \frac{2\sqrt{T}\Delta_h}{\mu n_{\scriptscriptstyle \!A}}$ is the parameter of $SG_d(\sigma_1^2)$ adding in the iterative Hessian matrices. 
\end{enumerate}
\end{lem}

\section{Proofs of the Main Theorems}
\label{append:main_thm_proofs}

\subsection{Proof of Theorem \ref{thm:secondmoment_bound}}
\label{append:proof_secondmoment_bound}
\begin{prf}
	The equation is equal to 
	\begin{equation}\label{eq:LU_equal}
		L\Sigma^{\scriptscriptstyle -\!1\!/\!2}\hat{\Sigma}\Sigma^{\scriptscriptstyle -\!1\!/\!2} \preceq \mathbf{I} \preceq U\Sigma^{\scriptscriptstyle -\!1\!/\!2}\hat{\Sigma}\Sigma^{\scriptscriptstyle -\!1\!/\!2}.
	\end{equation}
	Then, 
	\begin{equation*}
    \begin{aligned}
      \Sigma^{\scriptscriptstyle -\!1\!/\!2}\hat{\Sigma}\Sigma^{\scriptscriptstyle -\!1\!/\!2} &= \frac{1}{n}\sum_{i=1}^{n} (\Sigma^{\scriptscriptstyle -\!1\!/\!2}\bm{\upsilon}_i)(\Sigma^{\scriptscriptstyle -\!1\!/\!2}\bm{\upsilon}_i)^T \\
      &= \frac{1}{n}\sum_{i=1}^{n} \tilde{\bm{\upsilon}_i}\tilde{\bm{\upsilon}_i}^T \\
      &= \frac{1}{n}\tilde{\bm{\Upsilon}}^T\tilde{\bm{\Upsilon}}=\tilde{\Sigma},\\
    \end{aligned}
	\end{equation*}where $\tilde{\bm{\upsilon}}_i \sim subG(\mathbf{I})$, so $\tilde{\Sigma}=\frac{1}{n}\tilde{\bm{\Upsilon}}^T\tilde{\bm{\Upsilon}}$ is an estimation of $\mathbf{I}$. That means that 
	\begin{equation*}
		\begin{aligned}
			&\mathbf{I} \preceq U\Sigma^{\scriptscriptstyle -\!1\!/\!2}\hat{\Sigma}\Sigma^{\scriptscriptstyle -\!1\!/\!2}\Longleftrightarrow 1 \leq U\lambda_{min}(\tilde{\Sigma}), \\
			&L\Sigma^{\scriptscriptstyle -\!1\!/\!2}\hat{\Sigma}\Sigma^{\scriptscriptstyle -\!1\!/\!2} \preceq \mathbf{I} \Longleftrightarrow 1 \geq L\lambda_{max}(\tilde{\Sigma}).
		\end{aligned}
	\end{equation*} 
	From \textbf{Corollary \ref{cor:secondmoment}}, we know 
	\begin{equation*}
		\begin{aligned}
			&U = \Big(1 - O\Big(\sqrt{\frac{d}{n}} + \sqrt{\frac{\log(1/\eta)}{n}}\Big)\Big)^{-2} \Longrightarrow  1 \leq U\lambda_{min}(\tilde{\Sigma}),\\
			&L= \Big(1 + O\Big(\sqrt{\frac{d}{n}} + \sqrt{\frac{\log(1/\eta)}{n}}\Big)\Big)^{-2}\Longrightarrow  1 \geq L\lambda_{max}(\tilde{\Sigma}).
		\end{aligned}
	\end{equation*}
	$\hfill\square$
\end{prf}

\subsection{Proof of Theorem \ref{thm:DP_secmom}}
\begin{prf}
\textbf{1. Privacy}. Given two neighboring data sets with the different samples $\bm{\tilde{\xi}}$ and $\bm{\tilde{\xi}'}$, the $l_2$-sensitivity of the truncation second-moment matrix is 
\begin{equation*}
	\Big\|\frac{1}{n_\xi}\Big(\bm{\tilde{\xi}\tilde{\xi}}^T - \bm{\tilde{\xi}'\tilde{\xi}'}^T\Big)\Big\|_F \leq \frac{1}{n_\xi}\|\bm{\tilde{\xi}}\|_2^2 + \frac{1}{n_\xi}\|\bm{\tilde{\xi}'}\|_2^2 \leq \frac{2d(1 + \log(\frac{2n_{\xi}}{\eta}))}{n_\xi},
\end{equation*} where $\bm{\tilde{\xi}}$ and $\bm{\tilde{\xi}'}$ are the different samples in the neighboring data sets. From the Gaussian mechanism (\textbf{Theorem \ref{thm:gaussian-mechanism}}), we get the privacy guarantee.

\textbf{2. Noisy matrix bound}. It's a direct conclusion from the \textbf{Lemma \ref{lem:SGbound}}.$\hfill\square$
\end{prf}

\subsection{Proof of Theorem \ref{thm:public_inverse_error}}
\label{append:proof_public_inverse_error}

\begin{prf}
	\textbf{1. Recovery.} From \textbf{Corollary \ref{cor:utility_trunc}}, we know with at least probability $1 - O(\eta)$, the transformed data $\tilde{\bm{\xi}}_i$s are not truncated, namely, $\tilde{\bm{\xi}}_i = \hat{\Sigma}_{\scriptscriptstyle \upsilon}^{\scriptscriptstyle -\!1\!/\!2}\bm{\xi}_i,\ \ \forall i \in [n_\xi]$. That means 
	\begin{equation*}
		\hat{\Sigma}_{\scriptscriptstyle \upsilon}^{\scriptscriptstyle \!1\!/\!2}(\tilde{\Sigma}_{\scriptscriptstyle \!\xi} + \lambda \hat{\Sigma}_{\!\scriptscriptstyle \upsilon}^{\scriptscriptstyle -\!1})\hat{\Sigma}_{\scriptscriptstyle \upsilon}^{\scriptscriptstyle \!1\!/\!2} =  \hat{\Sigma}_{\scriptscriptstyle \upsilon}^{\scriptscriptstyle \!1\!/\!2}(\frac{1}{n_\xi} \sum_{i=1}^{n_\xi} \tilde{\bm{\xi}}_i \tilde{\bm{\xi}}_i^T + \lambda \hat{\Sigma}_{\scriptscriptstyle \upsilon}^{\scriptscriptstyle -\!1})\hat{\Sigma}_{\scriptscriptstyle \upsilon}^{\scriptscriptstyle \!1\!/\!2} = \hat{\Sigma}_{\scriptscriptstyle \xi} + \lambda \mathbf{I}. 
	\end{equation*}

	\textbf{2. Inverse error.} We discuss the following analysis under the transformed data $\tilde{\bm{\xi}}_i$s are not truncated, namely, the \textbf{Recovery} case holds. Especially, we will omit the probability $1 - O(\eta)$ from tools in the following analysis for simplicity.
	
	\textbf{2.1 Singular values bound.} From \textbf{Theorem \ref{thm:secondmoment_bound}}, the transformed data $\bm{\tilde{\xi}}_i \sim subG(\tilde{\Sigma})$ with $\tilde{\Sigma} = \hat{\Sigma}_{\scriptscriptstyle \upsilon}^{\scriptscriptstyle -\!1\!/\!2} \Sigma \hat{\Sigma}_{\scriptscriptstyle \upsilon}^{\scriptscriptstyle -\!1\!/\!2}$ and $\tilde{\Sigma}$ has the smallest and largest eigenvalues related with the number of public data $n_{\!\scriptscriptstyle \upsilon}$, 
	\begin{equation}\label{eq:singular_bound}
		\begin{aligned}
			\lambda_{min}(\tilde{\Sigma}) &\geqslant L = \frac{n_{\!\scriptscriptstyle \upsilon}}{(\sqrt{n_{\!\scriptscriptstyle \upsilon}} + O(\sqrt{d} + \sqrt{\log(\frac{1}{\eta})}))^2},\\
			\lambda_{max}(\tilde{\Sigma}) &\leqslant U = \frac{n_{\!\scriptscriptstyle \upsilon}}{(\sqrt{n_{\!\scriptscriptstyle \upsilon}} - O(\sqrt{d} + \sqrt{\log(\frac{1}{\eta})}))^2}.
		\end{aligned}
	\end{equation} The $\tilde{\Sigma}_{\scriptscriptstyle \!\xi}= \frac{\tilde{\bm{\Upxi}}^{\!\scriptscriptstyle T} \tilde{\bm{\Upxi}}}{n_{\!\scriptscriptstyle \xi}}$ is the estimation of $\tilde{\Sigma}$ and \textbf{Theorem \ref{lem:singular_bound}} shows
	\begin{equation*}
		\begin{aligned}
			\lambda_{min}(\tilde{\Sigma}_{\scriptscriptstyle \!\xi}) &\geqslant L(1 - O(\sqrt{\frac{d}{n_{\!\scriptscriptstyle \xi}}}+\sqrt{\frac{\log(1/\eta)}{n_{\!\scriptscriptstyle \xi}}}))^2,\\
			\lambda_{max}(\tilde{\Sigma}_{\scriptscriptstyle \!\xi}) &\leqslant U (1 + O(\sqrt{\frac{d}{n_{\!\scriptscriptstyle \xi}}}+\sqrt{\frac{\log(1/\eta)}{n_{\!\scriptscriptstyle \xi}}}))^2,
		\end{aligned}
	\end{equation*} where $L$ and $U$ are from the Eq.\eqref{eq:singular_bound}. Moreover, from \textbf{Lemma \ref{lem:welyinequal}} (Wely's inequality), we have
	\begin{equation*}
		\begin{aligned}
			\lambda_{min}(\tilde{\Sigma}_{\scriptscriptstyle \!\xi} + \lambda \hat{\Sigma}_{\!\scriptscriptstyle \upsilon}^{\scriptscriptstyle -\!1}) \geqslant \lambda_{min}(\tilde{\Sigma}_{\scriptscriptstyle \!\xi}) +  \lambda_{min}(\lambda \hat{\Sigma}_{\!\scriptscriptstyle \upsilon}^{\scriptscriptstyle -\!1}) &\geqslant L(1 - O(\sqrt{\frac{d}{n_{\!\scriptscriptstyle \xi}}}+\sqrt{\frac{\log(1/\eta)}{n_{\!\scriptscriptstyle \xi}}}))^2 + \lambda\|\hat{\Sigma}_{\!\scriptscriptstyle \upsilon}\|^{\scriptscriptstyle -\!1},\\
			\lambda_{max}(\tilde{\Sigma}_{\scriptscriptstyle \!\xi} + \lambda \hat{\Sigma}_{\!\scriptscriptstyle \upsilon}^{\scriptscriptstyle -\!1}) \leqslant \lambda_{max}(\tilde{\Sigma}_{\scriptscriptstyle \!\xi}) + \lambda_{max}(\lambda \hat{\Sigma}_{\!\scriptscriptstyle \upsilon}^{\scriptscriptstyle -\!1}) &\leqslant U (1 +O(\sqrt{\frac{d}{n_{\!\scriptscriptstyle \xi}}}+\sqrt{\frac{\log(1/\eta)}{n_{\!\scriptscriptstyle \xi}}}))^2  + \lambda\|\hat{\Sigma}_{\!\scriptscriptstyle \upsilon}^{\scriptscriptstyle -\!1}\|,
		\end{aligned}
	\end{equation*} where $\lambda$ is the regularization parameter.

	So we have the bound of $\|(\tilde{\Sigma}_{\scriptscriptstyle \!\xi} + \lambda \hat{\Sigma}_{\!\scriptscriptstyle \upsilon}^{\scriptscriptstyle -\!1})^{\scriptscriptstyle -\!1}\|_2 = \frac{1}{\lambda_{min}(\tilde{\Sigma}_{\scriptscriptstyle \!\xi} + \lambda \hat{\Sigma}_{\!\scriptscriptstyle \upsilon}^{\scriptscriptstyle -\!1})}$ and 
	\begin{equation}\label{eq:inv_sing_bound}
		\frac{1}{L(1 - O(\sqrt{\frac{d}{n_{\!\scriptscriptstyle \xi}}}+\sqrt{\frac{\log(1/\eta)}{n_{\!\scriptscriptstyle \xi}}}))^2 + \lambda\|\hat{\Sigma}_{\!\scriptscriptstyle \upsilon}\|^{\scriptscriptstyle \!-\!1}} \geqslant \frac{1}{\lambda_{min}(\tilde{\Sigma}_{\scriptscriptstyle \!\xi}) + \lambda\|\hat{\Sigma}_{\!\scriptscriptstyle \upsilon}\|^{\scriptscriptstyle \!-\!1}} \geqslant \|(\tilde{\Sigma}_{\scriptscriptstyle \!\xi} + \lambda\hat{\Sigma}_{\!\scriptscriptstyle \upsilon}^{\scriptscriptstyle -\!1})^{\scriptscriptstyle -\!1}\|_2.
	\end{equation}

	\textbf{2.2 Get Eq.\eqref{eq:inverse_trans_error}.} If $ \|(\tilde{\Sigma}_{\scriptscriptstyle \!\xi} + \lambda\hat{\Sigma}_{\!\scriptscriptstyle \upsilon}^{\scriptscriptstyle \!-\!1})^{\scriptscriptstyle -\!1}\| \|\mathbf{G}\| \leqslant \frac{ \|\mathbf{G}\|}{L(1 - O(\sqrt{\frac{d}{n_{\!\scriptscriptstyle \xi}}}+\sqrt{\frac{\log(1/\eta)}{n_{\!\scriptscriptstyle \xi}}}))^2 + \lambda\|\hat{\Sigma}_{\!\scriptscriptstyle \upsilon}\|^{\scriptscriptstyle \!-\!1}} \leqslant \frac{1}{2}$, \textbf{Lemma \ref{lem:inversebound}} gives the following bound
	\begin{equation*}
		\begin{aligned}
			\| (\tilde{\Sigma}_{\scriptscriptstyle \!\xi} + \lambda \hat{\Sigma}_{\!\scriptscriptstyle \upsilon}^{\scriptscriptstyle -\!1} + \mathbf{G})^{\scriptscriptstyle \!-\!1} - (\tilde{\Sigma}_{\scriptscriptstyle \!\xi} + \lambda \hat{\Sigma}_{\!\scriptscriptstyle \upsilon}^{\scriptscriptstyle -\!1})^{\scriptscriptstyle -\!1}\| &\leq \frac{\|\mathbf{G}\|\|(\tilde{\Sigma}_{\scriptscriptstyle \!\xi} + \lambda\hat{\Sigma}_{\!\scriptscriptstyle \upsilon}^{\scriptscriptstyle \!-\!1})^{\scriptscriptstyle -\!1}\|^2}{1 - \|(\tilde{\Sigma}_{\scriptscriptstyle \!\xi} +\lambda\hat{\Sigma}_{\!\scriptscriptstyle \upsilon}^{\scriptscriptstyle \!-\!1})^{\scriptscriptstyle -\!1}\|\|\mathbf{G}\|}\\
			& \leq 2\|\mathbf{G}\|\|(\tilde{\Sigma}_{\scriptscriptstyle \!\xi} + \lambda\hat{\Sigma}_{\!\scriptscriptstyle \upsilon}^{\scriptscriptstyle \!-\!1})^{\scriptscriptstyle -\!1}\|^2\\
			& \overset{(i)}{\leq} \frac{2\|\mathbf{G}\|}{L^2(1 -  O(\sqrt{\frac{d}{n_{\!\scriptscriptstyle \xi}}}+\sqrt{\frac{\log(1/\eta)}{n_{\!\scriptscriptstyle \xi}}}))^4 + \lambda^2\|\hat{\Sigma}_{\!\scriptscriptstyle \upsilon}\|^{\scriptscriptstyle \!-\!2}}, 		
		\end{aligned}
	\end{equation*}where $(i)$ is from the Eq.\eqref{eq:inv_sing_bound} and $\frac{1}{(a+b)^2} \leqslant \frac{1}{a^2 + b^2},a,b>0$. So we need to bound $\|\mathbf{G}\|$ so that the condition $\frac{ \|\mathbf{G}\|}{L(1 - O(\sqrt{\frac{d}{n_{\!\scriptscriptstyle \xi}}}+\sqrt{\frac{\log(1/\eta)}{n_{\!\scriptscriptstyle \xi}}}))^2 + \lambda\|\hat{\Sigma}_{\!\scriptscriptstyle \upsilon}\|^{\scriptscriptstyle \!-\!1}} \leqslant \frac{1}{2}$ holds. From \textbf{Theorem \ref{thm:DP_secmom}}, we know that the sufficient condition is the enough number of private data $n_{\!\scriptscriptstyle \xi}$ such that 
	\begin{equation*}
		\frac{\sqrt{d^3 \log(\frac{2d}{\eta})}(1 + \log(\frac{2n_{\!\scriptscriptstyle \xi}}{\eta}))}{(L(1 - O(\sqrt{\frac{d}{n_{\!\scriptscriptstyle \xi}}}+\sqrt{\frac{\log(1/\eta)}{n_{\!\scriptscriptstyle \xi}}}))^2 + \lambda\|\hat{\Sigma}_{\!\scriptscriptstyle \upsilon}\|^{\scriptscriptstyle \!-\!1})\cdot \mu \cdot n_{\! \scriptscriptstyle \xi}} \leqslant \frac{1}{2},
	\end{equation*} where $L = \frac{n_{\!\scriptscriptstyle \upsilon}}{(\sqrt{n_{\!\scriptscriptstyle \upsilon}} + O(\sqrt{d} + \sqrt{\log(\frac{1}{\eta})}))^2}$. Then, using \textbf{Theorem \ref{thm:DP_secmom}} again, we get 
	\begin{equation*}
		\| (\tilde{\Sigma}_{\scriptscriptstyle \!\xi} + \lambda \hat{\Sigma}_{\!\scriptscriptstyle \upsilon}^{\scriptscriptstyle -\!1} + \mathbf{G})^{\scriptscriptstyle \!-\!1} - (\tilde{\Sigma}_{\scriptscriptstyle \!\xi} + \lambda \hat{\Sigma}_{\!\scriptscriptstyle \upsilon}^{\scriptscriptstyle -\!1})^{\scriptscriptstyle -\!1}\| \leqslant \frac{\sqrt{d^3 \log(\frac{2d}{\eta})}}{\mu n_{\! \scriptscriptstyle \xi}} \cdot \frac{2(1 + \log(\frac{2n_{\!\scriptscriptstyle \xi}}{\eta}))}{L^2(1 -  O(\sqrt{\frac{d}{n_{\!\scriptscriptstyle \xi}}}+\sqrt{\frac{\log(1/\eta)}{n_{\!\scriptscriptstyle \xi}}}))^4 +  \lambda^2\|\hat{\Sigma}_{\!\scriptscriptstyle \upsilon}\|^{\scriptscriptstyle \!-\!2}}.
	\end{equation*}

	\textbf{2.3 Get Eq.\eqref{eq:inverse_orig_error}.} According to the \textbf{Recover} case, we have 
	\begin{equation*}
			\hat{\Sigma}_{\scriptscriptstyle \upsilon}^{\scriptscriptstyle -\!1\!/\!2}((\tilde{\Sigma}_{\scriptscriptstyle \!\xi} + \lambda \hat{\Sigma}_{\!\scriptscriptstyle \upsilon}^{\scriptscriptstyle -\!1} + \mathbf{G})^{\scriptscriptstyle \!-\!1} - (\tilde{\Sigma}_{\scriptscriptstyle \!\xi} + \lambda \hat{\Sigma}_{\!\scriptscriptstyle \upsilon}^{\scriptscriptstyle -\!1})^{\scriptscriptstyle -\!1})\hat{\Sigma}_{\scriptscriptstyle \upsilon}^{\scriptscriptstyle -\!1\!/\!2} =  \hat{\Sigma}_{\scriptscriptstyle \upsilon}^{\scriptscriptstyle -\!1\!/\!2}(\tilde{\Sigma}_{\scriptscriptstyle \!\xi} + \lambda \hat{\Sigma}_{\!\scriptscriptstyle \upsilon}^{\scriptscriptstyle -\!1} + \mathbf{G})^{\scriptscriptstyle \!-\!1}\hat{\Sigma}_{\scriptscriptstyle \upsilon}^{\scriptscriptstyle -\!1\!/\!2} - (\hat{\Sigma}_{\scriptscriptstyle \xi} + \lambda \mathbf{I})^{\scriptscriptstyle -\!1}.
	\end{equation*} So we have
	\begin{equation*}
		\begin{aligned}
			\|\hat{\Sigma}_{\scriptscriptstyle \upsilon}^{\scriptscriptstyle -\!1\!/\!2}(\tilde{\Sigma}_{\scriptscriptstyle \!\xi} + \lambda \hat{\Sigma}_{\!\scriptscriptstyle \upsilon}^{\scriptscriptstyle -\!1} + \mathbf{G})^{\scriptscriptstyle \!-\!1}\hat{\Sigma}_{\scriptscriptstyle \upsilon}^{\scriptscriptstyle -\!1\!/\!2} - (\hat{\Sigma}_{\scriptscriptstyle \xi} + \lambda \mathbf{I})^{\scriptscriptstyle -\!1}\| 
			&\leqslant \|\hat{\Sigma}_{\scriptscriptstyle \upsilon}^{\scriptscriptstyle -\!1\!/\!2}\|\|(\tilde{\Sigma}_{\scriptscriptstyle \!\xi} + \lambda \hat{\Sigma}_{\!\scriptscriptstyle \upsilon}^{\scriptscriptstyle -\!1} + \mathbf{G})^{\scriptscriptstyle \!-\!1} - (\tilde{\Sigma}_{\scriptscriptstyle \!\xi} + \lambda \hat{\Sigma}_{\!\scriptscriptstyle \upsilon}^{\scriptscriptstyle -\!1})^{\scriptscriptstyle -\!1}\|\|\hat{\Sigma}_{\scriptscriptstyle \upsilon}^{\scriptscriptstyle -\!1\!/\!2}\|\\
			& = \|\hat{\Sigma}_{\scriptscriptstyle \upsilon}^{\scriptscriptstyle -\!1}\|\|(\tilde{\Sigma}_{\scriptscriptstyle \!\xi} + \lambda \hat{\Sigma}_{\!\scriptscriptstyle \upsilon}^{\scriptscriptstyle -\!1} + \mathbf{G})^{\scriptscriptstyle \!-\!1} - (\tilde{\Sigma}_{\scriptscriptstyle \!\xi} + \lambda \hat{\Sigma}_{\!\scriptscriptstyle \upsilon}^{\scriptscriptstyle -\!1})^{\scriptscriptstyle -\!1}\|\\
			&\leqslant \frac{\sqrt{d^3 \log(\frac{2d}{\eta})}}{\mu n_{\! \scriptscriptstyle \xi}} \cdot \frac{\|\hat{\Sigma}_{\scriptscriptstyle \upsilon}^{\scriptscriptstyle -\!1}\|(1 + \log(\frac{2n_{\!\scriptscriptstyle \xi}}{\eta}))}{L^2(1 -  O(\sqrt{\frac{d}{n_{\!\scriptscriptstyle \xi}}}+\sqrt{\frac{\log(1/\eta)}{n_{\!\scriptscriptstyle \xi}}}))^4 +  \lambda^2\|\hat{\Sigma}_{\!\scriptscriptstyle \upsilon}\|^{\scriptscriptstyle \!-\!2}}.\\
		\end{aligned}
\end{equation*} $\hfill\square$
\end{prf}

\subsection{Proof of Theorem \ref{thm:nonpublic_inverse_error}}\label{append:proof_nonpublic_inverse_error}

\begin{prf}\textbf{1. Privacy}. Given two neighboring data sets with the different samples $\bm{\xi}$ and $\bm{\xi'}$, the $l_2$-sensitivity of the second-moment matrix is 
	\begin{equation*}
		\Big\|\frac{1}{n_\xi}\Big(\bm{\xi \xi}^T - \bm{\xi' \xi'}^T\Big)\Big\|_F \leq \frac{1}{n_\xi}\|\bm{\xi}\|_2^2 + \frac{1}{n_\xi}\|\bm{\xi'}\|_2^2 \leq \frac{2(\Tr(\Sigma) + d\log(\frac{n_{\scriptscriptstyle \!\xi}}{\eta}))}{n_\xi},
	\end{equation*} where $\bm{\xi}$ and $\bm{\xi'}$ are the different samples in the neighboring data sets. From the Gaussian mechanism (\textbf{Theorem \ref{thm:gaussian-mechanism}}), we get the privacy guarantee.

	\textbf{2. Bound error.}
	\textbf{Lemma \ref{lem:singular_bound}} shows 
	\begin{equation*}
		\hat{\Sigma}_{\scriptscriptstyle \xi} \geqslant \lambda_{\scriptscriptstyle min} (\Sigma)\big(1 - O(\sqrt{\frac{d}{n_{\!\scriptscriptstyle \xi}}}+\sqrt{\frac{\log(1/\eta)}{n_{\!\scriptscriptstyle \xi}}})\big)^2.
	\end{equation*} 
	According to \textbf{Lemma \ref{lem:welyinequal}}, we have 
	\begin{equation*}
		\|(\hat{\Sigma}_{\scriptscriptstyle \xi} + \lambda \mathbf{I})^{\scriptscriptstyle -\!1}\| \leqslant \frac{1}{\lambda_{\scriptscriptstyle min}(\hat{\Sigma}_{\scriptscriptstyle \xi}) + \lambda} \leqslant \frac{1}{\lambda_{\scriptscriptstyle min} (\Sigma)\big(1 - O(\sqrt{\frac{d}{n_{\!\scriptscriptstyle \xi}}}+\sqrt{\frac{\log(1/\eta)}{n_{\!\scriptscriptstyle \xi}}})\big)^2 + \lambda}.
	\end{equation*} If $\|(\hat{\Sigma}_{\scriptscriptstyle \xi} + \lambda \mathbf{I})^{\scriptscriptstyle -\!1}\|\|\mathbf{G}\| \leqslant \frac{\|\mathbf{G}\|}{\lambda_{\scriptscriptstyle min} (\Sigma)\big(1 - O(\sqrt{\frac{d}{n_{\!\scriptscriptstyle \xi}}}+\sqrt{\frac{\log(1/\eta)}{n_{\!\scriptscriptstyle \xi}}})\big)^2 + \lambda}\leqslant \frac{1}{2}$, \textbf{Lemma \ref{lem:inversebound}} gives the following bound
	\begin{equation*}
	\begin{aligned}
		\|(\hat{\Sigma}_{\scriptscriptstyle \xi} + \lambda \mathbf{I} + \mathbf{G})^{\scriptscriptstyle -\!1} - (\hat{\Sigma}_{\scriptscriptstyle \xi} + \lambda \mathbf{I})^{\scriptscriptstyle -\!1}\| &\leqslant \frac{\|\mathbf{G}\|\|(\hat{\Sigma}_{\scriptscriptstyle \xi} + \lambda \mathbf{I})^{\scriptscriptstyle -\!1}\|^2}{1 - \|(\hat{\Sigma}_{\scriptscriptstyle \xi} + \lambda \mathbf{I})^{\scriptscriptstyle -\!1}\|\|\mathbf{G}\|}\\
		&\leqslant 2\|\mathbf{G}\|\|(\hat{\Sigma}_{\scriptscriptstyle \xi} + \lambda \mathbf{I})^{\scriptscriptstyle -\!1}\|^2\\
		&\overset{(i)}{\leqslant} \frac{2\|\mathbf{G}\|}{\lambda^2_{\scriptscriptstyle min} (\Sigma)\big(1 - O(\sqrt{\frac{d}{n_{\!\scriptscriptstyle \xi}}}+\sqrt{\frac{\log(1/\eta)}{n_{\!\scriptscriptstyle \xi}}})\big)^4 + \lambda^2 },
	\end{aligned}
	\end{equation*} where $\frac{1}{(a+b)^2} \leqslant \frac{1}{a^2 + b^2},a,b>0$. So we need to bound $\|\mathbf{G}\|$ so that the condition $\frac{\|\mathbf{G}\|}{\lambda_{\scriptscriptstyle min} (\Sigma)\big(1 - O(\sqrt{\frac{d}{n_{\!\scriptscriptstyle \xi}}}+\sqrt{\frac{\log(1/\eta)}{n_{\!\scriptscriptstyle \xi}}})\big)^2 + \lambda}\leqslant ~\frac{1}{2}$ holds. From \textbf{Lemma \ref{lem:SGbound}}, we know that the sufficient condition is the enough number of private data $n_{\!\scriptscriptstyle \xi}$ such that
	\begin{equation*}
		\frac{\sqrt{d^{\scriptscriptstyle 3}\log(\frac{2d}{\eta})} (d^{\scriptscriptstyle \!-\!1}\!\Tr(\Sigma) + \log(\frac{n_{\scriptscriptstyle \!\xi}}{\eta}))}{\mu \cdot n_\xi \big(\lambda_{\scriptscriptstyle min} (\Sigma)\big(1 - O(\sqrt{\frac{d}{n_{\!\scriptscriptstyle \xi}}}+\sqrt{\frac{\log(1/\eta)}{n_{\!\scriptscriptstyle \xi}}})\big)^2 + \lambda \big)}\leqslant \frac{1}{2}.
	\end{equation*}Then, using \textbf{Lemma \ref{lem:SGbound}} again, we get
	\begin{equation*}
		\begin{aligned}
			\|(\hat{\Sigma}_{\scriptscriptstyle \xi} + \lambda \mathbf{I} + \mathbf{G})^{\scriptscriptstyle -\!1} - (\hat{\Sigma}_{\scriptscriptstyle \xi} + \lambda \mathbf{I})^{\scriptscriptstyle -\!1}\| &\leqslant \frac{2\|\mathbf{G}\|}{\lambda^2_{\scriptscriptstyle min} (\Sigma)\big(1 - O(\sqrt{\frac{d}{n_{\!\scriptscriptstyle \xi}}}+\sqrt{\frac{\log(1/\eta)}{n_{\!\scriptscriptstyle \xi}}})\big)^4 + \lambda^2}\\
			&\leqslant \frac{\sqrt{d^3 \log(\frac{2d}{\eta})}}{\mu n_{\! \scriptscriptstyle \xi}} \cdot \frac{ 2(d^{\scriptscriptstyle \!-\!1}\!\Tr(\Sigma) + \log(\frac{n_{\scriptscriptstyle \!\xi}}{\eta}))}{\lambda^2_{\scriptscriptstyle min} (\Sigma)\big(1 - O(\sqrt{\frac{d}{n_{\!\scriptscriptstyle \xi}}}+\sqrt{\frac{\log(1/\eta)}{n_{\!\scriptscriptstyle \xi}}})\big)^4 +  \lambda^2}.\\
		\end{aligned}
	\end{equation*} $\hfill\square$
\end{prf}

\subsection{Proof of Theorem \ref{thm:beta_transform}}
\label{append:proof_beta_transform}
\begin{prf}
Consider the event
\begin{equation*}
	\mathcal{E} = \Big\{\tilde{\mathbf{A}} \ \ \text{and} \ \  \tilde{\mathbf{y}}_{\!\!\scriptscriptstyle A} \ \ \text{are not truncated.}\Big\}.
\end{equation*} Due to the assumption of linear model, $\mathbf{y}_{\!\!\scriptscriptstyle A}$ is the sum of the sub-Gaussian randoms, so it is a sub-Gaussian random and $\tilde{y}_{\!\scriptscriptstyle A_i} \leqslant \sqrt{1 + \log(\frac{2n_{\!\scriptscriptstyle A}}{\eta})},\ w.p.\ 1 -O(\eta)$, from \textbf{Corollary \ref{cor:utility_trunc}}. We get the conclusion
\begin{equation*}
	\mathbb{P}(\mathcal{E}) \geq 1 - O(\eta).
\end{equation*}And, we will discuss the following analysis under the event. We know the original ridge regression is solved by the loss function
\begin{equation*}
	\mathcal{L}(\bm{\beta};\mathbf{A},\mathbf{y}_{\! \! \scriptscriptstyle A}) = \frac{1}{2n} ||\mathbf{y}_{\! \! \scriptscriptstyle A} - \mathbf{A}\bm{\beta}||_2^2 + \frac{\lambda}{2} ||\bm{\beta}||_2^2.
\end{equation*} Taking the $\hat{\Sigma}_{\!\scriptscriptstyle B}^{\scriptscriptstyle -\!1\!/\!2}$ into the loss function and replace $\mathbf{y}_{\! \! \scriptscriptstyle A}$ as $\tilde{\mathbf{y}}_{\! \! \scriptscriptstyle A}$ , we have
\begin{equation*}
	\mathcal{L}(\hat{\Sigma}_{\!\scriptscriptstyle B}^{\scriptscriptstyle -\!1\!/\!2}\bm{\beta};\mathbf{A},\tilde{\mathbf{y}}_{\! \! \scriptscriptstyle A}) = \frac{1}{2n} ||\tilde{\mathbf{y}}_{\! \! \scriptscriptstyle A} - \mathbf{A}\hat{\Sigma}_{\!\scriptscriptstyle B}^{\scriptscriptstyle -\!1\!/\!2}\bm{\beta}||_2^2 + \frac{\lambda}{2} ||\hat{\Sigma}_{\!\scriptscriptstyle B}^{\scriptscriptstyle -\!1\!/\!2}\bm{\beta}||_2^2= \frac{1}{2n} ||\tilde{\mathbf{y}}_{\! \! \scriptscriptstyle A} - \tilde{\mathbf{A}}\bm{\beta}||_2^2 + \frac{\lambda}{2} ||\hat{\Sigma}_{\!\scriptscriptstyle B}^{\scriptscriptstyle -\!1\!/\!2}\bm{\beta}||_2^2,
\end{equation*} where $\tilde{\mathbf{A}}$ is the transformed data. Then, the parameter $\bm{\tilde{\beta}}$ minimizes this formula, namely,
\begin{equation*}
	\begin{aligned}
		\bm{\tilde{\beta}} & = \arg\min\limits_{\bm{\beta}}\mathcal{L}(\hat{\Sigma}_{\!\scriptscriptstyle B}^{\scriptscriptstyle -\!1\!/\!2}\bm{\beta};\mathbf{A},\tilde{\mathbf{y}}_{\! \! \scriptscriptstyle A})\\
		&\overset{(i)}{=}  \Big(\frac{\tilde{\mathbf{A}}^{\!\scriptscriptstyle T} \tilde{\mathbf{A}}}{n_{\!\scriptscriptstyle A}} + \lambda \hat{\Sigma}_{\!\scriptscriptstyle B}^{\scriptscriptstyle -\!1}\Big)^{\!\!\scriptscriptstyle -\!1}\Big(\frac{\tilde{\mathbf{A}}^{\!\scriptscriptstyle T} \tilde{\mathbf{y}}_{\!\scriptscriptstyle A}}{n_{\!\scriptscriptstyle A}}\Big).\\
	\end{aligned}
\end{equation*} Concluding from the equation $(i)$, we have 
\begin{equation*}
	\begin{aligned}
		\bm{\tilde{\beta}}&\overset{(i)}{=}  \Big(\frac{\tilde{\mathbf{A}}^{\!\scriptscriptstyle T} \tilde{\mathbf{A}}}{n_{\!\scriptscriptstyle A}} + \lambda \hat{\Sigma}_{\!\scriptscriptstyle B}^{\scriptscriptstyle -\!1}\Big)^{\!\!\scriptscriptstyle -\!1}\Big(\frac{\tilde{\mathbf{A}}^{\!\scriptscriptstyle T} \tilde{\mathbf{y}}_{\!\scriptscriptstyle A}}{n_{\!\scriptscriptstyle A}}\Big)\\
		& =  \Big(\frac{\hat{\Sigma}_{\!\scriptscriptstyle B}^{\scriptscriptstyle -\!1\!/\!2}\mathbf{A}^{\!\scriptscriptstyle T} \mathbf{A}\hat{\Sigma}_{\!\scriptscriptstyle B}^{\scriptscriptstyle -\!1\!/\!2}}{n_{\!\scriptscriptstyle A}} + \lambda \hat{\Sigma}_{\!\scriptscriptstyle B}^{\scriptscriptstyle -\!1}\Big)^{\!\!\scriptscriptstyle -\!1}\Big(\frac{\hat{\Sigma}_{\!\scriptscriptstyle B}^{\scriptscriptstyle -\!1\!/\!2}}{\hat{\sigma}_{\!\scriptscriptstyle B}}\frac{\mathbf{A}^{\!\scriptscriptstyle T} \mathbf{y}_{\!\scriptscriptstyle A}}{n_{\!\scriptscriptstyle A}}\Big)\\
		& =  \frac{\hat{\Sigma}_{\!\scriptscriptstyle B}^{\scriptscriptstyle \!1\!/\!2}}{\hat{\sigma}_{\!\scriptscriptstyle B}} \Big(\frac{\mathbf{A}^{\!\scriptscriptstyle T} \mathbf{A} }{n_{\!\scriptscriptstyle A}} + \lambda \mathbf{I}\Big)^{\!\!\scriptscriptstyle -\!1}\Big(\frac{\mathbf{A}^{\!\scriptscriptstyle T} \mathbf{y}_{\!\scriptscriptstyle A}}{n_{\!\scriptscriptstyle A}}\Big)\\
		&\overset{(ii)}{=} \frac{\hat{\Sigma}_{\!\scriptscriptstyle B}^{\scriptscriptstyle \!1\!/\!2}}{\hat{\sigma}_{\!\scriptscriptstyle B}}\bm{\hat{\beta}},
	\end{aligned}
\end{equation*}where $\hat{\sigma}_{\!\scriptscriptstyle B} =  \sqrt{\frac{1}{n_{\!\scriptscriptstyle B}}\sum_{i=1}^{n_{\!\scriptscriptstyle B}} y_{\!\scriptscriptstyle B_i}^2}$. Combining the event $\mathcal{E}$ and the equation $(ii)$, we get the conclusion
\begin{equation*}
	\bm{\hat{\beta}} = \hat{\sigma}_{\!\scriptscriptstyle B}\hat{\Sigma}_{\!\scriptscriptstyle B}^{\scriptscriptstyle -\!1\!/\!2}\bm{\tilde{\beta}},
\end{equation*} with at least probability $1 - O(\eta)$. $\hfill\square$
\end{prf}

\subsection{Proof of Theorem \ref{thm:DP-PMTRR}}\label{append:proof_DP-PMTRR}

\begin{prf}
	\textbf{1. Privacy.} For convenience, we denote $\tilde{\Sigma}_{\scriptscriptstyle \!A} = \frac{\tilde{\mathbf{A}}^{\!\scriptscriptstyle T} \tilde{\mathbf{A}}}{n_{\!\scriptscriptstyle A}}$, $\tilde{\Sigma}_{\!Ay} = \frac{\tilde{\mathbf{A}}^{\!\scriptscriptstyle T} \tilde{\mathbf{y}}_{\!\scriptscriptstyle A}}{n_{\!\scriptscriptstyle A}}$ and $\hat{\Sigma}_{\!\scriptscriptstyle B} = \frac{\mathbf{B}^{\scriptscriptstyle T}\mathbf{B}}{n_{\scriptscriptstyle B}}$.
	Then, $\tilde{\Sigma}_{\!\scriptscriptstyle A} + \lambda \hat{\Sigma}_{\!\scriptscriptstyle B}^{\scriptscriptstyle -\!1} + \mathbf{G}$ satisfies $\mu$-GDP from \textbf{Theorem \ref{thm:DP_secmom}}, where $\hat{\Sigma}_{\!\scriptscriptstyle B}^{\scriptscriptstyle -\!1}$ is unrelated with privacy. We mainly discuss the second term $\tilde{\Sigma}_{\!Ay} + \mathbf{g}$. The sensitivity of $\tilde{\Sigma}_{\!Ay}$ is 
	\begin{equation*}
		\begin{aligned}
			\Delta_{\tilde{\Sigma}_{\!Ay}}&=\max\limits_{(\tilde{\mathbf{A}},\tilde{\mathbf{y}}_{\!\!A}) \sim (\tilde{\mathbf{A}'},\tilde{\mathbf{y}}'_{\!\!A'})}\frac{1}{n_{\scriptscriptstyle \!A}}\|\tilde{\mathbf{A}}^{\scriptscriptstyle \!T}\tilde{\mathbf{y}}_{\scriptscriptstyle \!\!A} - \tilde{\mathbf{A}}'^{\scriptscriptstyle \!T}\tilde{\mathbf{y}}'_{\scriptscriptstyle \!\!A'}\|_2 \\
			&  \overset{(i)}{=} \max\limits_{(\tilde{A},\tilde{y}_{\scriptscriptstyle \!A}) \atop (\tilde{A}',\tilde{y}'_{\!\!A'})} \frac{1}{n_{\scriptscriptstyle \!A}}\|\tilde{A}^{\scriptscriptstyle \!T} \tilde{y}_{\scriptscriptstyle \!A} - \tilde{A}'^{\scriptscriptstyle \!T} \tilde{y}'_{\scriptscriptstyle \!A'}\|_2\\
			& \leq \max\limits_{(\tilde{A},\tilde{y}_{\scriptscriptstyle \!A})}\frac{2}{n_{\scriptscriptstyle \!A}}\|\tilde{A}^{\scriptscriptstyle \!T} \tilde{y}_{\scriptscriptstyle \!A}\|_2\\
			& \leq  \max\limits_{(\tilde{A},\tilde{y}_{\scriptscriptstyle \!A})}\frac{2}{n_{\scriptscriptstyle \!A}}\|\tilde{A}^{\scriptscriptstyle \!T}\|_2\|\tilde{y}_{\scriptscriptstyle \!A}\|_2\\
			& \leq \frac{2\sqrt{d}}{n_{\scriptscriptstyle \!A}} (1 + \log(\frac{2n_{\scriptscriptstyle \!A}}{\eta})),\\
		\end{aligned} 
	\end{equation*}where $(\tilde{A},\tilde{y}_{\scriptscriptstyle \!A})$ and  $(\tilde{A}',\tilde{y}'_{\scriptscriptstyle \!\!A'})$ are the different samples in $(i)$. \textbf{Theorem \ref{thm:composition}} shows that $\bm{\tilde{\beta}}^{\scriptscriptstyle D\!P}$ satisfies $\sqrt{2}\mu$-GDP.

	\textbf{2. Accuracy.} We discuss the accuracy of the DP-PDTMRR under \textbf{Theorem \ref{thm:beta_transform}}. We firstly analyze the error between $\bm{\tilde{\beta}}^{\scriptscriptstyle D\!P}$ and $\bm{\tilde{\beta}}$,
	\begin{equation*}
		\begin{aligned}
			\|\bm{\tilde{\beta}}^{\scriptscriptstyle D\!P} - \bm{\tilde{\beta}}\|_2 &= \|(\tilde{\Sigma}_{\scriptscriptstyle \!A} + \lambda \hat{\Sigma}_{\!\scriptscriptstyle B}^{\scriptscriptstyle -\!1} + \mathbf{G})^{\scriptscriptstyle \!-\!1}(\tilde{\Sigma}_{\scriptscriptstyle \!Ay} + \mathbf{g}) - (\tilde{\Sigma}_{\scriptscriptstyle \!A}+\lambda \hat{\Sigma}_{\!\scriptscriptstyle B}^{\scriptscriptstyle -\!1})^{\scriptscriptstyle -\!1}\tilde{\Sigma}_{\scriptscriptstyle \!Ay}\|_2\\
			& \leqslant \underbrace{\| (\tilde{\Sigma}_{\scriptscriptstyle \!A} + \lambda \hat{\Sigma}_{\!\scriptscriptstyle B}^{\scriptscriptstyle -\!1} + \mathbf{G})^{\scriptscriptstyle \!-\!1} - (\tilde{\Sigma}_{\scriptscriptstyle \!A}+\lambda \hat{\Sigma}_{\!\scriptscriptstyle B}^{\scriptscriptstyle -\!1})^{\scriptscriptstyle -\!1}\|_2}_{(\ast)}\|\tilde{\Sigma}_{\scriptscriptstyle \!Ay}\|_2 + \underbrace{\|(\tilde{\Sigma}_{\scriptscriptstyle \!A} + \lambda \hat{\Sigma}_{\!\scriptscriptstyle B}^{\scriptscriptstyle -\!1} +\mathbf{G})^{\scriptscriptstyle \!-\!1}\|_2\|\mathbf{g}\|_2}_{(\diamond)},
		\end{aligned}
	\end{equation*}where $\tilde{\Sigma}_{\scriptscriptstyle \!A} = \frac{\tilde{\mathbf{A}}^{\!\scriptscriptstyle T} \tilde{\mathbf{A}}}{n_{\!\scriptscriptstyle A}}$ and $\tilde{\Sigma}_{\!Ay} = \frac{\tilde{\mathbf{A}}^{\!\scriptscriptstyle T} \tilde{\mathbf{y}}_{\!\scriptscriptstyle A}}{n_{\!\scriptscriptstyle A}}$.

	\textbf{2.1 {Bound $(\ast)$.}} Using \textbf{Theorem \ref{thm:public_inverse_error}}, we get 
	\begin{equation}\label{eq:ast}
		(\ast) \leqslant O\Big(\frac{\sqrt{d^3 \log(\frac{2d}{\eta})}}{\mu n_{\! \scriptscriptstyle A}} \cdot \frac{(1 + \log(\frac{2n_{\!\scriptscriptstyle A}}{\eta}))}{L^2(1 -  O(\sqrt{\frac{d}{n_{\!\scriptscriptstyle A}}}+\sqrt{\frac{\log(1/\eta)}{n_{\!\scriptscriptstyle A}}}))^4 +  \lambda^2\|\hat{\Sigma}_{\!\scriptscriptstyle B}\|^{\scriptscriptstyle \!-\!2}}\Big).
	\end{equation}

	\textbf{2.2 Bound $(\diamond)$.} From \textbf{Lemma \ref{lem:inversebound}}, we have 
	\begin{equation*}
		(\diamond) \leqslant 2\|(\tilde{\Sigma}_{\scriptscriptstyle \!A} + \lambda \hat{\Sigma}_{\!\scriptscriptstyle B}^{\scriptscriptstyle -\!1})^{\scriptscriptstyle \!-\!1}\|\|\mathbf{g}\|.
	\end{equation*}Since $\mathbf{g} \sim \mathcal{N}(0,\sigma_2^2\mathbf{I})$ with $\sigma_2 = \frac{2\sqrt{d}\left(1 + \log\left(\frac{2n_{\scriptscriptstyle \!A}}{\eta}\right)\right)}{\mu \cdot n_{\scriptscriptstyle \!A}}$, and by applying \textbf{Lemma~\ref{lem:gaussiannorm}}, we obtain
	\begin{equation*}
		\begin{aligned}
			\|\mathbf{g}\| & \leqslant \sigma_2 (\sqrt{d} + \sqrt{\log(\frac{1}{\eta})})\\
			&\leqslant O\Big(\frac{(d + \sqrt{d \log(\frac{1}{\eta})})(1 + \log(\frac{2n_{\!\scriptscriptstyle A}}{\eta}))}{\mu \cdot n_{\!\scriptscriptstyle A}}\Big).
		\end{aligned}
	\end{equation*} As the proof in \textbf{Theorem \ref{thm:public_inverse_error}}, we have
	\begin{equation*}
		\frac{1}{L(1 - O(\sqrt{\frac{d}{n_{\!\scriptscriptstyle A}}}+\sqrt{\frac{\log(1/\eta)}{n_{\!\scriptscriptstyle A}}}))^2 + \lambda\|\hat{\Sigma}_{\!\scriptscriptstyle B}\|^{\scriptscriptstyle \!-\!1}}  \geqslant \|(\tilde{\Sigma}_{\scriptscriptstyle \!A} + \lambda\hat{\Sigma}_{\!\scriptscriptstyle B}^{\scriptscriptstyle -\!1})^{\scriptscriptstyle -\!1}\|,
	\end{equation*} where $L = \frac{n_{\scriptscriptstyle \!B}}{(\sqrt{n_{\scriptscriptstyle \!B}} + O(\sqrt{d} + \sqrt{\log(\frac{1}{\eta})}))^2}$. So, we get 
	\begin{equation}\label{eq:diamond}
		(\diamond) \leqslant O\Big(\frac{(d + \sqrt{d} \log(\frac{1}{\eta}))}{\mu \cdot n_{\!\scriptscriptstyle A}} \cdot \frac{(1 + \log(\frac{2n_{\!\scriptscriptstyle A}}{\eta}))}{L(1 - O(\sqrt{\frac{d}{n_{\!\scriptscriptstyle A}}}+\sqrt{\frac{\log(1/\eta)}{n_{\!\scriptscriptstyle A}}}))^2 + \lambda\|\hat{\Sigma}_{\!\scriptscriptstyle B}\|^{\scriptscriptstyle \!-\!1}}\Big).
	\end{equation}Note that the order $(\diamond)$ is smaller than $(\ast)$, so we can ignore the order of $(\diamond)$ in the final bound.

	\textbf{2.3 Bound $\tilde{\Sigma}_{\scriptscriptstyle \!Ay}$.} Because of the compatibility of the matrix norm, we have 
	\begin{equation}\label{eq:Sigma_Ay}
		\|\tilde{\Sigma}_{\scriptscriptstyle \!Ay}\| \leqslant \hat{\sigma}_{\scriptscriptstyle \!B}^{\scriptscriptstyle -\!1} \|\hat{\Sigma}_{\scriptscriptstyle \!B}^{\scriptscriptstyle -\!1\!/\!2}\| \|\hat{\Sigma}_{\scriptscriptstyle \!Ay}\| \overset{(i)}{\leqslant} \hat{\sigma}_{\scriptscriptstyle \!B}^{\scriptscriptstyle -\!1} \|\hat{\Sigma}_{\scriptscriptstyle \!B}^{\scriptscriptstyle -\!1\!/\!2}\|\cdot O(\|\Sigma\| \|\bm{\beta}\|),
	\end{equation}where $(i)$ is due to $\|\frac{\mathbf{A}^{\scriptscriptstyle \!\!T}\!\mathbf{A}}{n_{\scriptscriptstyle \!A}}\| \to \|\Sigma\|$ and $\|\frac{\mathbf{A}^{\scriptscriptstyle \!\!T}\!\mathbf{\bm{\upepsilon}}}{n_{\scriptscriptstyle \!A}}\| \to o(1)$, as $n_{\scriptscriptstyle A} \to \infty$. 
	
	\textbf{2.4 Final bound.} Combining the equations \eqref{eq:ast}, \eqref{eq:diamond} and \eqref{eq:Sigma_Ay}, we have
	\begin{equation*}
		\|\bm{\tilde{\beta}}^{\scriptscriptstyle DP} - \bm{\tilde{\beta}}\| \leqslant O\Big(\frac{\sqrt{d^3 \log(\frac{2d}{\eta})}}{\mu n_{\! \scriptscriptstyle A}} \cdot \frac{\hat{\sigma}_{\scriptscriptstyle \!B}^{\scriptscriptstyle -\!1} \|\hat{\Sigma}_{\scriptscriptstyle \!B}^{\scriptscriptstyle -\!1\!/\!2}\| \|\Sigma\| \|\bm{\beta}\| (1 + \log(\frac{2n_{\!\scriptscriptstyle A}}{\eta}))}{L^2(1 -  O(\sqrt{\frac{d}{n_{\!\scriptscriptstyle A}}}+\sqrt{\frac{\log(1/\eta)}{n_{\!\scriptscriptstyle A}}}))^4 +  \lambda^2\|\hat{\Sigma}_{\!\scriptscriptstyle B}\|^{\scriptscriptstyle \!-\!2}}\Big),
	\end{equation*} with at least probability $1 - O(\eta)$.

	\textbf{3. Original Bound.}  From \textbf{Theorem \ref{thm:beta_transform}}, the output $\bm{\bar{\beta}}^{\scriptscriptstyle D\!P} =\hat{\sigma}_{\!\scriptscriptstyle B}\cdot\hat{\Sigma}_{\!\scriptscriptstyle B}^{\scriptscriptstyle -\!1\!/\!2}\cdot\bm{\tilde{\beta}}^{\scriptscriptstyle D\!P}$ is bounded by, with at least probability $1 - O(\eta)$,
	\begin{equation*}
		\|\hat{\sigma}_{\!\scriptscriptstyle B}\hat{\Sigma}_{\!\scriptscriptstyle B}^{\scriptscriptstyle -\!1\!/\!2}\bm{\tilde{\beta}}^{\scriptscriptstyle DP} - \bm{\hat{\beta}}\| \leqslant O\Big(\frac{\sqrt{d^3 \log(\frac{2d}{\eta})}}{\mu n_{\! \scriptscriptstyle A}} \cdot \frac{\|\hat{\Sigma}_{\scriptscriptstyle \!B}^{\scriptscriptstyle -\!1}\| \|\Sigma\| \|\bm{\beta}\| (1 + \log(\frac{2n_{\!\scriptscriptstyle A}}{\eta}))}{L^2(1 -  O(\sqrt{\frac{d}{n_{\!\scriptscriptstyle A}}}+\sqrt{\frac{\log(1/\eta)}{n_{\!\scriptscriptstyle A}}}))^4 +  \lambda^2\|\hat{\Sigma}_{\!\scriptscriptstyle B}\|^{\scriptscriptstyle \!-\!2}}\Big).
	\end{equation*}
$\hfill\square$
\end{prf}

\subsection{Proof of Theorem \ref{thm:DP-RR}}\label{append:proof_DP-RR}

\begin{prf}
	\textbf{1. Privacy.} From \textbf{Theorem \ref{thm:DP_secmom}}, $\hat{\Sigma}_{\scriptscriptstyle \!A} + \lambda \mathbf{I} + \mathbf{G}$ satisfies $\mu$-GDP. Then, we consider the sensitivity of the term $\hat{\Sigma}_{\scriptscriptstyle \!Ay}$
	\begin{equation*}
		\begin{aligned}
			\Delta_{\hat{\Sigma}_{\!Ay}}&=\max\limits_{(\mathbf{A},\mathbf{y}_{\!\!A}) \sim (\mathbf{A}',\mathbf{y}'_{\!\!A'})}\frac{1}{n_{\scriptscriptstyle \!A}}\|\mathbf{A}^{\scriptscriptstyle \!T}\mathbf{y}_{\scriptscriptstyle \!\!A} - \mathbf{A}'^{\scriptscriptstyle \!T}\mathbf{y}'_{\scriptscriptstyle \!\!A'}\|_2 \\
			&  \overset{(i)}{=} \max\limits_{(A,y_{\scriptscriptstyle \!A}) \atop (A,y'_{\!\!A'})} \frac{1}{n_{\scriptscriptstyle \!A}}\|\hat{A}^{\scriptscriptstyle \!T} y_{\scriptscriptstyle \!A} - A'^{\scriptscriptstyle \!T} y'_{\scriptscriptstyle \!A'}\|_2\\
			& \leq  \max\limits_{(A,y_{\scriptscriptstyle \!A})}\frac{2}{n_{\scriptscriptstyle \!A}}\|A^{\scriptscriptstyle \!T}\|_2\|y_{\scriptscriptstyle \!A}\|_2\\
			& \leq \frac{2(\Tr(\Sigma) + d\log(\frac{n_{\scriptscriptstyle \!A}}{\eta}))\|\bm{\beta}\|}{n_{\scriptscriptstyle \!A}}, \\
		\end{aligned} 
	\end{equation*}so we have $\hat{\Sigma}_{\scriptscriptstyle \!Ay} + \mathbf{g}$ is $\mu$-GDP. Through the composition property, we get $\hat{\bm{\beta}}_{\scriptscriptstyle \! D\!P}$ satisfies $\sqrt{2}\mu$-GDP.

	\textbf{2. Bound.} We decompose the $\|\hat{\bm{\beta}}_{\scriptscriptstyle \! D\!P}  - \hat{\bm{\beta}}\| $ as two terms $(\ast)$ and $(\diamond)$,
	\begin{equation*}
		\begin{aligned}
			\|\hat{\bm{\beta}}_{\scriptscriptstyle \! D\!P}  - \hat{\bm{\beta}}\| 
			&\leqslant \underbrace{\| (\hat{\Sigma}_{\scriptscriptstyle \!A} + \lambda\mathbf{I} + \mathbf{G})^{\scriptscriptstyle \!-\!1} - (\hat{\Sigma}_{\scriptscriptstyle \!A}+\lambda \mathbf{I})^{\scriptscriptstyle -\!1}\|_2}_{(\ast)}\|\hat{\Sigma}_{\scriptscriptstyle \!Ay}\| + \underbrace{\|(\hat{\Sigma}_{\scriptscriptstyle \!A} + \lambda \mathbf{I} +\mathbf{G})^{\scriptscriptstyle \!-\!1}\|_2\|\mathbf{g}\|_2}_{(\diamond)}.\\
		\end{aligned}
	\end{equation*}From \textbf{Theorem \ref{thm:nonpublic_inverse_error}}, we have 
	\begin{equation*}
		(\ast) \leqslant \frac{\sqrt{d^3 \log(\frac{2d}{\eta})}}{\mu n_{\! \scriptscriptstyle \! A}} \cdot \frac{ d^{\scriptscriptstyle \!-\!1}\!\Tr(\Sigma) + \log(\frac{n_{\scriptscriptstyle \!\xi}}{\eta})}{\lambda^2_{\scriptscriptstyle min} (\Sigma)\big(1 - O(\sqrt{\frac{d}{n_{\!\scriptscriptstyle \! A}}}+\sqrt{\frac{\log(1/\eta)}{n_{\!\scriptscriptstyle \! A}}})\big)^4 +  \lambda^2}.
	\end{equation*}From \textbf{Lemma \ref{lem:inversebound}}, we have 
	\begin{equation*}
		(\diamond) \leqslant 2\|(\hat{\Sigma}_{\scriptscriptstyle \!A} + \lambda \mathbf{I})^{\scriptscriptstyle \!-\!1}\|\|\mathbf{g}\|.
	\end{equation*}Since $\mathbf{g} \sim \mathcal{N}(0,\sigmoid^2_2),\ \sigma_2 =  \frac{2(\Tr(\Sigma) + d\log(\frac{n_{\scriptscriptstyle \!A}}{\eta}))\|\bm{\beta}\|}{\mu \cdot n_{\scriptscriptstyle \!A}}$ and, by applying \textbf{Lemma \ref{lem:gaussiannorm}}, we have
	\begin{equation*}
		\begin{aligned}
			\|\mathbf{g}\| & \leqslant \sigma_2 (\sqrt{d} + \sqrt{\log(\frac{1}{\eta})})\\
			&\leqslant O\Big(  \frac{(\sqrt{d^3} + d\sqrt{\log(\frac{1}{\eta})})(d^{\scriptscriptstyle \!-\!1}\!\Tr(\Sigma) + \log(\frac{n_{\scriptscriptstyle \!\xi}}{\eta}))\|\bm{\beta}\|}{\mu \cdot n_{\scriptscriptstyle \!A}}\Big).
		\end{aligned}
	\end{equation*}And the term $\|\hat{\Sigma}_{\scriptscriptstyle \!Ay}\| \leqslant O(\|\Sigma\|\|\bm{\beta}\|)$. 
	
	\textbf{3. Conclusion.} So combining these inequalities, we have 
	\begin{equation*}
		\|\hat{\bm{\beta}}_{\scriptscriptstyle \! D\!P}  - \hat{\bm{\beta}}\| \leqslant O\Big(\frac{\sqrt{d^3 \log(\frac{2d}{\eta})}}{\mu n_{\! \scriptscriptstyle \! A}} \cdot \frac{ \|\Sigma\| \|\bm{\beta}\|(d^{\scriptscriptstyle \!-\!1}\!\Tr(\Sigma) + \log(\frac{n_{\scriptscriptstyle \!\xi}}{\eta}))}{\lambda^2_{\scriptscriptstyle min} (\Sigma)\big(1 - O(\sqrt{\frac{d}{n_{\!\scriptscriptstyle \! A}}}+\sqrt{\frac{\log(1/\eta)}{n_{\!\scriptscriptstyle \! A}}})\big)^4 +  \lambda^2}\Big).
	\end{equation*}

\end{prf}

\subsection{Proof of Theorem \ref{thm:Equivalent estimation}}\label{append:proof_Equivalent estimation}
\begin{prf}
	\textbf{1. Equivalent Estimation.} We get the equivalent conclusion by the following steps:
	\begin{equation*}
		\begin{aligned}
			\tilde{\mathcal{L}}(\bm{\beta};\mathbf{\tilde{A}}) &= -\frac{1}{n_{\scriptscriptstyle \!A}}\sum_{i=1}^{n_{\scriptscriptstyle \!A}} l_i(\tilde{A}_i^T \bm{\beta}) + \frac{\lambda}{2}\|\hat{\Sigma}_{\scriptscriptstyle \!B}^{\scriptscriptstyle -\!1\!/\!2}\bm{\beta}\|_2^2\\
			&= -\frac{1}{n_{\scriptscriptstyle \!A}}\sum_{i=1}^{n_{\scriptscriptstyle \!A}} l_i(A_i^{\scriptscriptstyle \!T}\hat{\Sigma}_{\scriptscriptstyle \!B}^{\scriptscriptstyle -\!1\!/\!2}\bm{\beta}) + \frac{\lambda}{2}\|\hat{\Sigma}_{\scriptscriptstyle \!B}^{\scriptscriptstyle -\!1\!/\!2}\bm{\beta}\|_2^2\\
			&= \mathcal{L}(\hat{\Sigma}_{\scriptscriptstyle \!B}^{\scriptscriptstyle -\!1\!/\!2}\bm{\beta};\mathbf{A}).
		\end{aligned}
	\end{equation*} Because of the convexity of the loss function $\mathcal{L}$ on the $\bm{\beta}$ and $\hat{\Sigma}_{\scriptscriptstyle \!B}^{\scriptscriptstyle \!1\!/\!2}$ is a defined transformation holding the convexity, we know
	\begin{equation*}
			\tilde{\bm{\beta}} = \arg\min_{\bm{\beta}} \mathcal{L}(\hat{\Sigma}_{\scriptscriptstyle \!B}^{\scriptscriptstyle -\!1\!/\!2}\bm{\beta};\mathbf{A}) = \arg\min_{\hat{\Sigma}_{\scriptscriptstyle \!B}^{\scriptscriptstyle \!1\!/\!2}\bm{\beta}} \mathcal{L}(\bm{\beta};\mathbf{A}) = \hat{\Sigma}_{\scriptscriptstyle \!B}^{\scriptscriptstyle \!1\!/\!2}\arg\min_{\bm{\beta}} \mathcal{L}(\bm{\beta};\mathbf{A}) = \hat{\Sigma}_{\scriptscriptstyle \!B}^{\scriptscriptstyle \!1\!/\!2}\hat{\bm{\beta}}.
	\end{equation*} 
	\textbf{2. Affine invariance of Newton's method.} We simplify the $ \mathcal{L}(\bm{\beta};\mathbf{A})$ as $\mathcal{L}(\bm{\beta})$ and let $\hat{\bm{\beta}} = \hat{\Sigma}_{\scriptscriptstyle \!B}^{\scriptscriptstyle -\!1\!/\!2}\tilde{\bm{\beta}}$ and $\mathcal{H}(\tilde{\bm{\beta}}) = \mathcal{L}(\hat{\Sigma}_{\scriptscriptstyle \!B}^{\scriptscriptstyle -\!1\!/\!2} \tilde{\bm{\beta}})$.  The Newton step is given by
	\begin{equation*}
		\begin{aligned}
			\tilde{\bm{\beta}}^{\scriptscriptstyle (t + 1)} &= \tilde{\bm{\beta}}^{\scriptscriptstyle (t)} - (\nabla^2_{\tilde{\bm{\beta}}}  \mathcal{H}(\tilde{\bm{\beta}}^{\scriptscriptstyle (t)}))^{\scriptscriptstyle -\!1}(\nabla_{\tilde{\bm{\beta}}}  \mathcal{H}(\tilde{\bm{\beta}}^{\scriptscriptstyle (t)}))\\
			&= \tilde{\bm{\beta}}^{\scriptscriptstyle (t)} - (\hat{\Sigma}_{\scriptscriptstyle \!B}^{\scriptscriptstyle -\!1\!/\!2}\nabla^2_{\hat{\bm{\beta}}}  \mathcal{L}(\hat{\bm{\beta}}^{\scriptscriptstyle (t)})\hat{\Sigma}_{\scriptscriptstyle \!B}^{\scriptscriptstyle -\!1\!/\!2})^{\scriptscriptstyle -\!1}(\hat{\Sigma}_{\scriptscriptstyle \!B}^{\scriptscriptstyle -\!1\!/\!2} \nabla_{\hat{\bm{\beta}}}  \mathcal{L}(\hat{\bm{\beta}}^{\scriptscriptstyle (t)}))\\
			& = \tilde{\bm{\beta}}^{\scriptscriptstyle (t)} - \hat{\Sigma}_{\scriptscriptstyle \!B}^{\scriptscriptstyle \!1\!/\!2}(\nabla^2_{\hat{\bm{\beta}}}  \mathcal{L}(\hat{\bm{\beta}}^{\scriptscriptstyle (t)}))^{\scriptscriptstyle -\!1}(\nabla_{\hat{\bm{\beta}}}  \mathcal{L}(\hat{\bm{\beta}}^{\scriptscriptstyle (t)})).\\
		\end{aligned}
	\end{equation*}Multiplying both sides by $\hat{\Sigma}_{\scriptscriptstyle \!B}^{\scriptscriptstyle -\!1\!/\!2}$, we have
	\begin{equation*}
		\hat{\bm{\beta}}^{\scriptscriptstyle (t+1)} = \hat{\bm{\beta}}^{\scriptscriptstyle (t)} - (\nabla^2_{\hat{\bm{\beta}}}  \mathcal{L}(\hat{\bm{\beta}}^{\scriptscriptstyle (t)}))^{\scriptscriptstyle -\!1}(\nabla_{\hat{\bm{\beta}}}  \mathcal{L}(\hat{\bm{\beta}}^{\scriptscriptstyle (t)})).
	\end{equation*}$\hfill\square$
\end{prf}
\subsection{Proof of Theorem \ref{thm:DPLR_pri}}
\begin{prf}
	We analyze every iteration $t$. The sensitivity of $\nabla \tilde{\mathcal{L}}(\bm{\beta};\mathbf{\tilde{A}})$ is 
	\begin{equation*}
		\begin{aligned}
			\max\limits_{(\mathbf{\tilde{A}},\mathbf{y})\sim(\mathbf{\tilde{A}\!'},\mathbf{y}')}\|\nabla \tilde{\mathcal{L}}(\bm{\beta};\mathbf{\tilde{A}})-\nabla \tilde{\mathcal{L}}(\bm{\beta};\mathbf{\tilde{A}\!'})\| &= \max\limits_{(\mathbf{\tilde{A}},\mathbf{y})\sim(\mathbf{\tilde{A}\!'},\mathbf{y}')}\|-\frac{1}{n_{\scriptscriptstyle \!A}} \tilde{\mathbf{A}}^{\scriptscriptstyle \!T} (\mathbf{y} - \mathbf{\tilde{p}})+\frac{1}{n_{\scriptscriptstyle \!A}} \tilde{\mathbf{A}}^{\scriptscriptstyle \!'T} (\mathbf{y'} - \mathbf{\tilde{p}})\|\\
			&\leqslant \max\limits_{\tilde{A}} \frac{2}{n_{\scriptscriptstyle \!A}}\|\tilde{A}\|\\
			&\leqslant \frac{2\sqrt{d(1 + \log(\frac{2n_{\!\scriptscriptstyle A}}{\eta}))}}{n_{\scriptscriptstyle \!A}}.
		\end{aligned}
	\end{equation*}The sensitivity of $\nabla^2 \tilde{\mathcal{L}}(\bm{\beta};\mathbf{\tilde{A}})$ is 
	\begin{equation*}
		\begin{aligned}
			\max\limits_{\mathbf{\tilde{A}},\mathbf{\tilde{A}\!'}}\|\nabla^2 \tilde{\mathcal{L}}(\bm{\beta};\mathbf{\tilde{A}}) -\nabla^2 \tilde{\mathcal{L}}(\bm{\beta};\mathbf{\tilde{A\!'}})\|_F &= \max\limits_{\mathbf{\tilde{A}},\mathbf{\tilde{A}\!'}} \|\frac{1}{n_{\scriptscriptstyle A}}\mathbf{\tilde{A}^{\scriptscriptstyle \!\!T} \tilde{W}_{\scriptscriptstyle \!\!\!\beta} \tilde{A}} - \frac{1}{n_{\scriptscriptstyle A}}\mathbf{\tilde{A}^{\scriptscriptstyle \!'\!T} \tilde{W}_{\scriptscriptstyle \!\!\!\beta} \tilde{A}\!'}\|_F\\
			&\leqslant \max\limits_{\tilde{A}} \frac{2}{n_{\scriptscriptstyle A}}\tilde{p}(1-\tilde{p}) \|\tilde{A}\|_2^2\\
			&\leqslant \frac{d(1 + \log(\frac{2n_{\!\scriptscriptstyle A}}{\eta}))}{2n_{\scriptscriptstyle A}}.
		\end{aligned}
	\end{equation*}So every iteration satisfies $\sqrt{\frac{2}{T}}\mu$-GDP and the final output satisfies $\sqrt{2}\mu$-GDP. $\hfill\square$
\end{prf}
\subsection{Proof of Theorem \ref{thm:DP-PMTLR}}\label{appendix:DP_PMTLR_proof}
\begin{prf}
	\textbf{Lemma \ref{lem:loss_property}} illustrates that the starting value condition $\|\nabla\tilde{\mathcal{L}}(\bm{\beta}^{(0)};\tilde{\mathbf{A}})\|_2 \leq \frac{\gamma_L^2}{C_{\scriptscriptstyle \!\tilde{A}}}$ induces to the following inequality
	\begin{equation}\label{eq:thm_DP_PMTLR}
		\frac{C_{\scriptscriptstyle \!\tilde{A}}}{2\gamma^2_L}\|\nabla \tilde{\mathcal{L}}(\bm{\beta}^{(t+1)};\mathbf{A})\|_2 \leq \big(\frac{C_{\scriptscriptstyle \!\tilde{A}}}{2\gamma^2_L} \|\nabla \tilde{\mathcal{L}}(\bm{\beta}^{(0)};\tilde{\mathbf{A}})\|_2\big)^2 + C_{pri},
	\end{equation} where $C_{pri}$ is a constant about $O(\frac{C_{\scriptscriptstyle \!\tilde{A}}\sigma_1\sqrt{d\log(Td/\eta)}}{\gamma_L^3})$, $\sigma_1 = \frac{\sqrt{T}d(1 + \log(\frac{2n_{\!\scriptscriptstyle A}}{\eta}))}{2 \mu n_{\scriptscriptstyle A}}$, $\gamma_{\scriptscriptstyle \!L} = \tau_{\scriptscriptstyle 0}L\big(1 - O\big(\sqrt{\frac{d}{n_{\scriptscriptstyle \!A}}} + \sqrt{\frac{\log(1/\eta)}{n_{\scriptscriptstyle \!A}}}\big)\big)^2 + \frac{\lambda}{\lambda_{\scriptscriptstyle \max}(\hat{\Sigma}_{\scriptscriptstyle \!B})}$ and $L = \frac{n_{\scriptscriptstyle \!B}}{(\sqrt{n_{\scriptscriptstyle \!B}} + O(\sqrt{d} + \sqrt{\log(\frac{1}{\eta})}))^2}$, $C_{\scriptscriptstyle \!\!\tilde{A}} = \frac{(d(1 + \log(n_{\scriptscriptstyle \!A}/\eta)))^{\scriptscriptstyle 3\!/\!2}}{6\sqrt{3}}$, and $\eta > 0$. Moreover, we simplify $C_{pri} = O(\frac{\sqrt{Td^3\log(Td/\eta)}(1 + \log(\frac{2n_{\scriptscriptstyle \!A}}{\eta}))}{\mu \cdot n_A \cdot \gamma^2_L})$. We prove Eq.\eqref{eq:thm_DP_PMTLR} by mathematical induction: we can get the inequality when $t=1$ by Eq.\eqref{eq:loss_property}, and then we assume the inequality holds for some $t > 1$. Then, we get 

	\begin{equation*}
		\begin{aligned}
			\frac{C_{\scriptscriptstyle \!\tilde{A}}}{2\gamma^2_L}\|\nabla \tilde{\mathcal{L}}(\bm{\beta}^{(t+1)};\mathbf{A})\|_2 &\overset{(i)}{\leq} \big(\frac{C_{\scriptscriptstyle \!\tilde{A}}}{2\gamma^2_L} \|\nabla \tilde{\mathcal{L}}(\bm{\beta}^{(t)};\tilde{\mathbf{A}})\|_2\big)^2 + C_{pri}\\
			&\overset{(ii)}{\leq} \Big(\big(\frac{C_{\scriptscriptstyle \!\tilde{A}}}{2\gamma^2_L} \|\nabla \tilde{\mathcal{L}}(\bm{\beta}^{(0)};\tilde{\mathbf{A}})\|_2\big)^{2^t} + 3C_{pri}\Big)^2 + C_{pri}\\
			&\overset{(iii)}{\leq} \Big(\frac{C_{\scriptscriptstyle \!\tilde{A}}}{2\gamma^2_L} \|\nabla \tilde{\mathcal{L}}(\bm{\beta}^{(0)};\tilde{\mathbf{A}})\|_2\Big)^{2^{(t+1)}} +C_{pri}(\frac{3}{2} + 9C_{pri}) +C_{pri}\\
			&\overset{(iv)}{\leq} \Big(\frac{C_{\scriptscriptstyle \!\tilde{A}}}{2\gamma^2_L} \|\nabla \tilde{\mathcal{L}}(\bm{\beta}^{(0)};\tilde{\mathbf{A}})\|_2\Big)^{2^{(t+1)}} + 3 C_{pri},
		\end{aligned}
	\end{equation*}where the inequality $(i)$ is from \eqref{eq:loss_property}; the inequality $(ii)$ comes from the induction; the inequality $(iii)$ follows by $\frac{C_{\scriptscriptstyle \!\tilde{A}}}{2\gamma^2_L} \|\nabla \tilde{\mathcal{L}}(\bm{\beta}^{(t)};\tilde{\mathbf{A}})\|_2\leq\frac{1}{2}$; and we can guarantee the inequality $(iv)$ through controlling the private sample size $n_A$, making $C_{priv} \leq \frac{1}{18}$. That completes the conclusion Eq.\eqref{eq:thm_DP_PMTLR}. 

	The inequality \eqref{eq:thm_DP_PMTLR} illustrates that 
	\begin{equation*}
		\frac{C_{\scriptscriptstyle \!\tilde{A}}}{2\gamma^2_L}\|\nabla \tilde{\mathcal{L}}(\bm{\beta}^{(T)};\mathbf{A})\|_2 \leq \big(\frac{1}{2}\big)^{2^T} + 3C_{pri},
	\end{equation*} and when $T > \frac{1}{\log2}\log\big(\frac{\log2}{C_{pri}}\big)=O(\log\log(n_A))$, we get
	\begin{equation*}
		4C_{pri} \geq \frac{C_{\scriptscriptstyle \!\tilde{A}}}{2\gamma^2_L}\|\nabla \tilde{\mathcal{L}}(\bm{\beta}^{(T)};\mathbf{A})\|_2 \overset{(i)}{\geq} \frac{C_{\scriptscriptstyle \!\tilde{A}}}{2\gamma_L}\|\bm{\beta}^{(T)} - \tilde{\bm{\beta}}\|_2, 
	\end{equation*}where the inequality $(i)$ follows for the $\gamma_L$-strong convexity. Moreover, we have 
	\begin{equation*}
		\|\bm{\beta}^{(T)} - \tilde{\bm{\beta}}\|_2 \leq \frac{8\gamma_L}{C_{\scriptscriptstyle \!\tilde{A}}}C_{pri} \leq O(\frac{\gamma_L}{C_{\scriptscriptstyle \!\tilde{A}}} \frac{C_{\scriptscriptstyle \!\tilde{A}}\sigma_1\sqrt{d\log(Td/\eta)}}{\gamma_L^3})\leq O(\frac{\sigma_1\sqrt{d\log(Td/\eta)}}{\gamma_L^2}),
	\end{equation*}$\sigma_1 = \frac{\sqrt{T}d(1 + \log(\frac{2n_{\!\scriptscriptstyle A}}{\eta}))}{2 \mu n_{\scriptscriptstyle A}}$, $\gamma_{\scriptscriptstyle \!L} = \tau_{\scriptscriptstyle 0}L\big(1 - O\big(\sqrt{\frac{d}{n_{\scriptscriptstyle \!A}}} + \sqrt{\frac{\log(1/\eta)}{n_{\scriptscriptstyle \!A}}}\big)\big)^2 + \frac{\lambda}{\lambda_{\scriptscriptstyle \max}(\hat{\Sigma}_{\scriptscriptstyle \!B})}$ and $L = \frac{n_{\scriptscriptstyle \!B}}{(\sqrt{n_{\scriptscriptstyle \!B}} + O(\sqrt{d} + \sqrt{\log(\frac{1}{\eta})}))^2}$.
	$\hfill\square$.
\end{prf}
\begin{rem}
	 We remark the proof of \textbf{Theorem \ref{thm:DP-LR}} here. Considering $\tilde{\mathcal{L}}(\bm{\beta};\mathbf{\tilde{A}}) = \mathcal{L}(\bm{\beta};\mathbf{A})$ and its properties' parameters in \textbf{Lemma \ref{lem:strongconvex}} and \textbf{Lemma \ref{lem:loss_property}}, then the rest of the proof follows the proof of \textbf{Theorem \ref{thm:DP-PMTLR}} in \textbf{Appendix \ref{appendix:DP_PMTLR_proof}}.
\end{rem}

\subsection{Proof of Theorem \ref{thm:DP-PMTGLM}}\label{appendix:DP_PMTGLM_proof}
\begin{prf}
	\textbf{1. Privacy.} The sensitivity of gradient $\nabla_{\!\! \beta}\tilde{\mathcal{L}}(\bm{\beta};\tilde{\mathbf{A}})$ is 
	\begin{equation*}
		\begin{aligned}
				\max_{(\tilde{\mathbf{A}},\mathbf{y})\sim(\tilde{\mathbf{A'}},\mathbf{y'})}\|\nabla_{\!\! \beta}\tilde{\mathcal{L}}(\bm{\beta};\tilde{\mathbf{A}}) - \nabla_{\!\! \beta}\tilde{\mathcal{L}}(\bm{\beta};\tilde{\mathbf{A'}}) \|
				&= \max_{(\tilde{A}_i,y_i)\sim(\tilde{A}'_i,y'_i)}\frac{1}{n_{\scriptscriptstyle \!A}}\|\nabla_{\!\! \beta} l_i(\tilde{A}_i^\top \bm{\beta}) -  \nabla_{\!\! \beta} l_i(\tilde{A'}_i^\top \bm{\beta})\|\\
				&= \max_{(\tilde{A}_i,y_i)\sim(\tilde{A}'_i,y'_i)}\frac{1}{n_{\scriptscriptstyle \!A}}\|(y_i - b'(\tilde{A}_i^\top \bm{\beta})) \tilde{A}_i - (y'_i - b'(\tilde{A'}_i^\top \bm{\beta})) \tilde{A}'_i\|\\
				&\leq \frac{2}{n_{\scriptscriptstyle \!A}} \max_{\tilde{A}_i,y_i} \{|y_i|\| \tilde{A}_i\| +  |b'(\tilde{A'}_i^\top \bm{\beta})|\|\tilde{A}'_i\|\}\\
				&\leq \frac{2\sqrt{d(1 + \log(\frac{2n_{\!\scriptscriptstyle A}}{\eta}))}}{n_{\scriptscriptstyle \!A}} (R_y + M_{b'}),
		\end{aligned}
	\end{equation*}where $M_{b'} = max_{z} b'(z), z\in [-R_\beta\sqrt{d(1 + \log(\frac{2n_{\!\scriptscriptstyle A}}{\eta}))},R_\beta\sqrt{d(1 + \log(\frac{2n_{\!\scriptscriptstyle A}}{\eta}))}]$ and $\|\bm{\beta}\|_2 \leq R_\beta$.
	The sensitivity of gradient $\nabla^2_{\!\! \beta}\tilde{\mathcal{L}}(\bm{\beta};\tilde{\mathbf{A}})$ is  
		\begin{equation*}
		\begin{aligned}
				\max_{(\tilde{\mathbf{A}},\mathbf{y})\sim(\tilde{\mathbf{A'}},\mathbf{y'})}\|\nabla^2_{\!\! \beta}\tilde{\mathcal{L}}(\bm{\beta};\tilde{\mathbf{A}}) - \nabla^2_{\!\! \beta}\tilde{\mathcal{L}}(\bm{\beta};\tilde{\mathbf{A'}}) \|_F
				&= \max_{(\tilde{A}_i,y_i)\sim(\tilde{A}'_i,y'_i)}\frac{1}{n_{\scriptscriptstyle \!A}}\|\nabla^2_{\!\! \beta} l_i(\tilde{A}_i^\top \bm{\beta}) -  \nabla^2_{\!\! \beta} l_i(\tilde{A'}_i^\top \bm{\beta})\|_F\\
				&\leq \max_{\tilde{A}_i,\tilde{A}'_i} \frac{2}{n_{\scriptscriptstyle \!A}}\|b''(\tilde{A}_i^\top \bm{\beta}) \tilde{A}_i \tilde{A}_i^\top\|_F\\
				&\leq \max_{\tilde{A}_i} \frac{2}{n_{\scriptscriptstyle \!A}}|b''(\tilde{A}_i^\top \bm{\beta})| \|\tilde{A}_i\|_2^2\\
				&\leq \frac{2}{n_{\scriptscriptstyle \!A}} d(1 + \log(\frac{2n_{\!\scriptscriptstyle A}}{\eta}))M_{b''}, 
		\end{aligned}
	\end{equation*}where $M_{b''} = max_{z} b''(z), z\in [-R_\beta\sqrt{d(1 + \log(\frac{2n_{\!\scriptscriptstyle A}}{\eta}))},R_\beta\sqrt{d(1 + \log(\frac{2n_{\!\scriptscriptstyle A}}{\eta}))}]$ and $\|\bm{\beta}\|_2 \leq R_\beta$.

	Hence, let $\sigma_1 =\frac{2\sqrt{T}d(1 + \log(\frac{2n_{\!\scriptscriptstyle A}}{\eta}))M_{b''}}{\mu n_{\scriptscriptstyle \!A}}$ and $\sigma_2 = \frac{2\sqrt{Td(1 + \log(\frac{2n_{\!\scriptscriptstyle A}}{\eta}))}(R_y + M_{b'})}{\mu n_{\scriptscriptstyle \!A}}$. Considering the $T$ times compositions of GDP (\textbf{Theorem \ref{thm:composition}}), the final output satisfy $\sqrt{2}\mu-$GDP.

	\textbf{2. Convergence.} Considering the conditions and \textbf{Lemma \ref{lem:loss_property}}, the proof is similar to the proof of \textbf{Theorem \ref{thm:DP-PMTLR}}.
	$\hfill\square$

\end{prf}

\section{Proofs of lemmas}\label{append:proof_lemmas}
\subsection{Proof of Lemma \ref{lem:nonisotropic_norm}}\label{append:proof_nonisotropic_norm}
\begin{prf}
For simplicity, we assume that $K \geq 1$. Since the random $x_i$ is sub-gaussian, $x_i^2$ is sub-exponential, and more precisely 
\begin{equation*}
  \begin{aligned}
    \big\|\|\mathbf{x}\|_2^2\big\|_{\psi_1} &= \| \sum_{i=1}^{d} x_i^2\|_{\psi_1}\\
    &\leq \sum_{i=1}^{d} \|x_i^2\|_{\psi_1} \\
    &\leq d\max_i \|x^2_i\|_{\psi_1} \\
    &= d \max_i \|x_i\|_{\psi_2}^2.
  \end{aligned}
\end{equation*}
Then, we compute the expectation of $\|\mathbf{x}\|_2^2$
\begin{equation*}
  \begin{aligned}
    \mathbb{E} \|\mathbf{x}\|_2^2 &= \mathbb{E} \mathbf{x}^T \mathbf{x}\\
    & =\mathbb{E} \Tr(\mathbf{x}^T \mathbf{x}) \\
    &= \mathbb{E} \Tr(\mathbf{x} \mathbf{x}^T) \\
    &= \Tr(\mathbb{E} \mathbf{x} \mathbf{x}^T) \\
    &= \Tr(\Sigma) \\
    &= \sum_{i=1}^{d} \lambda_i.
  \end{aligned}
\end{equation*}
Reconsider the tail bound of the centered sub-exponential random, we have 
\begin{equation*}
	\mathbb{P}\Big[\Big|\|\mathbf{x}\|_2^2 - \mathbb{E}\|\mathbf{x}\|_2^2\Big| > t \Big] \leq 2\exp\big(-\frac{ct}{dK^2}\big),
\end{equation*}where $K = \max_i \|x_i\|_{\psi_2}$ and $c$ is an absolute constant. 	$\hfill\square$
\end{prf}

\subsection{Proof of Corollary \ref{cor:utility_trunc}}
\begin{prf}
	Combining \textbf{Theorem \ref{thm:secondmoment_bound}} and the $n_{\xi}$ union bound of \textbf{Lemma \ref{lem:nonisotropic_norm}}, we get the proof. The second inequality is because $U$ and $L$ tend to $1$ so that $\Tr(\tilde{\Sigma}) \to d$. $\hfill\square$
\end{prf}

\subsection{Proof of Lemma \ref{lem:strongconvex}}\label{append:proof_strongconvex}
\begin{prf}
	\textbf{1. Strong convexity of $\mathcal{L}(\bm{\beta};\mathbf{A})$.} From the assumption \textbf{\eqref{eq:Hessian_assumption}} and \textbf{Theorem \ref{lem:singular_bound}}, with at least probability $1 - \eta$, the $\mathcal{L}(\bm{\beta};\mathbf{A})$ is $\gamma_{\scriptscriptstyle \Sigma}$-strong convexity. Namely, the Hessian matrix 
	\begin{equation}\label{eq:Hessian_L}
		\begin{aligned}
			\nabla^2 \mathcal{L}(\bm{\beta};\mathbf{A}) = \mathbf{H_{\scriptscriptstyle \!\beta}} + \lambda \mathbf{I} &\succcurlyeq (\tau_{\scriptscriptstyle 0} \lambda_{\scriptscriptstyle \min} (\hat{\Sigma}_{\scriptscriptstyle \!A}) + \lambda)\mathbf{I}\\
			&\succcurlyeq \Big(\tau_{\scriptscriptstyle 0} \lambda_{\scriptscriptstyle \min}(\Sigma) \big(1 - O\big(\sqrt{\frac{d}{n_{\scriptscriptstyle \!A}}} + \sqrt{\frac{\log(1/\eta)}{n_{\scriptscriptstyle \!A}}}\big)\big)^2 + \lambda\Big) \mathbf{I}.\\
		\end{aligned}
	\end{equation}

	\textbf{2. Strong convexity of $\tilde{\mathcal{L}}(\bm{\beta};\mathbf{\tilde{A}})$.} From the proof of \textbf{Theorem \ref{thm:Equivalent estimation}}, the the loss function $\tilde{\mathcal{L}}(\bm{\beta};\mathbf{\tilde{A}})$ is equivalent to $\mathcal{L}(\hat{\Sigma}_{\scriptscriptstyle \!B}^{\scriptscriptstyle -\!1\!/\!2}\bm{\beta};\mathbf{A})$. So the Hessian matrix of $\tilde{\mathcal{L}}(\bm{\beta};\mathbf{\tilde{A}})$ is
	\begin{equation*}
		\nabla^2 \tilde{\mathcal{L}}(\bm{\beta};\mathbf{\tilde{A}}) = \nabla^2 \mathcal{L}(\hat{\Sigma}_{\scriptscriptstyle \!B}^{\scriptscriptstyle -\!1\!/\!2}\bm{\beta};\mathbf{A}) =   \hat{\Sigma}_{\scriptscriptstyle \!B}^{\scriptscriptstyle -\!1\!/\!2}(\mathbf{H}_{\scriptscriptstyle \hat{\Sigma}_{\!\scaleto{B}{3pt}}^{\scaleto{-\frac{1}{2}}{5pt}} \!\beta} + \lambda \mathbf{I})\hat{\Sigma}_{\scriptscriptstyle \!B}^{\scriptscriptstyle -\!1\!/\!2}.
	\end{equation*} Due to \textbf{Theorem \ref{thm:secondmoment_bound}}, with at least probability $1 - 2\eta$, we have
	\begin{equation*}
		\begin{aligned}
			\hat{\Sigma}_{\scriptscriptstyle \!B}^{\scriptscriptstyle -\!1\!/\!2}(\mathbf{H}_{\scriptscriptstyle \hat{\Sigma}_{\!\scaleto{B}{3pt}}^{\scaleto{-\frac{1}{2}}{5pt}} \!\beta} + \lambda \mathbf{I})\hat{\Sigma}_{\scriptscriptstyle \!B}^{\scriptscriptstyle -\!1\!/\!2} &\succcurlyeq (\tau_{\scriptscriptstyle 0} \lambda_{\scriptscriptstyle \min} (\hat{\Sigma}_{\scriptscriptstyle \!B}^{\scriptscriptstyle -\!1\!/\!2}\hat{\Sigma}_{\scriptscriptstyle \!A}\hat{\Sigma}_{\scriptscriptstyle \!B}^{\scriptscriptstyle -\!1\!/\!2}) + \lambda \cdot  \lambda_{\scriptscriptstyle \!\min}(\hat{\Sigma}_{\scriptscriptstyle \!B}^{\scriptscriptstyle -\!1}))\mathbf{I}\\
			&\succcurlyeq (\tau_{\scriptscriptstyle 0}L\big(1 - O\big(\sqrt{\frac{d}{n_{\scriptscriptstyle \!A}}} + \sqrt{\frac{\log(1/\eta)}{n_{\scriptscriptstyle \!A}}}\big)\big)^2 + \lambda \cdot \lambda_{\scriptscriptstyle \!\min}(\hat{\Sigma}_{\scriptscriptstyle \!B}^{\scriptscriptstyle -\!1}))\mathbf{I},\\
		\end{aligned}
	\end{equation*}where $L = \frac{n_{\scriptscriptstyle \!B}}{(\sqrt{n_{\scriptscriptstyle \!B}} + O(\sqrt{d} + \sqrt{\log(\frac{1}{\eta})}))^2}$. Here, $\tau_{\scriptscriptstyle 0}$ is still the lower bound, because the transformation doesn't change the classified accuracy of the final estimation.$\hfill\square$
\end{prf}

\subsection{Proof of Lemma \ref{lem:Lipschitz}}\label{append:proof_Lipschitz}
\begin{prf}
	We begin to get the Lipschitz constant of $l_i(A_i^T \bm{\beta})$. We know the Hessian matrix of $l_i(A_i^T \bm{\beta})$ is
	\begin{equation*}
		\nabla^2_{\bm{\beta}} l_i(A_i^T \bm{\beta}) = p_i(1-p_i)A_i A_i^{\scriptscriptstyle \!T},
	\end{equation*}and its third-order derivative is
	\begin{equation*}
		\nabla^3_{\bm{\beta}} l_i(A_i^T \bm{\beta}) = p_i(1-p_i)(1 -2p_i)A_i\otimes A_i\otimes A_i.
	\end{equation*} So the Lipschitz constant of Hessian matrix is
	\begin{equation*}
		 \sup_{\bm{\beta}} \nabla^3_{\bm{\beta}} l_i(A_i^T \bm{\beta}) \leqslant \sup_{0<p_i\leqslant\frac{1}{2}} \|p_i(1-p_i)(1 -2p_i)\|\sup_{A_i}\|A_i\otimes A_i\otimes A_i\| \leqslant \frac{\sup_{i}\|A_i\|^3}{6\sqrt{3}}.
	\end{equation*} That illustrates $C_{\scriptscriptstyle \!\!A}=\frac{\sup_{i}\|A_i\|^3}{6\sqrt{3}}$ and
	\begin{equation*}
		\|\nabla^2_{\bm{\beta}} l_i(A_i^T \bm{\beta}) - \nabla^2_{\bm{\beta}} l_i(A_i^T \bm{\beta'})\| \leqslant C_{\scriptscriptstyle \!\!A}\|\bm{\beta} - \bm{\beta'}\|.
	\end{equation*}
	Then, we have the Lipschitz continuity of Hessian matrix $\mathbf{H_{\scriptscriptstyle \!\beta}}$ is
	\begin{equation}
		\begin{aligned}
			\|\mathbf{H_{\scriptscriptstyle \!\beta}} - \mathbf{H_{\scriptscriptstyle \!\beta'}}\|_2 &= \|\frac{1}{n_{\scriptscriptstyle \!A}}\sum_{i=1}^{n_{\scriptscriptstyle \!A}} \nabla^2_{\bm{\beta}} l_i(A_i^T \bm{\beta}) - \frac{1}{n_{\scriptscriptstyle \!A}}\sum_{i=1}^{n_{\scriptscriptstyle \!A}} \nabla^2_{\bm{\beta}} l_i(A_i^T \bm{\beta'})\|_2 \\
			&\leqslant\frac{1}{n_{\scriptscriptstyle \!A}}\sum_{i=1}^{n_{\scriptscriptstyle \!A}}\| \nabla^2_{\bm{\beta}} l_i(A_i^T \bm{\beta}) -  \nabla^2_{\bm{\beta}} l_i(A_i^T \bm{\beta'})\|_2 \\
			&\leqslant C_{\scriptscriptstyle \!\!A}\|\bm{\beta} - \bm{\beta'}\|_2.\\
		\end{aligned}
	\end{equation} The Lipschitz continuity of the Hessian matrix $\mathbf{H_{\scriptscriptstyle \!\beta}}$ is proved. $\hfill\square$
\end{prf}

\subsection{Algorithms}

\begin{algorithm}[H]
	\caption{Differentially Private Logistic Regression (DP-LR)}\label{alg:DP-LR}
	\begin{algorithmic}[1]
		\STATE {\bfseries Input:} Private dataset $\{(\mathbf{A}_i,y_i) \in \mathbb{R}^{d}\times\{0,1\} \}^{n_{\!\scriptscriptstyle A}}_{i=1}$. Parameters $\mu$, $\lambda$, $d$, $n_{\!\scriptscriptstyle A}$, and $\eta$.
		
		\STATE {\bfseries Private parameter:} 
			\begin{center}
				$\sigma_1 = \frac{\sqrt{T}(\Tr(\Sigma) + d\log(\frac{n_{\scriptscriptstyle \!A}}{\eta}))}{2 \mu n_{\scriptscriptstyle A}}$, 
				$\sigma_2 =\frac{2\sqrt{T(\Tr(\Sigma) + d\log(\frac{n_{\scriptscriptstyle \!A}}{\eta}))}}{\mu n_{\scriptscriptstyle \!A}}.$
			\end{center}
		\STATE{\bfseries Intialization: } $\bm{\beta}^{\scriptscriptstyle (0)} = \mathbf{0} $
			\FOR{$t = 0,...,T $ }{
				\STATE {\bfseries Gaussian mechanism and Newton's update:} 
				\begin{center}
					$\bm{{\beta}}^{\scriptscriptstyle (t+1)} = \bm{{\beta}}^{\scriptscriptstyle (t)} - \Big(\nabla^2 \mathcal{L}(\bm{\beta}^{\scriptscriptstyle (t)};\mathbf{A}) + \mathbf{G}\Big)^{\!\!\scriptscriptstyle -\!1}\Big(\nabla \mathcal{L}(\bm{\beta}^{\scriptscriptstyle (t)};\mathbf{A}) + \mathbf{g}\Big),$
				\end{center}where $\mathbf{G} \sim SG_d(\sigma_1^2)$ and $\mathbf{g} \sim \mathcal{N}(0,\sigma_2^2\mathbf{I})$.}
			\ENDFOR

		\STATE {\bfseries Output:}  DP estimator $\bm{\beta}^{(T)}$.
	\end{algorithmic}
\end{algorithm}
\end{appendix}

%% file: bib/ref.bib
@inproceedings{abadi2016deep,
  title={Deep learning with differential privacy},
  author={Abadi, Martin and Chu, Andy and Goodfellow, Ian and McMahan, H Brendan and Mironov, Ilya and Talwar, Kunal and Zhang, Li},
  booktitle={Proceedings of the 2016 ACM SIGSAC conference on computer and communications security},
  pages={308--318},
  year={2016}
}

@inproceedings{amid2022public,
  title={Public data-assisted mirror descent for private model training},
  author={Amid, Ehsan and Ganesh, Arun and Mathews, Rajiv and Ramaswamy, Swaroop and Song, Shuang and Steinke, Thomas and Suriyakumar, Vinith M and Thakkar, Om and Thakurta, Abhradeep},
  booktitle={International Conference on Machine Learning},
  pages={517--535},
  year={2022},
  organization={PMLR}
}

@article{bassily2020learning,
  title={Learning from mixtures of private and public populations},
  author={Bassily, Raef and Moran, Shay and Nandi, Anupama},
  journal={Advances in neural information processing systems},
  volume={33},
  pages={2947--2957},
  year={2020}
}

@article{bie2022private,
  title={Private estimation with public data},
  author={Bie, Alex and Kamath, Gautam and Singhal, Vikrant},
  journal={Advances in neural information processing systems},
  volume={35},
  pages={18653--18666},
  year={2022}
}

@article{dwork2014algorithmic,
  title={The algorithmic foundations of differential privacy},
  author={Dwork, Cynthia and Roth, Aaron and others},
  journal={Foundations and Trends{\textregistered} in Theoretical Computer Science},
  volume={9},
  number={3--4},
  pages={211--407},
  year={2014},
  publisher={Now Publishers, Inc.}
}

@article{ferrando2021combining,
  title={Combining public and private data},
  author={Ferrando, Cecilia and Gillenwater, Jennifer and Kulesza, Alex},
  journal={arXiv preprint arXiv:2111.00115},
  year={2021}
}

@article{ji2013differential,
  title={Differential privacy based on importance weighting},
  author={Ji, Zhanglong and Elkan, Charles},
  journal={Machine learning},
  volume={93},
  pages={163--183},
  year={2013},
  publisher={Springer}
}

@inproceedings{kairouz2021nearly,
  title={(Nearly) Dimension Independent Private ERM with AdaGrad Rates via Publicly Estimated Subspaces},
  author={Kairouz, Peter and Diaz, Monica Ribero and Rush, Keith and Thakurta, Abhradeep},
  booktitle={Conference on Learning Theory},
  pages={2717--2746},
  year={2021},
  organization={PMLR}
  }

@article{laurent2000adaptive,
  title={Adaptive estimation of a quadratic functional by model selection},
  author={Laurent, Beatrice and Massart, Pascal},
  journal={Annals of statistics},
  pages={1302--1338},
  year={2000},
  publisher={JSTOR}
}

@inproceedings{liu2021leveraging,
  title={Leveraging public data for practical private query release},
  author={Liu, Terrance and Vietri, Giuseppe and Steinke, Thomas and Ullman, Jonathan and Wu, Steven},
  booktitle={International Conference on Machine Learning},
  pages={6968--6977},
  year={2021},
  organization={PMLR}
}

@inproceedings{nandi2020privately,
  title={Privately answering classification queries in the agnostic pac model},
  author={Nandi, Anupama and Bassily, Raef},
  booktitle={Algorithmic Learning Theory},
  pages={687--703},
  year={2020},
  organization={PMLR}
}

@inproceedings{nasr2023effectively,
  title={Effectively using public data in privacy preserving machine learning},
  author={Nasr, Milad and Mahloujifar, Saeed and Tang, Xinyu and Mittal, Prateek and Houmansadr, Amir},
  booktitle={International Conference on Machine Learning},
  pages={25718--25732},
  year={2023},
  organization={PMLR}
}

@article{vershynin2010introduction,
  title={Introduction to the non-asymptotic analysis of random matrices},
  author={Vershynin, Roman},
  journal={arXiv preprint arXiv:1011.3027},
  year={2010}
}

@inproceedings{sheffet2017differentially,
  title={Differentially private ordinary least squares},
  author={Sheffet, Or},
  booktitle={International Conference on Machine Learning},
  pages={3105--3114},
  year={2017},
  organization={PMLR}
}

@article{bernstein2019differentially,
  title={Differentially private bayesian linear regression},
  author={Bernstein, Garrett and Sheldon, Daniel R},
  journal={Advances in Neural Information Processing Systems},
  volume={32},
  year={2019}
}

@misc{wine_quality_186,
  author       = {Cortez, Paulo and Cerdeira, A. and  Almeida, F.and  Matos, T. and Reis, J.},
  title        = {{Wine Quality}},
  year         = {2009},
  howpublished = {UCI Machine Learning Repository},
  note         = {{DOI}: https://doi.org/10.24432/C56S3T}
}

@article{wang2018revisiting,
  title={Revisiting differentially private linear regression: optimal and adaptive prediction \& estimation in unbounded domain},
  author={Wang, Yu-Xiang},
  journal={arXiv preprint arXiv:1803.02596},
  year={2018}
}

@article{dong2022gaussian,
  title={Gaussian differential privacy},
  author={Dong, Jinshuo and Roth, Aaron and Su, Weijie J},
  journal={Journal of the Royal Statistical Society: Series B (Statistical Methodology)},
  volume={84},
  number={1},
  pages={3--37},
  year={2022},
  publisher={Wiley Online Library}
}

@article{awan2023canonical,
  title={Canonical noise distributions and private hypothesis tests},
  author={Awan, Jordan and Vadhan, Salil},
  journal={The Annals of Statistics},
  volume={51},
  number={2},
  pages={547--572},
  year={2023},
  publisher={Institute of Mathematical Statistics}
}

@article{avella2023differentially,
  title={Differentially private inference via noisy optimization},
  author={Avella-Medina, Marco and Bradshaw, Casey and Loh, Po-Ling},
  journal={The Annals of Statistics},
  volume={51},
  number={5},
  pages={2067--2092},
  year={2023},
  publisher={Institute of Mathematical Statistics}
}

@article{zhao2025minimax,
  title={Minimax rates of convergence for sliced inverse regression with differential privacy},
  author={Zhao, Wenbiao and Zhu, Xuehu and Zhu, Lixing},
  journal={Computational Statistics \& Data Analysis},
  volume={201},
  pages={108041},
  year={2025},
  publisher={Elsevier}
}

@article{cao2023privacy,
  title={Privacy-Preserving Distributed Learning via Newton Algorithm},
  author={Cao, Zilong and Guo, Xiao and Zhang, Hai},
  journal={Mathematics},
  volume={11},
  number={18},
  pages={3807},
  year={2023},
  publisher={MDPI}
}

@inproceedings{nasr2023tight,
  title={Tight auditing of differentially private machine learning},
  author={Nasr, Milad and Hayes, Jamie and Steinke, Thomas and Balle, Borja and Tram{\`e}r, Florian and Jagielski, Matthew and Carlini, Nicholas and Terzis, Andreas},
  booktitle={32nd USENIX Security Symposium (USENIX Security 23)},
  pages={1631--1648},
  year={2023}
}

@article{bu2020deep,
  title={Deep learning with gaussian differential privacy},
  author={Bu, Zhiqi and Dong, Jinshuo and Long, Qi and Su, Weijie J},
  journal={Harvard data science review},
  volume={2020},
  number={23},
  pages={10--1162},
  year={2020}
}

@article{ivkin2019communication,
  title={Communication-efficient distributed SGD with sketching},
  author={Ivkin, Nikita and Rothchild, Daniel and Ullah, Enayat and Stoica, Ion and Arora, Raman and others},
  journal={Advances in Neural Information Processing Systems},
  volume={32},
  year={2019}
}

@inproceedings{koloskova2023revisiting,
  title={Revisiting gradient clipping: Stochastic bias and tight convergence guarantees},
  author={Koloskova, Anastasia and Hendrikx, Hadrien and Stich, Sebastian U},
  booktitle={International Conference on Machine Learning},
  pages={17343--17363},
  year={2023},
  organization={PMLR}
}

@article{ganesh2023faster,
  title={Faster differentially private convex optimization via second-order methods},
  author={Ganesh, Arun and Haghifam, Mahdi and Steinke, Thomas and Guha Thakurta, Abhradeep},
  journal={Advances in Neural Information Processing Systems},
  volume={36},
  pages={79426--79438},
  year={2023}
}

@article{bi2023distribution,
  title={Distribution-invariant differential privacy},
  author={Bi, Xuan and Shen, Xiaotong},
  journal={Journal of econometrics},
  volume={235},
  number={2},
  pages={444--453},
  year={2023},
  publisher={Elsevier}
}

@article{wang2019distributed,
  title={Distributed logistic regression with differential privacy},
  author={Wang, Puyu and Zhang, Hai},
  journal={Sci. Sin. Inform. doi},
  volume={10},
  year={2019}
}

@book{wainwright2019high,
  title={High-dimensional statistics: A non-asymptotic viewpoint},
  author={Wainwright, Martin J},
  volume={48},
  year={2019},
  publisher={Cambridge university press}
}

@misc{banknote_authentication_267,
  title        = {Banknote Authentication},
  author       = {{Lohweg, Volker}},
  year         = {2012},
  howpublished = {UCI Machine Learning Repository},
  note         = {{DOI}: https://doi.org/10.24432/C55P57}
}

@misc{bank_marketing_222,
  author       = {Moro, S. and Rita, P. and Cortez, P.},
  title        = {{Bank Marketing}},
  year         = {2014},
  howpublished = {UCI Machine Learning Repository},
  note         = {{DOI}: https://doi.org/10.24432/C5K306}
}

@misc{combined_cycle_power_plant_294,
  author       = {Tfekci, Pnar and Kaya, Heysem},
  title        = {{Combined Cycle Power Plant}},
  year         = {2014},
  howpublished = {UCI Machine Learning Repository},
  note         = {{DOI}: https://doi.org/10.24432/C5002N}
}

@article{Su2020DeepLG,
author = {Bu, Zhiqi and Dong, Jinshuo and Long, Qi and Su, Weijie},
year = {2020},
month = {07},
pages = {},
title = {Deep Learning with Gaussian Differential Privacy},
volume = {2020},
journal = {Harvard Data Science Review},
doi = {10.1162/99608f92.cfc5dd25}
}

@article{Dwork2006CalibratingNT,
  title={Calibrating Noise to Sensitivity in Private Data Analysis},
  author={Cynthia Dwork and Frank McSherry and Kobbi Nissim and Adam D. Smith},
  journal={J. Priv. Confidentiality},
  year={2006},
  volume={7},
  pages={17-51},
  url={https://api.semanticscholar.org/CorpusID:2468323}
}
